\documentclass[english]{article}
\usepackage[T1]{fontenc}
\usepackage[latin9]{inputenc}
\usepackage{geometry}
\usepackage{color}
\usepackage{babel}
\usepackage{float}
\usepackage{textcomp}
\usepackage{amsmath}
\usepackage{amsthm}
\usepackage{amssymb}
\usepackage{graphicx}
\usepackage{setspace}
\usepackage[unicode=true,pdfusetitle,
 bookmarks=true,bookmarksnumbered=false,bookmarksopen=false,
 breaklinks=false,pdfborder={0 0 1},backref=false,colorlinks=false]
 {hyperref}

\makeatletter
\theoremstyle{plain}
\newtheorem{thm}{\protect\theoremname}
\theoremstyle{plain}
\newtheorem{prop}{\protect\propositionname}
\theoremstyle{plain}
\newtheorem{lem}{\protect\lemmaname}

\makeatother

\providecommand{\lemmaname}{Lemma}
\providecommand{\propositionname}{Proposition}
\providecommand{\theoremname}{Theorem}

\begin{document}
\begin{doublespace}
\begin{center}
\textbf{\textcolor{black}{\large{}David vs Goliath (You against the
Markets), }}{\large\par}
\par\end{center}

\begin{center}
\textbf{\textcolor{black}{\large{}A Dynamic Programming Approach to
Separate the Impact and Timing of Trading Costs}}{\large\par}
\par\end{center}

\begin{center}
\textbf{Ravi Kashyap (ravi.kashyap@stern.nyu.edu)}
\par\end{center}

\begin{center}
\textbf{City University of Hong Kong }
\par\end{center}

\begin{center}
\textbf{April 20, 2020}
\par\end{center}

\begin{center}
\par\end{center}

\begin{center}
\par\end{center}

\begin{center}
AMS Subject Codes: 68T37 Reasoning under uncertainty; 49L20 Dynamic
programming method; 91A10 Noncooperative games 
\par\end{center}
\end{doublespace}

\begin{center}
\textbf{\textcolor{blue}{\href{https://doi.org/10.1016/j.physa.2019.122848}{Edited Version: Kashyap, R. (2020).  David vs Goliath (You against the Markets), A Dynamic Programming Approach to Separate the Impact and Timing of Trading Costs. Physica A: Statistical Mechanics and its Applications, 545, 122848. }}}\tableofcontents{}
\par\end{center}
\begin{doublespace}

\section{Abstract}
\end{doublespace}

\begin{doublespace}
We develop a fundamentally different stochastic dynamic programming
model of trading costs. Built on a strong theoretical foundation,
our model provides insights to market participants by splitting the
overall move of the security price during the duration of an order
into the Market Impact (price move caused by their actions) and Market
Timing (price move caused by everyone else) components. We derive
formulations of this model under different laws of motion of the security
prices, starting with a simple benchmark scenario and extending this
to include multiple sources of uncertainty, liquidity constraints
due to volume curve shifts and relating trading costs to the spread.

We develop a numerical framework that can be used to obtain optimal
executions under any law of motion of prices and demonstrate the tremendous
practical applicability of our theoretical methodology including the
powerful numerical techniques to implement them.\textit{ }Our decomposition
of trading costs into Market Impact and Market Timing allows us to
deduce the zero sum game nature of trading costs.\textit{ }It holds
numerous lessons for dealing with complex systems, wherein reducing
the complexity by splitting the many sources of uncertainty can lead
to better insights in the decision process.
\end{doublespace}

\pagebreak{}
\begin{doublespace}

\section{Introduction}
\end{doublespace}

\begin{doublespace}
The recent blockbuster book, David and Goliath: Underdogs, Misfits,
and the Art of Battling Giants (Gladwell 2013), talks about the advantages
of disadvantages, which in the legendary battle refers to (among other
things) the nimbleness that David possesses due to his smaller size
and lack of armor, that comes in handy while defeating the massive
and seemingly unbeatable Goliath. Despite the inspiring tone of the
story the efforts of the most valiant financial market participant
can seem puny and turn out to be inadequate, as it gets undone when
dealing with the gargantuan and mysterious temperament of uncertainty
in the markets. 
\end{doublespace}

Another main feature of the David versus Goliath story is the tool
(sling\footnote{A sling is a projectile weapon typically used to throw a blunt projectile
such as a stone, clay, or lead \textquotedbl sling-bullet\textquotedbl .
It is also known as the shepherd's sling. \href{https://en.wikipedia.org/wiki/Sling_(weapon)}{Sling (Weapon), Wikipedia Link}}) that David uses to defeat Goliath. In this article, we hope to provide
tools for market participants to contend with the Goliath-like uncertainty
in financial markets. A trader\textquoteright s conundrum is whether
(and how much) to trade during a given interval or wait for the next
interval when the price momentum is more favorable to his direction
of trading. But given the nature of uncertainty in the social sciences,
any weapon might prove to be insufficient compared to the sling that
delivered the fatal blow to Goliath, until perhaps, one can discern
the ability to read the minds of all the market participants. That
being said, the techniques in this paper will go a long way towards
helping participants and making their life easier when confronting
the markets. In addition, the mechanisms we provide can be useful
for combating uncertainty and aiding better decision making in many
areas of the social sciences. 

\begin{doublespace}
\textbf{\textit{We develop a new stochastic dynamic programming model
of trading costs (section \ref{sec:Alternative-and-Practical}) based
on the Bellman principle of optimality. Built on a strong theoretical
foundation, this model can provide insights to market participants
by splitting the overall move of the security price during the duration
of an order into the Market Impact (price move caused by their actions)
and Market Timing (price move caused by everyone else) components.}}
Plugging different distributions of prices and volumes into this framework
can help traders decide when to bear higher Market Impact by trading
more in the hope of offsetting the cost of trading at a higher price
later. We derive formulations of this model under different laws of
motion of the security prices. We start with a benchmark scenario
and extend this to include multiple sources of uncertainty, liquidity
constraints due to volume curve shifts and relating trading costs
to the spread (section \ref{sec:Extensions-Price-Motion}).

The unique aspect of our approach to trading costs is a method of
splitting the overall move of the security price during the duration
of an order into two components (Collins \& Fabozzi 1991; Treynor
1994; Yegerman \& Gillula 2014). One component gives the costs of
trading, that arise from the decision process that went into executing
that particular order, as captured by the price moves caused by the
executions that comprise that order. The other component gives the
costs of trading, that arise due to the decision process of all the
other market participants, during the time this particular order was
being filled. This second component is inferred, since it is not possible
to calculate it directly (at least with the present state of technology
and publicly available data) and it is the difference between the
overall trading costs and the first component, which is the trading
cost of the executions that make up that order alone. The first and
the second component arise due to competing forces, one from the actions
of a particular participant, and the other from the actions of everyone
else, that would be looking to fulfill similar objectives. 
\end{doublespace}

(Sections \ref{subsec:Deeper-Intuition-from}; \ref{subsec:Related-Literature})
seek to develop a deeper intuition for our methodology and review
the relevant literature. (Section \ref{sec:Dynamic-Recursive-Trading})
introduces the notation, terminology and has a discussion of foundational
concepts. (Sections \ref{sec:Alternative-and-Practical}; \ref{sec:Extensions-Price-Motion})
have the innovations from using our dynamic programming model under
different laws of motion of prices. \textbf{\textit{We develop a numerical
technique (section \ref{sec:Numerical-Framework-for}) that can be
used to obtain optimal executions under any law of motion of prices,
using a modification of the technique for pricing American options
(Longstaff \& Schwartz 2001). Our results demonstrate the tremendous
practical applicability of our theoretical framework including the
numerical techniques to implement them. The decomposition of trading
costs into Market Impact and Market Timing allows us to deduce the
zero sum game nature of trading costs (section \ref{subsec:Trading-Costs-as}).
It holds numerous lessons for dealing with complex systems, wherein
reducing the complexity by splitting the many sources of uncertainty
can lead to better insights in the decision process}}\footnote{\textbf{\textit{To elaborate on this, in any social system it would
be helpful to first distinguish the different participants and how
their actions contributes to uncertainty. If this is possible, then
understanding these components of uncertainty can sometimes help in
the analysis of social systems. For example, if we are looking to
analyze the shopping patterns in a mall, if we can distinguish shoppers
who buy on impulse and shoppers who buy after looking for discounts,
we might be better able to forecast sales and analyze this system
better. Also, our study can aid in the understanding of complex non-linear
phenomena, such as the evolution of prices in financial markets by
considering the price changes as being caused by multiple sources
of uncertainty. Such an approach of understanding the various sources
of uncertainty can be useful in the study of complicated physical
phenomena as well.}}}\textbf{\textit{.}}
\begin{doublespace}

\subsection{\label{subsec:Deeper-Intuition-from}Deeper Intuition from Realistic
Trading Situations }
\end{doublespace}

\begin{doublespace}
Naturally, it follows that each particular participant can only influence
to a greater degree the cost that arises from his actions as compared
to the actions of others over which he has lesser influence; but an
understanding of the second component can help him plan and alter
his actions to counter any adversity that might arise from the latter.
Any good trader would do this intuitively as an optimization process,
that would minimize costs over two variables direct impact and timing,
the output of which recommends either slowing down or speeding up
his executions. With our methodology, traders now actually have a
quantitative indicator to fine tune their decision process. When we
decompose the costs, it would be helpful to try and understand how
the two sub costs could vary as a proportion of the total. The volatility
in these two components, which would arise from different sources
(market conditions) would require different responses and hence would
affect the optimization problem mentioned above. Hence, based on an
understanding of the two components and the situation at hand, traders
would know which cost would be the more unpredictable one and hence
focus their efforts on minimizing the costs arising from that component.

The key innovation can be explained as follows: 
\end{doublespace}
\begin{enumerate}
\begin{doublespace}
\item A jump up in price on an execution that comprises a buy order is considered
adverse and attributed as impact, while a fall in price is not. Yes,
the price could fall further if not for the backstop provided by the
executions that comprise the buy order; but the key aspect to remember
here is the bilateral nature of trading. A price fall for the buyer
(or a benefit for him) is impact for the seller (and hence adverse);
and the seller bears the impact cost in this case. To understand this
better, we need to remember that if there is a lack of liquidity a
buyer can only bid up the price in the hopes of obtaining enough shares
to meet his demand and it is these jumps in price in a direction,
adverse to his direction of trading that are attributed as his market
impact.
\item Most trading cost models consider elaborate theories of the price
drifting around, but what actually happens during the transfer of
securities is one party, usually, has an upper hand and that is the
portion we look to measure as impact for the other party. The key
fallout from measuring impact this way is that we have a better way
to measure the effect of our actions from when we have a concrete
advantage, to when we are okay to put up with a certain disadvantage. 
\item The message from this reality is that despite our ambitions to optimize
the entire trading process, what we can control is the market impact
due to our trades; the market timing, which is the impact for our
counter parties is dependent on the decision process of these other
market participants and hence beyond the domain of what we can hope
(or choose) to optimize. 
\item While no measure of trading costs is perfect and complete, this methodology
goes a long way in actually providing tangible ways for someone to
understand the effect of their decision process and the associated
implementation of trades.
\end{doublespace}
\end{enumerate}
\begin{doublespace}
Another analogy to understand this methodology is to think of each
execution as effecting a state transition from one price level to
another. The impact is then the cost or charge involved to make the
state transition. We can also think of the change in price levels
as moving from one station to another in a train and the ticket price
is the cost involved to make this journey. If there is excess demand
to travel from one station to another, the ticket price, which is
the same for everyone at a particular point in time, changes accordingly
and only those that are willing to pay can make the journey. That
we are considering the state transitions for each execution at millisecond
intervals means that we are building from the bottom up and aggregating
smaller effects into an overall impact number for the order based
on the executions that comprise it. Theoretically since it is possible
that multiple parties could execute simultaneously (two or more buyers
and / or sellers on each side), the question of which of the parties
is more responsible for causing the price level to change and whether
there needs to be a proportional allocation of the price jump does
not set in, since all the parties are travelers on the same journey
and they all have to pay the ticket price. Though, for executions
that happen through a continuous auction process at larger intervals
of time, a proportional allocation based on the size of each parties
execution might be a possible alternative and will be pursued in later
papers.

(Figures \ref{fig:Reversion-Distributions-Shorter-Horizon}; \ref{fig:Reversion-Distributions-Longer-Horizon})
show the reversion\footnote{\begin{doublespace}
Reversion indicates the movement in prices, after an order has completed,
in the direction which is beneficial to the direction of trading.
For example it is considered positive if the movement is downwards
after a buy order has completed since buying moves the prices upwards.
\end{doublespace}
} in prices after an order has completed, broken down by volatility
and momentum buckets which are defined below. The full order sample
includes 148,812 institutional orders from 70+ countries with 17 countries
having at-least a thousand orders each \footnote{\begin{doublespace}
The actual financial market data cannot be disclosed due to confidentiality
reasons. The simulation data and related software can be made available
upon request.
\end{doublespace}
}. This global sample illustrates that the phenomenon is not restricted
to any single country. The reversion is based on two measures:
\end{doublespace}
\begin{enumerate}
\begin{doublespace}
\item In time, 5 minutes and 60 minutes after an order has completed.
\item In multiples of the order size, one times and five times the size
of the order.
\end{doublespace}
\end{enumerate}
\begin{doublespace}
The five Trade Momentum buckets are based on the side adjusted percentage
return during the order's trading interval:
\end{doublespace}
\begin{enumerate}
\begin{doublespace}
\item Significant Adverse (<-2\%) 
\item Adverse (-1/3\% thru -2\%) 
\item Neutral (-1/3\% thru +1/3\%) 
\item Favorable (+1/3\% thru 2\%) 
\item Significant Favorable (>+2\%) 
\end{doublespace}
\end{enumerate}
\begin{doublespace}
The four Trade Volatility buckets are based on the coefficient of
variation\footnote{\begin{doublespace}
In probability theory and statistics, the coefficient of variation
(CV), also known as relative standard deviation (RSD), is a standardized
measure of dispersion of a probability distribution or frequency distribution.
It is defined as the ratio of the standard deviation to the mean.
It shows the extent of variability in relation to the mean of the
population. \href{https://en.wikipedia.org/wiki/Coefficient_of_variation}{Coefficient of Variation, Wikipedia Link}
\end{doublespace}
}of prices during the execution horizon:
\end{doublespace}
\begin{enumerate}
\begin{doublespace}
\item High Volatility (>0.0050) 
\item Moderate Volatility (0.0010 thru 0.0050) 
\item Low Volatility (0.000000000000001 thru 0.0010) 
\item No Volatility (<= 0.000000000000001) 
\end{doublespace}
\end{enumerate}
\begin{doublespace}
(Figure \ref{fig:Reversion-Distributions-Shorter-Horizon}) has two
reversion indicators: 5 minutes on the left half and one times the
order size on the right half. (Figure \ref{fig:Reversion-Distributions-Longer-Horizon})
also has two reversion indicators: 60 minutes on the left half and
five times the order size on the right half. For each reversion indicator,
there are four volatility buckets on the X-axis and three momentum
buckets are on the Y-axis since we club together the orders in the
Significant Adverse and Adverse buckets; and the Significant Favorable
and Favorable buckets. Each order from the sample will be represented
as a bubble in all four reversion indicators and within one of the
four volatility buckets and one of the three momentum buckets. The
size of the bubbles indicates the relative magnitude of the order
and its position on the vertical axis signifies the reversion amount
in basis points. The box (grey region between two horizontal lines)
and the whisker (two horizontal lines separated by a vertical line)
capture the areas where 25\% and 75\% of the sample resides. 
\end{doublespace}

Not surprisingly, the momentum and reversion are higher in periods
of greater volatility, as seen more clearly from the measures based
on multiples of the order size (right half of Figures \ref{fig:Reversion-Distributions-Shorter-Horizon};
\ref{fig:Reversion-Distributions-Longer-Horizon}). The higher volatility
accentuates the efforts required to trade in such an environment.
This illustrates the issue that traders face and the optimization
process that is followed where they try to benefit from positive momentum
and try to avoid adverse momentum by trading more when adverse momentum
is anticipated, while being conscious of the level of volatility.
\begin{doublespace}

\subsection{\label{subsec:Related-Literature}Related Literature}
\end{doublespace}

\begin{doublespace}
\textbf{\textit{Building on the foundation laid by (Bertsimas \& Lo
1998), another popular way to decompose trading costs is into temporary
and permanent impact (see, Almgren \& Chriss 2001; Almgren 2003; and
Almgren, Thum, Hauptmann \& Li 2005). While the theory behind this
approach is extremely elegant and considers both linear and nonlinear
functions of the variables for estimating the impact, a practical
way to compute it requires measuring the price a certain interval
after the order. This interval is ambiguous and could lead to lower
accuracy while using this measure.}}

More recent extensions include: minimizing the mean and variance of
the costs of trading for the case of market orders only to derive
explicit formulas for the optimal trading strategies (Huberman \&
Stanzl 2005); considering quadratic variation as a reasonable risk
measure rather than variance, (Forsyth, Kennedy, Tse \& Windcliff
2012); the problem faced by an investor who must liquidate a given
basket of assets over a finite time horizon (Schied, Schöneborn \&
Tehranchi 2010); (Almgren \& Lorenz 2007) derive optimal strategies
where the execution accelerates when the price moves in the trader\textquoteright s
favor, and slows when the price moves adversely;\footnote{(Kissell \& Malamut 2006) term such adaptive strategies \textquotedblleft aggressive-in-the-money\textquotedblright ;
a \textquotedblleft passive-in-the-money\textquotedblright{} strategy
would react oppositely. They assume that the investor's utility has
constant absolute risk aversion (CARA) and that the asset prices are
given by a very general continuous-time, multi-asset price impact
model and show that the investor does no worse if he narrows his search
to deterministic strategies. CARA has exponential utility of the form
$u\left(c\right)=1-e^{-\alpha c}$, so that the absolute risk aversion,
$A\left(c\right)=-\frac{u''\left(c\right)}{u'\left(c\right)}=\alpha$,
a constant. \href{https://en.wikipedia.org/wiki/Risk_aversion}{Wikipedia Link on Risk Aversion}}

(Schied \& Schöneborn 2009) use a stochastic control approach\footnote{Stochastic control or stochastic optimal control is a subfield of
control theory that deals with the existence of uncertainty either
in observations or in the noise that drives the evolution of the system.
\href{https://en.wikipedia.org/wiki/Stochastic_control}{Wikipedia Link on Stochastic Control} }, building upon the continuous time model of (Almgren 2003), and show
that the value function and optimal control satisfy certain nonlinear
parabolic partial differential equations that can be solved numerically.
(Kato 2014) develops a mathematical model of optimal execution, by
formulating it as a stochastic control problem in the continuous time
domain. (Gatheral \& Schied 2011) find a closed-form solution for
the optimal trade execution strategy in the Almgren-Chriss framework
assuming the underlying unaffected stock price (stock price before
the impact or before the transaction occurs) process is a GBM; (Schied
2013) investigates the robustness of this strategy with respect to
misspecification of the law of the underlying unaffected stock price
process. (Guo \& Zervos 2015) study the optimal execution problem
in the context of a continuous time model with multiplicative price
impact, involving singular control rather than absolutely continuous
control. \footnote{\begin{doublespace}
In classical control problems (Shreve 1988), the cumulative displacement
of the state caused by control is the integral of the control process
(or some function of it), and so is absolutely continuous. In impulse
control, this cumulative displacement has jumps between which it is
either constant or absolutely continuous. Bounded variation control
(defined to include any stochastic control problem in which one restricts
the cumulative displacement of the state caused by control to be of
bounded variation on finite time intervals) admits both these possibilities
and also the possibility that the displacement of the state caused
by the optimal control is singularly continuous, at least with positive
probability over some interval of time.
\end{doublespace}
}

Building on empirical evidence (Lillo, Farmer \& Mantegna 2003) that
instantaneous market impact is a strongly concave function of the
volume, well approximated by a power law function at least for trading
rates that are not too high; (Curato, Gatheral \& Lillo 2017) find
that the discretized cost function exhibits a rugged landscape with
many local minima separated by peaks. (Huberman \& Stanzl 2004) provide
theoretical arguments showing that in the absence of quasi-arbitrage
(availability of a sequence of round-trip trades that generate infinite
expected profits with an infinite Sharpe ratio, that is infinite expected
profits per unit of risk), permanent price-impact functions must be
linear; though empirical investigations suggest that the shape of
the limit order book (LOB) can be more complex (Hopman 2007). (Gabaix,
Gopikrishnan \& Stanley 2006) present a theory in which spikes in
trading volume and returns, and hence stock market volatility, are
created by a combination of news and the trades of large investors
explaining the power law distribution of price impact. (Brunnermeier
\& Pedersen 2005; Carlin, Lobo \& Viswanathan 2007) are extensions
to situations with several competing traders, wherein if one trader
is forced to liquidate his holdings, other traders also sell creating
downward price pressure and buy back the assets later at a lower price.

\textbf{\textit{In contrast to many studies, where the dynamics of
the asset price process is taken as a given fundamental, (Obizhaeva
\& Wang 2013) proposed a market impact model that derives its dynamics
from an underlying model of a LOB. In this model, the ask part of
the LOB consists of a uniform distribution of shares offered at prices
higher than the current best ask price}}\textbf{.} 

(Alfonsi, Fruth \& Schied 2010) extend this by allowing for a general
shape of the LOB defined via a given density function, which can accommodate
empirically observed LOB shapes and obtain a nonlinear price impact
of market orders. (Predoiu, Shaikhet \& Shreve 2011) derive optimal
strategies, (under a general shape of the LOB), that are a mixture
of lump purchases and continuous purchases with the rate of purchase
set to match the order book resilience. (Fruth, Schöneborn \& Urusov
2014) analyze optimal strategies for a risk neutral investor when
liquidity varies deterministically (liquidity is time dependent; depth
and resilience can be independently time-dependent in contrast to
the LOB model of Obizhaeva \& Wang 2013) and find that in the case
of extreme changes in liquidity, it can even be optimal to completely
refrain from trading in periods of low liquidity. Empirical studies
based on the LOB model are (Biais, Hillion \& Spatt 1995; Potters
\& Bouchaud 2003; Bouchaud, Gefen, Potters \& Wyart 2004; Weber \&
Rosenow 2005).

A related strand of literature looks at models of the LOB from the
perspective of dealers seeking to submit optimal strategies (maximize
the utility of total terminal wealth) of bid and ask orders. (Ho \&
Stoll 1981) analyze the optimal prices for a monopolistic dealer in
a single stock when faced with a stochastic demand to trade, modeled
by a continuous time Poisson jump process, and facing return uncertainty,
modeled by diffusion processes. (Ho and Stoll 1980), consider the
problem of dealers under competition (each dealer's pricing strategy
depends not only on his own current and expected inventory position
and his other characteristics, but also on the current and expected
inventory and other characteristics of the competitor) and show that
the bid and ask prices are shown to be related to the reservation
(or indifference) prices of the agents.

(Cont, Stoikov \& Talreja 2010) describe a stylized model for the
dynamics of a limit order book, where the order flow is described
by independent Poisson processes, and estimate the model parameters
from high-frequency order book time-series data from the Tokyo Stock
Exchange. (Cont, Kukanov \& Stoikov 2014) study the price impact of
order book events - limit orders, market orders and cancellations
- using the NYSE Trades and Quotes data for fifty randomly selected
stocks. (Avellaneda \& Stoikov 2008) combine the utility framework
with the microstructure of actual limit order books, as described
in the econo-physics literature, to infer reasonable arrival rates
of buy and sell orders; (Du, Zhu \& Zhao 2016) extend the price dynamics
to follow a GBM in which the drift part is updated by Bayesian learning
in the beginning of the transaction day to capture the trader's estimate
of other traders\textquoteright{} target sizes and directions.

(Cont \& Kukanov 2017) focus on the order placement problem, which
is to choose an order type - market or limit order - and which trading
venue(s) to submit it to, when there are multiple alternatives. A
numerical algorithm for solving the order placement problem in a general
case is provided using a robust modification of the Robbins-Monro
stochastic approximation technique (Robbins \& Monro 1951; Nemirovski,
Juditsky, Lan \& Shapiro 2009). (Guo, de Larrard \& Ruan 2017) derive
optimal placement strategies for both static and dynamic cases (in
the static case, as opposed to the dynamic case, a strategy is completely
decided before execution takes place, that is at $t=0$, and is unchanged
over the entire order internal), under a correlated random walk model,
with mean-reversion for the best ask/bid price.

\textbf{\textit{While our work focuses on separating impact and timing
in the (Bertsimas \& Lo 1998) framework; a natural and interesting
continuation would be to extend this separation to models of the limit
order book discussed above (Obizhaeva \& Wang 2013). }}

Models of market impact and the design of better trading strategies
are becoming an integral part of the present trend at automation and
the increasing use of algorithms. (Jain 2005) assembles the dates
of announcement and actual introduction of electronic trading by the
leading exchange of 120 countries to examine the long term and medium
term impact of automation. He finds that automation of trading on
a stock exchange has a long-term impact on listed firms\textquoteright{}
cost of equity. (Hendershott, Jones \& Menkveld 2011) perform an empirical
study on New York Stock Exchange stocks and find that algorithmic
trading and liquidity are positively related. It is worth noting a
contrasting result from an earlier study. (Venkataraman 2001) compares
securities on the New York Stock Exchange (NYSE) (a floor-based trading
structure with human intermediaries, specialists and floor brokers)
and the Paris Bourse (automated limit-order trading structure). He
finds that execution costs might be higher on automated venues even
after controlling for differences in adverse selection, relative tick
size, and economic attributes. This means fully automated exchanges,
which anecdotally seems to be the way ahead, need to take special
care to formulate rules to help liquidity providers better control
the risks of order exposure. 
\end{doublespace}

\textbf{\textit{What this also means is that, the design of better
strategies and models is crucial to survive and thrive in this continuing
trend at automation. Our paper aims to fill the gap in existing models
of trading costs, which are theoretically elegant but are not readily
applicable to real life trading situations, since they do not allow
participants to gauge how they are performing in comparison to the
other participants with whom they are competing for liquidity. Our
models have a strong theoretical foundation but they can be applied
to actual trading situations due to the insights they provide to participants.
In addition, our numerical framework can be be used to obtain optimal
execution schedules under any law of motion of prices.}}
\begin{doublespace}

\section{\label{sec:Dynamic-Recursive-Trading}Dynamic Recursive Trading Cost
Model }
\end{doublespace}

\begin{doublespace}
A dynamic programming\footnote{Dynamic programming is both a mathematical optimization method and
a computer programming method. Developed by Richard Bellman in the
1950s, it has found applications in numerous fields, from aerospace
engineering to economics. The technique refers to simplifying a complicated
problem by breaking it down into simpler sub-problems in a recursive
manner. While some decision problems cannot be taken apart this way,
decisions that span several points in time do often break apart recursively.
We suggest that this technique has been hinted at in several works
of Eastern philosophy that boils down to: Do your best at this moment
based on the present situation and the best of your abilities, forget
(don\textquoteright t worry) about the future (results) and the best
that can happen will happen (Swami 1983; Stokey, Lucas \& Prescott
1989; \href{https://en.wikipedia.org/wiki/Dynamic_programming}{Dynamic Programming, Wikipedia Link}).} approach lends itself naturally to modeling optimal execution strategies.
(Bertsimas \& Lo 1998) start with a simple arithmetic random walk
for the law of motion of prices and later extend it to a Geometric
Brownian motion. Their approach and extensions result in closed form
or numerical solutions for many scenarios (section \ref{subsec:Related-Literature}).
As (section \ref{subsec:Related-Literature}) also highlights, existing
dynamic programming methods to optimizing trading costs and execution
scheduling are of limited use to practitioners and traders since they
do not provide a way for them to understand how their actions at each
stage would affect the price (as opposed to the combined effect of
everyone else or the market) and thereby pointing out specific aspects
of the system that they can hope to influence. Hence, we start with
the benchmark dynamic programming problem (section \ref{subsec:Benchmark-Dynamic-Programming})
and modify the reward function in the Bellman equation to suit our
innovation in later sections.

The objective of any trading program is to formulate a trading trajectory,
or a list of total pending shares, $W_{1},...,W_{T+1}$ at the end
of each time period\footnote{\begin{doublespace}
We define all the variables as we introduce them in the text but (Appendix
\ref{subsec:Dictionary-of-Notation}) has a complete dictionary of
all the notation.
\end{doublespace}
}. Here, $T$ is the total duration of trading. For simplicity, time
in measured in unit intervals giving, $t=1,2,...,T$. $W_{t}$ then
becomes the number of units that we still need to trade at time $t$.
$\bar{S}$ is the total number of shares that need to be traded. Conditions
($W_{1}=\bar{S}\;\&\;W_{T+1}=0$ ) together imply that $\bar{S}$
must be executed by period $T$ (this is an assumption that there
will be no unexecuted shares once the total time duration is completed;
this is a constraint to be satisfied while seeking the trading schedule). 
\end{doublespace}

A trading strategy can equivalently be represented by the list of
executions completed, $S_{1},...,S_{T}$. $S_{t}$ is the number of
shares acquired in period $t$ at price $P_{t}$. Clearly, $\bar{S}=\underset{j=1}{\overset{T}{\sum}}S_{j}$.
This gives, $W_{t}=W_{t-1}-S_{t-1}$ or $S_{t-1}=W_{t-1}-W_{t}$ is
the number of units traded between times $t-1$ and $t$ at price
$P_{t-1}$. That is we go from $W_{t-1}$ unexecuted shares at time
period $t-1$ to $W_{t}$ remaining shares at time $t$ by filling
$S_{t-1}$ shares at price $P_{t-1}$. $W_{t}$ and $S_{t}$ are related
as below. 
\begin{equation}
W_{t}=\bar{S}-\underset{j=1}{\overset{t-1}{\sum}}S_{j}=\underset{j=t}{\overset{T}{\sum}}S_{j}\qquad,t=1,...,T.
\end{equation}

\begin{doublespace}

\subsection{\label{subsec:Benchmark-Dynamic-Programming}Benchmark Dynamic Programming
Model}
\end{doublespace}

\begin{doublespace}
This is the simplest scenario where the trader would try to minimize
the overall acquisition value of his holdings. This is also the benchmark
scenario in (Bertsimas \& Lo 1998). In this case, securities are being
bought. It is then logical to set a no sales constraint when the objective
is to buy securities. The baseline objective function and constraints
are written as, 
\begin{equation}
\underset{\left\{ S_{t}\right\} }{\min}\:E_{1}\left[\sum_{t=1}^{T}S_{t}P_{t}\right]
\end{equation}
\begin{equation}
\sum_{t=1}^{T}S_{t}=\bar{S}\;,S_{t}\geq0\;,W_{1}=\bar{S},\;W_{T+1}=0\;,\;W_{t}=W_{t-1}-S_{t-1}
\end{equation}
The law of motion of price, $P_{t}$ for the buy scenario can be written
as, 
\begin{equation}
P_{t}=P_{t-1}+\theta S_{t}+\varepsilon_{t}\;,\theta>0\:,E\left[\varepsilon_{t}\left|S_{t},P_{t-1}\right.\right]=0
\end{equation}
\[
\varepsilon_{t}\sim N\left(0,\sigma_{\varepsilon}^{2}\right)\equiv\text{Zero Mean IID (Independent Identically Distributed) random shock or white noise}
\]
We also follow the convention that the shares are positive when we
buy and negative when we sell. The law of motion of price, $P_{t}$
for the sell scenario then becomes, 
\begin{equation}
P_{t}=P_{t-1}-\theta S_{t}+\varepsilon_{t}\;,\theta>0\:,E\left[\varepsilon_{t}\left|S_{t},P_{t-1}\right.\right]=0
\end{equation}
This price evolution and convention for the buy and sell scenarios
ensures that the buyer and the seller have the same price. A trade
happens only when the buyer and seller agree upon the price and they
both face the same shock in this case. In the rest of the discussion
we only consider the price evolution for the buy scenario since this
treatment applies with simple modifications when securities are sold. 
\end{doublespace}

The law of motion includes two distinct components: the dynamics of
$P_{t}$ in the absence of our trade, (the trades of other may be
causing prices to fluctuate) and the impact that our trade of $S_{t}$
shares has on the execution price $P_{t}$. This simple price change
relationship assumes that the former component is given by an arithmetic
random walk and the latter component is a linear function of trade
size so that a purchase of $S_{t}$ shares may be executed at the
prevailing price plus an impact premium of $\theta S_{t}$. Here,
$\theta$ captures the effect of transaction size on the price. In
the absence of this transaction, the price process evolves as a pure
arithmetic random walk. This then implies that from any participants
view, the sum of all the price movements or the new price levels established
by all other participants evolves as a random walk. For simplicity,
we ignore the no sales constraint, $S_{t}\geq0$. 

\begin{doublespace}
The Bellman equation is based on the observation that a solution or
optimal control $\left\{ S_{1}^{*},S_{2}^{*},...,S_{T}^{*}\right\} $
must also be optimal for the remaining program at every intermediate
time $t.$ That is, for every $t,$ $1<t<T$ the sequence $\left\{ S_{t}^{*},S_{t+1}^{*},...,S_{T}^{*}\right\} $
must still be optimal for the remaining program $E_{t}\left[\sum_{k=t}^{T}P_{k}S_{k}\right]$.
The below relates the optimal value of the objective function $V_{t}$
in period $t$ to its optimal value in period $t+1:$
\begin{equation}
V_{t}\left(P_{t-1},W_{t}\right)=\underset{\left\{ S_{t}\right\} }{\min}\:E_{t}\left[P_{t}S_{t}+V_{t+1}\left(P_{t},W_{t+1}\right)\right]
\end{equation}

\end{doublespace}
\begin{doublespace}

\subsection{Terminology and High-Level Mathematical Expressions}
\end{doublespace}

\begin{doublespace}
We now introduce some terminology used throughout the discussion.
We also provide simple mathematical expressions to convey the intuition
behind our methodology. (Sections \ref{sec:Alternative-and-Practical};
\ref{sec:Extensions-Price-Motion}) have a rigorous mathematical treatment.
\end{doublespace}
\begin{enumerate}
\begin{doublespace}
\item \label{enu:Total-Slippage}Total Slippage: The overall price move
on the security during the order duration. This is also a proxy for
the implementation shortfall (Perold 1988; Treynor 1981; section \ref{subsec:The-Implementation-Shortfall}).
\footnote{\begin{doublespace}
It is worth mentioning that there are many similar metrics used in
practice and this concept gets used in situations for which it is
not ideally suited (Yegerman \& Gillula 2014). While the usefulness
of the Implementation Shortfall, or slippage, as a measure to understand
the price shortfalls that can arise between constructing a portfolio
and implementing it is not to be debated; slippage needs to be supplemented
with more granular metrics when used in situations where the effectiveness
of algorithms or the availability of liquidity need to be gauged. 
\end{doublespace}
}
\item \label{enu:Market-Impact-(MI)}Market Impact (MI): The price moves
caused by the executions that comprise the order under consideration.
In short, the MI is a proxy for the impact on the price from the liquidity
demands of an order. This metric is generally negative or zero since
in most cases the best impact we can have is usually no impact\footnote{\begin{doublespace}
This metric is negative if we follow a convention to show it as a
cost; but we indicate it as positive quantity throught the paper 
\end{doublespace}
}. 
\item \label{enu:Market-Timing}Market Timing: The price moves that happen
due to the combined effect of all the other market participants during
the order duration. 
\item Market Impact Estimate (MIE): An estimate of the Market Impact, (point
\ref{enu:Market-Impact-(MI)}), based on recent market conditions.
The MIE calculation is the result of a simulation which considers
the number of executions required to fill an order and the price moves
encountered while filling this order. It depends on the market micro-structure
as captured by the trading volume and the price probability distribution
that factors upticks and down-ticks. This simulation can be controlled
with certain parameters that dictate the liquidity demanded on the
order, the style of trading, order duration, start and end of trading
times. In short, the MIE is an estimated proxy for the impact on the
price from the liquidity demands of an order. 
\item Market Timing Estimate (MTE): This is an estimate of the Market Timing,
(point \ref{enu:Market-Timing}), based on recent market conditions.
The MTE calculation is highly dependent on the price volatility and
hence the longer the duration, the higher we can expect the timing
to be. It is helpful to consider an upper bound and lower bound for
the MTE or a range for the MTE for the duration of trading.\footnote{\begin{enumerate}
\begin{doublespace}
\item The following equations, expressed in simple mathematical terms to
facilitate easier understanding, govern the relationships between
the variables mentioned above.
\end{doublespace}
\begin{itemize}
\begin{doublespace}
\item Total Slippage = Market Impact + Market Timing 
\item \{Total Price Slippage = Your Price Impact + Price Impact From Everyone
Else (Price Drift)\} 
\item Market Impact Estimate = Market Impact Prediction = $f$ (Execution
Size, Liquidity Demand) 
\item Execution Size = $g$(Execution Parameters, Market Conditions) 
\item Liquidity Demand = $h$(Execution Parameters, Market Conditions) 
\item Execution Parameters <->vector comprising (Order Size, Security, Side,
Trading Style, Timing Decisions) 
\item Market Conditions <-> vector comprising (Price Movement, Volume Changes,
Information Set)
\end{doublespace}
\end{itemize}
\begin{doublespace}
\item Here, $f,g,h$ are functions. We could impose concavity conditions
on these functions, but arguably, similar results are obtained by
assuming no such restrictions and fitting linear or non-linear regression
coefficients, which could be non-concave or even discontinuous allowing
for jumps in prices and volumes. The specific functional forms used
could vary across different groups of securities or even across individual
securities or even across different time periods for the same security.
The crucial aspect of any such estimation is the comparison with the
costs on real orders, as outlined earlier. Simpler models are generally
more helpful in interpreting the results and for updating the model
parameters. (Hamilton 1994) and (Gujarati 1995) are classic texts
on econometric methods and time series analysis that accentuate the
need for parsimonious models.
\item All the variables are measured in basis points to facilitate ease
of comparison and aggregation across different groups. It is possible
to measure these in cents per share and also in dollar value or other
currency terms. 
\end{doublespace}
\end{enumerate}
}
\end{doublespace}
\end{enumerate}
\begin{doublespace}

\subsection{\label{subsec:The-Implementation-Shortfall}The Implementation Shortfall}
\end{doublespace}

\begin{doublespace}
As a refresher, the total slippage or implementation shortfall is
derived below with the understanding that we need to use the Expectation
operator when we are working with estimates or future prices. $P_{0}$
can be any reference price or benchmark used to measure the slippage.
It is generally taken to be the arrival price or the price at which
the portfolio manager would like to complete the purchase of the portfolio.
\footnote{\begin{doublespace}
(Kissell 2006) provides more details including the formula where the
portfolio may be partly executed. 
\end{doublespace}
}
\begin{equation}
\text{Paper Return}=\bar{S}P_{T}-\bar{S}P_{0}
\end{equation}
\begin{equation}
\text{Real Portfolio Return}=\bar{S}P_{T}-\left(\sum_{t=1}^{T}S_{t}P_{t}\right)
\end{equation}
\begin{align}
\text{Implementation Shortfall} & =\text{Paper Return}-\text{Real Portfolio Return}\\
 & =\left(\sum_{t=1}^{T}S_{t}P_{t}\right)-\bar{S}P_{0}
\end{align}
This can be written as,
\begin{align}
\text{Implementation Shortfall} & =\left(\sum_{t=1}^{T}S_{t}P_{t}\right)-\bar{S}P_{0}\\
 & =\left(\sum_{t=1}^{T}S_{t}P_{t}\right)-P_{0}\left(\sum_{t=1}^{T}S_{t}\right)\\
 & =S_{1}\left(P_{1}-P_{0}\right)+S_{2}\left(P_{2}-P_{0}\right)+...+S_{T}\left(P_{T}-P_{0}\right)
\end{align}
\begin{align}
\text{Implementation Shortfall} & =S_{1}\left(P_{1}-P_{0}\right)\\
 & +S_{2}\left(P_{2}-P_{1}\right)+S_{2}\left(P_{1}-P_{0}\right)\\
 & +S_{3}\left(P_{3}-P_{2}\right)+S_{3}\left(P_{2}-P_{1}\right)+S_{3}\left(P_{1}-P_{0}\right)+\;...\;\\
 & +S_{T}\left(P_{T}-P_{T-1}\right)+S_{T}\left(P_{T-1}-P_{T-2}\right)+...+S_{T}\left(P_{1}-P_{0}\right)
\end{align}
The innovation we introduce would incorporate our earlier discussion
about breaking the total impact or slippage, Implementation Shortfall,
into the part from the participants own decision process, Market Impact,
and the part from the decision process of all other participants,
Market Timing. This Market Impact, would capture the actions of the
participant, since at each stage the penalty a participant incurs
should only be the price jump caused by their own trades and that
is what any participant can hope to minimize. A subtle point is that
the Market Impact portion need only be added up when new price levels
are established. If the price moves down and moves back up (after
having gone up once earlier and having been already counted in the
Impact), we need not consider the later moves in the Market Impact
(and hence implicitly left out from the Market Timing as well). This
alternate measure, which does not consider subsequent price moves
down and up after having gone up once earlier, would only account
for the net move in the prices but would not show the full extent
of aggressiveness and the push and pull between market participants
and hence is not considered here, though it can be useful to know
and can be easily incorporated while running simulations. We discuss
two formulations of our measure of the Market Impact in (sections
\ref{subsec:Market-Impact-Simple}; \ref{subsec:Market-Impact-Complex}).
The reason for calling them simple and complex will become apparent
as we continue the discussion.
\end{doublespace}
\begin{doublespace}

\subsection{\label{subsec:Market-Impact-Simple}Market Impact Simple Formulation}
\end{doublespace}

\begin{doublespace}
The simple market impact formulation does not consider the impact
of the new price level established on all the future trades that are
yet to be done. From a theoretical perspective it is useful to study
this since it provides a closed form solution and illustrates the
immense practical application of separating impact and timing. This
approach can be a useful aid in markets that are clearly not trending
and where the order size is relatively small compared to the overall
volume traded, ensuring that any new price level established does
not linger on for too long and prices gets reestablished due to the
trades of other participants. This property is akin to checking that
shocks to the system do not take long to dissipate and equilibrium
levels (or rather new pseudo equilibrium levels) are restored quickly.
Our measure of the Market Impact then becomes, 
\begin{equation}
\text{Market Impact}=\sum_{t=1}^{T}\left\{ \max\left[\left(P_{t}-P_{t-1}\right),0\right]S_{t}\right\} 
\end{equation}
The Market Timing is then given by,
\begin{align}
\text{\text{Market Timing}} & =\text{Implementation Shortfall}-\text{\text{Market Impact}}\\
 & =\left(\sum_{t=1}^{T}S_{t}P_{t}\right)-\bar{S}P_{0}-\sum_{t=1}^{T}\left\{ \max\left[\left(P_{t}-P_{t-1}\right),0\right]S_{t}\right\} 
\end{align}
(Appendix \ref{subsec:Market-Impact-Simple-Example}) has some illustrative
examples. 
\end{doublespace}
\begin{doublespace}

\subsection{\textcolor{black}{\label{subsec:Market-Impact-Complex}Market Impact
Complex Formulation}}
\end{doublespace}

\begin{doublespace}
Another measure of the Market Impact can be formulated as below which
represents the idea that when a participant seeks liquidity and establishes
a new price level, all the pending shares or the unexecuted program
is affected by this new price level. This is a more realistic approach
since the action now will explicitly affect the shares that are not
yet executed. This measure can be written as, 
\begin{equation}
\text{Market Impact}=\sum_{t=1}^{T}\left\{ \max\left[\left(P_{t}-P_{t-1}\right),0\right]W_{t}\right\} 
\end{equation}
The Market Timing is then given by,
\begin{align}
\text{\text{Market Timing}} & =\text{Implementation Shortfall}-\text{\text{Market Impact}}\\
 & =\left(\sum_{t=1}^{T}S_{t}P_{t}\right)-\bar{S}P_{0}-\sum_{t=1}^{T}\left\{ \max\left[\left(P_{t}-P_{t-1}\right),0\right]W_{t}\right\} 
\end{align}
(Appendix \ref{subsec:Market-Impact-Complex-Example}) has some illustrative
examples. 
\end{doublespace}
\begin{doublespace}

\subsection{\label{subsec:Trading-Costs-as}Trading Costs as a Zero Sum Game}
\end{doublespace}

\begin{doublespace}
A formal study of trading costs in the financial markets using the
tools of game theory can lead to many interesting conclusions\footnote{\begin{doublespace}
(Fama 1970) is a discussion of fair games and efficient markets; (Kyle
1985, Foster \& Viswanathan 1990) solve for the Nash equilibrium when
trading is viewed as a game between market makers and traders; (Hill
1990) considers transaction costs using a game theoretic model with
opportunistic behavior; (Klemperer 2004) is an overview of how auctions
can explain financial crashes and trading frenzies.
\end{doublespace}
}. Even without a set up specific to game theory one of the results
we obtain, though fairly evident but perhaps surprising given the
extent of trading that takes place in today's markets, is that in
any given time period the sum of market impact and the sum of market
timing across all market participants equals zero.

This is immediately obvious in the case that there are only two participants
(one is the buyer, the other is the seller and without two participants
we do not have a market or a trade) and there is only one single interval,
since negative implementation shortfall for the buyer shows up as
positive implementation shortfall for the seller; the impact for the
buyer shows up as timing for the seller and vice versa. We note that
the total amount bought in any interval is equal to the total amount
sold. When there are more than two participants and multiple intervals,
if we consider the actions in each interval and add up the impact
and timing figures across everyone, it shows the zero sum nature of
the trading game \footnote{\begin{doublespace}
For different types of zero sum games and methods of solving them,
see: (Brown 1951; Gale, Kuhn \& Tucker 1951; Von Neumann \& Morgenstern
1953; Von Neumann 1954; Rapoport 1973; Crawford 1974; Laraki \& Solan
2005; Hamadène 2006); (Bodie \& Taggart 1978; Bell \& Cover 1980;
Turnbull 1987; Hill 2006; Chirinko \& Wilson 2008) consider zero sum
games in the financial context.
\end{doublespace}
}. The result holds for both the simple and complex formulations of
market impact.
\end{doublespace}
\begin{thm}
\begin{doublespace}
\label{Trading-costs-Zero-Sum-Game}Trading costs are a zero sum game.
The sum of market impact and market timing across all participants,
in any given time interval, should equal zero. 
\[
\text{Total Market Impact}+\text{Total \text{Market Timing}}=0
\]
\end{doublespace}
\end{thm}
\begin{proof}
\begin{doublespace}
Appendix \ref{subsec:Proof-of-Proposition-Trading-Costs-Zero-Sum}. 
\end{doublespace}
\end{proof}
\begin{doublespace}
Though we refrain from a longer discussion for the sake of brevity;
it should be immediately apparent that the zero sum nature of trading
costs is applicable outside the financial markets to all manner of
trades within international / intra-national finance and the exchange
of all types of goods and services. Another aspect we point out is
the difference in the proportion of timing and impact between financial
markets and trading in other products. The relative ease with which
products can be liquidated and / or the extent to which they are either
consumption or investment goods, affects this property (Kashyap 2014).
\end{doublespace}
\begin{doublespace}

\section{\label{sec:Alternative-and-Practical}Alternative and Practical Dynamic
Market Impact Model}
\end{doublespace}

\begin{doublespace}
In this section, we discuss the benchmark law of motion of prices
while optimizing the simple and complex market impact formulations.
Other extensions of the law of motion of prices are considered in
Section \ref{sec:Extensions-Price-Motion} in the Appendix. Prices
can be negative under the benchmark law of motion. Section \ref{subsec:Optimization-Linear-Percentage-Impact}
in the Appendix shows how to set up the numerical solution when the
no impact prices evolve as a Geometric Brownian Motion, which ensures
that prices remain positive under mild restrictions on the parameters. 
\end{doublespace}
\begin{doublespace}

\subsection{Simple Formulation of the Benchmark Law of Price Motion}
\end{doublespace}

\begin{doublespace}
Incorporating the Simple Market Impact formulation from section \ref{subsec:Market-Impact-Simple},
the benchmark objective function and the Bellman equation from section
\ref{subsec:Benchmark-Dynamic-Programming} can be modified as,
\begin{equation}
\underset{\left\{ S_{t}\right\} }{\min}\:E_{1}\left[\sum_{t=1}^{T}\left\{ \max\left[\left(P_{t}-P_{t-1}\right),0\right]S_{t}\right\} \right]
\end{equation}
\begin{equation}
\sum_{t=1}^{T}S_{t}=\bar{S}\;,S_{t}\geq0\;,W_{1}=\bar{S},\;W_{T+1}=0\;,\;W_{t}=W_{t-1}-S_{t-1}
\end{equation}
\begin{equation}
P_{t}=P_{t-1}+\theta S_{t}+\varepsilon_{t}\;,\theta>0\:,E\left[\varepsilon_{t}\left|S_{t},P_{t-1}\right.\right]=0\;,\;\varepsilon_{t}\sim N\left(0,\sigma_{\varepsilon}^{2}\right)
\end{equation}
The Bellman equation then becomes, 
\begin{equation}
V_{t}\left(P_{t-1},W_{t}\right)=\underset{\left\{ S_{t}\right\} }{\min}\:E_{t}\left[\max\left\{ \left(P_{t}-P_{t-1}\right),0\right\} S_{t}+V_{t+1}\left(P_{t},W_{t+1}\right)\right]
\end{equation}
\textbf{\textit{One additional constraint that is necessary is to
restrict the amount of shares available for trading in any time period
when the price in that time period drops in comparison to the previous
time period. The algorithm in section \ref{sec:Numerical-Framework-for}
shows how these constraints can be set. This is a practical consideration,
since a drop in price is impact for the sellers and timing for the
buyers (as a reminder, we are buyers). Hence when the price decreases
in comparison to the previous time period, the amount of shares or
liquidity is limited and the seller decides how much to make available.
When prices are rising, we can justify not having that criteria, since
the buyer can bid up the price and decide how much impact they want
to incur. A more thorough approach would ensure that the liquidity
follows a process of its own and captures this dynamic of sellers
and buyers being able to prop the prices from falling further or rising
higher respectively. In the extension we consider in section \ref{subsec:Introducing-Liquidity-Constraint},
some of these aspects can be factored in.}}

By starting at the end, (time $T$) and applying the modified Bellman
equation, the law of motion for $P_{t}$, the relation between pending
and executed shares, and the boundary conditions recursively, the
optimal control can be derived as functions of the state variables
that characterize the information that the investor must have to make
his decision in each period. In particular, the optimal value function,
$V_{T}\left(\cdots\right)$, as a function of the two state variables
$P_{T-1}$ and $W_{T}$ is given by, 
\begin{equation}
V_{T}\left(P_{T-1},W_{T}\right)=\underset{\left\{ S_{T}\right\} }{\min}\:E_{T}\left[\max\left\{ \left(P_{T}-P_{T-1}\right),0\right\} S_{T}\right]
\end{equation}
Here, the remaining shares $W_{T+1}$ must be zero since there is
no choice but to execute all the remaining shares, $W_{T}$. We then
have the optimal trade size, $S_{T}^{*}=W_{T}$ and an expression
for $V_{T}$ as,
\begin{equation}
V_{T}\left(P_{T-1},W_{T}\right)=E_{T}\left[\max\left\{ \left(\theta W_{T}+\varepsilon_{T}\right),0\right\} W_{T}\right]
\end{equation}

\end{doublespace}
\begin{prop}
\begin{doublespace}
\label{The-value-function-convexity}The value function for the last
but one time period is convex and can be written as, 
\begin{eqnarray*}
V_{T-1}\left(P_{T-2},W_{T-1}\right) & = & \underset{\left\{ S_{T-1}\right\} }{\min}\left[S_{T-1}\sigma_{\varepsilon}\psi\left(\xi S_{T-1}\right)+\left(W_{T-1}-S_{T-1}\right)\sigma_{\varepsilon}\psi\left\{ \xi\left(W_{T-1}-S_{T-1}\right)\right\} \right]\\
 &  & \text{Here, }\psi\left(u\right)=u+\phi\left(u\right)/\Phi\left(u\right)\;,\;\xi=\frac{\theta}{\sigma_{\varepsilon}},
\end{eqnarray*}
Also, $\phi$ and $\mathbf{\Phi}$ are the standard normal Probability
Density Function, PDF, and Cumulative Distribution Function CDF, respectively.
\end{doublespace}
\end{prop}
\begin{proof}
\begin{doublespace}
Appendix \ref{subsec:Proof-of-Proposition-Convexity}.
\end{doublespace}
\end{proof}
\begin{doublespace}
Figure \ref{fig:Convexity-of-Distribution} illustrates the shape
of some combinations of the distribution functions that we are working
with. For the value function we have, the condition for convexity
can be derived as $\theta>\left(3\sigma_{\varepsilon}/4\right)$. 
\end{doublespace}
\begin{prop}
\begin{doublespace}
\label{The-number-of-benchmark-simple}The number of shares to be
executed in each time period follows a linear law. $S_{T-1}^{*}=W_{T-1}^{*}/2,\quad\ldots\quad,S_{T-K-1}^{*}=W_{T-K-1}/\left(K+2\right)$
$,S_{T-K}^{*}=W_{T-K}/\left(K+1\right)$ and the corresponding value
functions are, 
\[
V_{T-K-1}\left(P_{T-K-2},W_{T-K-1}\right)=\sigma_{\varepsilon}W_{T-K-1}\left[\frac{\theta}{\sigma_{\varepsilon}}\frac{W_{T-K-1}}{\left(K+2\right)}+\frac{\phi\left(\frac{\theta}{\sigma_{\varepsilon}}\frac{W_{T-K-1}}{\left(K+2\right)}\right)}{\Phi\left(\frac{\theta}{\sigma_{\varepsilon}}\frac{W_{T-K-1}}{\left(K+2\right)}\right)}\right]
\]
\end{doublespace}
\end{prop}
\begin{proof}
\begin{doublespace}
Appendix \ref{subsec:Proof-of-Proposition-benchmark-simple}
\end{doublespace}
\end{proof}
\begin{doublespace}
We can see that a minimum exists at each stage. The simple solution
follows from the linear rule where the price impact $\theta S_{t}$
does not depend on either the prevailing price, $P_{t-1}$, or the
size of the unexecuted order $W_{t}$ and hence the price impact function
is the same in each period and independent from one period to the
next. It is easily shown that, $S_{1}^{*}=S_{2}^{*}=\cdots=S_{T}^{*}=\bar{S}/T$.
This simply means that the best execution strategy is simply to divide
the total order or the total shares $\bar{S}$ into $T$ equal amounts
and trade them at regular intervals. (Bertsimas \& Lo 1998) has a
more detailed discussion. Supposing a closed form solution was absent,
we could approximate the solution (numerically solved) using $S_{T-1}^{*}\approx\xi_{0}+\xi_{1}W_{T-1}+\xi_{2}\left(W_{T-1}\right)^{2}$
or $S_{T-1}^{*}\approx\xi_{0}\left(W_{T-1}\right)^{\xi_{1}}$. We
can also set $S_{T-1}^{*}\approx\omega_{1}\left(W_{T-1}\right)$ using
any well behaved (continuous and differentiable) function, $\omega_{1}$.
We could also include the last known price, $P_{t-1}$, or other state
variables into the above approximation. We discuss this technique
in detail including numerical examples in section \ref{sec:Numerical-Framework-for}.
This numerical approximation approach is simple to implement and lends
itself easily to solutions even in the more complex laws of motion
to follow in section \ref{sec:Extensions-Price-Motion}.

Going forward, to lighten the notion, we will drop the {*} superscript
on the number of shares to be executed in each time period, $S_{T}^{*}$,
where there is less likelihood of confusion.
\end{doublespace}
\begin{doublespace}

\subsection{Complex Formulation of the Benchmark Law of Price Motion}
\end{doublespace}

\begin{doublespace}
Incorporating the Complex Market Impact formulation from the earlier
section \ref{subsec:Market-Impact-Complex}, the objective function
and the Bellman equation from section \ref{subsec:Benchmark-Dynamic-Programming}
can be modified as,
\begin{equation}
\underset{\left\{ S_{t}\right\} }{\min}\:E_{1}\left[\sum_{t=1}^{T}\left\{ \max\left[\left(P_{t}-P_{t-1}\right),0\right]W_{t}\right\} \right]
\end{equation}
\begin{equation}
\sum_{t=1}^{T}S_{t}=\bar{S}\;,S_{t}\geq0\;,W_{1}=\bar{S},\;W_{T+1}=0\;,\;W_{t}=W_{t-1}-S_{t-1}
\end{equation}
\begin{equation}
P_{t}=P_{t-1}+\theta S_{t}+\varepsilon_{t}\;,\theta>0\:,E\left[\varepsilon_{t}\left|S_{t},P_{t-1}\right.\right]=0\;,\;\varepsilon_{t}\sim N\left(0,\sigma_{\varepsilon}^{2}\right)
\end{equation}
The Bellman equation then becomes, 
\begin{equation}
V_{t}\left(P_{t-1},W_{t}\right)=\underset{\left\{ S_{t}\right\} }{\min}\:E_{t}\left[\max\left\{ \left(P_{t}-P_{t-1}\right),0\right\} W_{t}+V_{t+1}\left(P_{t},W_{t+1}\right)\right]
\end{equation}
The optimal value function, $V_{T}\left(\cdots\right)$, as a function
of the two state variables $P_{T-1}$ and $W_{T}$ is given by, 
\begin{equation}
V_{T}\left(P_{T-1},W_{T}\right)=\underset{\left\{ S_{T}\right\} }{\min}\:E_{T}\left[\max\left\{ \left(P_{T}-P_{T-1}\right),0\right\} W_{T}\right]
\end{equation}
Here, the remaining shares $W_{T+1}$ must be zero since there is
no choice but to execute all the remaining shares, $W_{T}$. We then
have the optimal trade size, $S_{T}^{*}=W_{T}$ and an expression
for $V_{T}$ as,
\begin{equation}
V_{T}\left(P_{T-1},W_{T}\right)=E_{T}\left[\max\left\{ \left(\theta W_{T}+\varepsilon_{T}\right),0\right\} W_{T}\right]
\end{equation}

\end{doublespace}
\begin{prop}
\begin{doublespace}
\label{The-number-of-benchmark-complex}The value function for the
last but one time period is a convex function with a unique minimum,
since it is the sum of the portions shown to be convex above (Proposition
\ref{The-value-function-convexity}), another convex function and
a linear component.
\[
V_{T-1}\left(P_{T-2},W_{T-1}\right)=\underset{\left\{ S_{T-1}\right\} }{\min}\left[W_{T-1}\sigma_{\varepsilon}\psi\left(\xi S_{T-1}\right)+\left(W_{T-1}-S_{T-1}\right)\sigma_{\varepsilon}\psi\left\{ \xi\left(W_{T-1}-S_{T-1}\right)\right\} \right]
\]
\[
\text{Here, }\psi\left(u\right)=u+\phi\left(u\right)/\Phi\left(u\right)\;;\;\xi=\frac{\theta}{\sigma_{\varepsilon}}\;;\;\text{Note that, }W_{T-1}=S_{T-1}+W_{T}
\]
The number of shares to be executed in subsequent time periods and
the corresponding value function are obtained by solving, 

\begin{eqnarray*}
W_{T-1}+\frac{\xi\left(W_{T-1}-S_{T-1}\right)^{2}\phi\left(\xi\left\{ W_{T-1}-S_{T-1}\right\} \right)}{\Phi\left(\xi\left\{ W_{T-1}-S_{T-1}\right\} \right)}+\left(W_{T-1}-S_{T-1}\right)\left[\frac{\phi\left(\xi\left\{ W_{T-1}-S_{T-1}\right\} \right)}{\Phi\left(\xi\left\{ W_{T-1}-S_{T-1}\right\} \right)}\right]^{2} &  & =\\
2\left(W_{T-1}-S_{T-1}\right)+\frac{1}{\xi}\frac{\phi\left(\xi\left\{ W_{T-1}-S_{T-1}\right\} \right)}{\Phi\left(\xi\left\{ W_{T-1}-S_{T-1}\right\} \right)}+\frac{\xi W_{T-1}S_{T-1}\phi\left(\xi S_{T-1}\right)}{\Phi\left(\xi S_{T-1}\right)}+W_{T-1}\left[\frac{\phi\left(\xi S_{T-1}\right)}{\Phi\left(\xi S_{T-1}\right)}\right]^{2}
\end{eqnarray*}
\end{doublespace}
\end{prop}
\begin{proof}
\begin{doublespace}
Appendix \ref{subsec:Proof-of-Proposition-benchmark-complex}. 
\end{doublespace}
\end{proof}
\begin{doublespace}
The simple rule established earlier, $S_{T-1}=W_{T-1}/2$, no longer
applies here and we need numerical solutions at each stage. The complexity
that gets included in this scenario, when we consider the rest of
the unexecuted program into the market impact function, can be seen
from this expression. We illustrate numerical techniques for obtaining
optimal executions in section (\ref{sec:Numerical-Framework-for}).
\end{doublespace}
\begin{doublespace}

\section{\label{sec:Numerical-Framework-for}Numerical Framework for Optimal
Execution}
\end{doublespace}

\begin{doublespace}
Below we develop a numerical framework that can provide optimal executions
for any law of motion of prices. We specifically illustrate how we
can solve the formulations from section \ref{sec:Alternative-and-Practical}
with this numerical technique. It should shortly become clear how
this solution technique can be applied under any scenario of price
changes including multiple sources of uncertainty. The central idea
is similar to the American option pricing methodology (Longstaff \&
Schwartz 2001) that approximates the ex post realized payoffs from
continuation on functions of the values of the state variables. In
our case, we use least squares to approximate the conditional expectation
of the number of shares to execute as a function of the state variables
at each stage. The following points capture a high level essence of
the algorithm.
\end{doublespace}

\subsection{Optimal Execution Algorithm}
\begin{enumerate}
\begin{doublespace}
\item We create a matrix with the number of columns equal to the number
of time periods and number of rows equal to the number of different
price paths we desire (total number of simulations we are running).
The first column in the matrix corresponds to the starting price,
$P_{0}$, and the total number of shares to execute, $W_{1}$, before
the start of the first time period, $T=1$. The second column has
to hold the price, $P_{1}$, and the remaining number of shares to
execute, $W_{2}$, before the start of the second time period, $T=2$.
Each node (row and column) in the matrix has to contain the price
and the number of shares to execute before the start of the corresponding
time period.
\item The price at the start of any time period and the price innovation
sampled from a suitable distribution ($\varepsilon_{t}\sim N\left(0,\sigma_{\varepsilon}^{2}\right)$
in our case) along with the number of shares executed during that
time period incorporated into the corresponding law of motion give
us the number of shares that still remain to be executed before the
start of the next time period and the starting price point for the
next time period. Any additional sources of uncertainty can also be
included to obtain the next price level.
\item We randomly sample the remaining number of shares to be executed at
the start of the second time period and thereafter from a uniform
distribution by imposing suitable constraints. The upper limit for
the uniform distribution can be the shares remaining at the start
of the previous time period and the lower limit can be zero. During
this process, the upper and lower limits for the uniform distribution
can be changed to impose constraints on the minimum or maximum amounts
we wish to execute during the previous time period\footnote{\begin{doublespace}
In the actual numerical algorithm we implement, for simplicity, we
directly sample the number of shares to execute from a uniform distribution
with suitable constraints imposed and calculate the remaining shares
using $W_{t}=W_{t-1}-S_{t-1}$.
\end{doublespace}
}.
\item Continuing this iteratively, we obtain a matrix where each node represents
a different scenario of price and remaining number of shares to be
executed before the start of the next time period. Each column contains
many different combinations of price and remaining number of shares
at the start of the corresponding time period.
\end{doublespace}
\item Starting from the last time period, at each node, we compute the optimal
number of shares to execute during that time period and later ones
with complete knowledge of the innovations ($\varepsilon_{t}$) that
unfold on that path, using well-known optimization techniques. For
the complex impact function, we use the solnp package in R (Ghalanos,
Theussl \& Ghalanos 2012; Ye 1988); for the simple impact function,
we allocate the remaining shares to the remaining time periods based
on whether the corresponding innovations are negative and how negative
they are.
\begin{enumerate}
\item Considering the below example of obtaining the optimal executions
when we are minimizing the complex impact function under the benchmark
law of price motion, we write the objective function as, 
\begin{equation}
\underset{\left\{ S_{t}\right\} }{\min}\left[\sum_{t=1}^{t=T}\left\{ \max\left(\theta S_{t}+\varepsilon_{t},0\right)\left(\sum_{j=t}^{j=T}S_{j}\right)\right\} \right]
\end{equation}
Here, $\sum_{t=1}^{t=T}S_{t}=W_{1}$ ; $S_{t},W_{1}\geq0$ and $\theta,\varepsilon_{t}\in R$,
that is they are real numbers. Note that, $P_{t}-P_{t-1}=\theta S_{t}+\varepsilon_{t}$
for the benchmark law of price motion.
\begin{itemize}
\begin{doublespace}
\item As an example, for $T=3$, the complete objective function will be,
\begin{equation}
\underset{\left\{ S_{1},S_{2},S_{3}\right\} }{\min}\left[\max\left(\theta S_{1}+\varepsilon_{1},0\right)\left(S_{1}+S_{2}+S_{3}\right)+\max\left(\theta S_{2}+\varepsilon_{2},0\right)\left(S_{2}+S_{3}\right)+\max\left(\theta S_{3}+\varepsilon_{3},0\right)\left(S_{3}\right)\right]\label{eq:Complete-Objective-T-3}
\end{equation}
\item For the last time period, $T=3$, the optimal number of shares, $S_{3}^{*}=W_{3}$.
Since this is the last time period we execute all the shares that
are remaining at the start of the time period, $W_{3}$.
\end{doublespace}
\item When we are at time period, $T=2$, we optimize $S_{2},S_{3}$ using
the Rsolnp library such that the following function is minimized,
\begin{equation}
\underset{\left\{ S_{2},S_{3}\right\} }{\min}\left[\max\left(\theta S_{2}+\varepsilon_{2},0\right)\left(S_{2}+S_{3}\right)+\max\left(\theta S_{3}+\varepsilon_{3},0\right)\left(S_{3}\right)\right]
\end{equation}
Here, $\sum_{t=2}^{t=3}S_{t}=W_{2}$ ; $W_{2}$ would have a different
value on each price path or for each row in our matrix and we need
to perform this optimization exercise on each simulation path.
\item When we are at time period, $T=1$, we optimize $S_{1,}S_{2},S_{3}$
using the Rsolnp library such that the following function is minimized.
This is the same as our complete optimization objective in Eq: \ref{eq:Complete-Objective-T-3}.
\begin{equation}
\underset{\left\{ S_{1},S_{2},S_{3}\right\} }{\min}\left[\max\left(\theta S_{1}+\varepsilon_{1},0\right)\left(S_{1}+S_{2}+S_{3}\right)+\max\left(\theta S_{2}+\varepsilon_{2},0\right)\left(S_{2}+S_{3}\right)+\max\left(\theta S_{3}+\varepsilon_{3},0\right)\left(S_{3}\right)\right]
\end{equation}
Here, $\sum_{t=1}^{t=3}S_{t}=W_{1}$ ; $W_{1}$ is the total number
of shares we start with and it would have a different value on each
price path or for each row in our matrix. We need to perform this
optimization on each simulation path.
\item If there are $N$ simulation price paths and $T$ time periods we
need to make a total of $N*\left(T-1\right)$ calls to the Rsolnp
routine.
\end{itemize}
\item Considering the below example of obtaining the optimal executions
when we are minimizing the simple impact function under the benchmark
law of price motion, we write the objective function as, 
\begin{equation}
\underset{\left\{ S_{t}\right\} }{\min}\left[\sum_{t=1}^{t=T}\left\{ \max\left(\theta S_{t}+\varepsilon_{t},0\right)\left(S_{t}\right)\right\} \right]
\end{equation}
Here, $\sum_{t=1}^{t=T}S_{t}=W_{1}$ ; $S_{t},W_{1}\geq0$ and $\theta,\varepsilon_{t}\in R$,
that is they are real numbers. Note that, $P_{t}-P_{t-1}=\theta S_{t}+\varepsilon_{t}$
\begin{itemize}
\begin{doublespace}
\item As an example, for $T=3$ the full objective function can be written
as,
\begin{equation}
\underset{\left\{ S_{1},S_{2},S_{3}\right\} }{\min}\left[\max\left(\theta S_{1}+\varepsilon_{1},0\right)\left(S_{1}\right)+\max\left(\theta S_{2}+\varepsilon_{2},0\right)\left(S_{2}\right)+\max\left(\theta S_{3}+\varepsilon_{3},0\right)\left(S_{3}\right)\right]\label{eq:Complete-Opt-Simple-T-3}
\end{equation}
\item For the last time period, $T=3$, the optimal number of shares, $S_{3}^{*}=S_{3}$. 
\end{doublespace}
\item When we are time period, $T=2,$ we optimize $S_{2},S_{3}$ such that
the following function is minimized,
\begin{equation}
\underset{\left\{ S_{2},S_{3}\right\} }{\min}\left[\max\left(\theta S_{2}+\varepsilon_{2},0\right)\left(S_{2}\right)+\max\left(\theta S_{3}+\varepsilon_{3},0\right)\left(S_{3}\right)\right]
\end{equation}
Here, $\sum_{t=2}^{t=3}S_{t}=W_{2}$ ; $W_{2}$ would have a different
value on each price path simulation or for each row in our matrix.
The remaining shares $W_{2}$ are distributed to time periods that
have negative innovations, starting with earlier time periods, until
the execution size times the impact parameter $\theta$ plus the innovation
equals zero for a particular time period. When this condition is satisfied,
we incur zero impact $\left\{ \theta S_{t}+\varepsilon_{t}=0\Rightarrow P_{t}-P_{t-1}=0\right\} $.
After the execution size times the impact parameter $\theta$ plus
the innovation equals zero for all time periods, any further leftover
shares are allocated, while giving precedence to earlier time periods
and then allocating to subsequent time periods, up-to the maximum
execution limit for each time period, since the execution of these
shares will cause an equal jump up in the prices (having an equal
impact in the objective function) and it is better to execute sooner
rather than later.
\item When we are at time period, $T=1$, we optimize $S_{1,}S_{2},S_{3}$
such that the following function is minimized. This is the same as
our complete optimization objective in Eq: \ref{eq:Complete-Opt-Simple-T-3}.
\begin{equation}
\underset{\left\{ S_{1},S_{2},S_{3}\right\} }{\min}\left[\max\left(\theta S_{1}+\varepsilon_{1},0\right)\left(S_{1}\right)+\max\left(\theta S_{2}+\varepsilon_{2},0\right)\left(S_{2}\right)+\max\left(\theta S_{3}+\varepsilon_{3},0\right)\left(S_{3}\right)\right]
\end{equation}
Here, $\sum_{t=1}^{t=3}S_{t}=W_{1}$ ; $W_{1}$ would have a different
value on each price path simulation or for each row in our matrix.
The remaining shares $W_{1}$ are distributed similar to the methodology
described above. 
\end{itemize}
\end{enumerate}
\begin{doublespace}
\item We then run a regression across all the rows in the matrix (this is
a cross sectional regression across the simulated paths) with the
independent variables as the price, $P_{t-1}$, and the number of
shares remaining to be executed before the time period starts, $W_{t}$,
and the optimal number of shares to execute during that time period,
$S_{t}$, as the dependent variable. It is to be understood that $W_{t}$,
$P_{t-1}$, and $S_{t}$ denote the values of the variables across
all the paths (total number of simulations) in our sample. We use
a regression model such as the one below (Eq: \ref{eq:Regression-Optimal-Execs}).
It should be clear that we can extend this to purely non-linear regressions
or a combination of linear and non-linear components. Inclusion of
additional variables such as spread, volume, number of trades, etc.
are other possible extensions.
\end{doublespace}

\begin{equation}
E\left[S_{t}\left|W_{t,}P_{t-1}\right.\right]=\beta_{0}+\beta_{1}W_{t}+\beta_{2}W_{t}^{2}+\beta_{3}P_{t-1}+\beta_{4}P_{t-1}^{2}+\beta_{5}W_{t}P_{t-1}\label{eq:Regression-Optimal-Execs}
\end{equation}
For $T=2$,
\begin{equation}
E\left[S_{2}\left|W_{2,}P_{1}\right.\right]=\beta_{0}+\beta_{1}W_{2}+\beta_{2}W_{2}^{2}+\beta_{3}P_{1}+\beta_{4}P_{1}^{2}+\beta_{5}W_{2}P_{1}
\end{equation}

\item Likewise, we continue backwards in time and obtain regression coefficients
for each time period. The regression coefficients can then be used
to calculate the optimal number of shares before the start of each
time period. At each stage, we adjust the number of shares remaining
before the time period starts based on the difference between the
simulated number of shares to execute and the conditional expected
value of the number of shares to execute as given by the above regression
equation. For $T=2$, this adjusted number of remaining shares, $\hat{W}_{2}$,
is given by,
\begin{equation}
\hat{W}_{2}=S_{3}+\left(S_{2}-E\left[S_{2}\left|W_{2,}P_{1}\right.\right]\right)
\end{equation}
\end{enumerate}

\subsection{Sample Results with Mean-Variance of Execution Costs}

\begin{doublespace}
For the complex impact formulation, the table in (Figure \ref{fig:Regression-Co-efficients-Complex})
gives the regression coefficients when the number of time periods,
$T=20$, the total number of shares to execute, $W_{1}=\bar{S}=100,000$
and the initial price, $P_{0}=50$. Starting from the initial time
period in the first row, the optimal executions, the price path and
other parameters are also shown\footnote{\begin{doublespace}
The table in (Figure \ref{fig:Regression-Co-efficients-Complex})
has the following information: simulation sample size; maximum shares
available when prices move up; maximum shares available when prices
move down; $\theta$, the impact parameter; $\sigma_{\varepsilon}$,
a proxy for the stock price volatility; the average of the simple
market impact; the average of the complex market impact; the total
execution costs; referring to the regression equation, (Eq: \ref{eq:Regression-Optimal-Execs},
$\beta_{0}+\beta_{1}W_{t}+\beta_{2}W_{t}^{2}+\beta_{3}P_{t-1}+\beta_{4}P_{t-1}^{2}+\beta_{5}W_{t}P_{t-1}$),
$\beta_{0}$, the regression intercept; $\beta_{1}$, the coefficient
of the shares remaining before the start of time period $t$, $W_{t}$;
$\beta_{2}$, the coefficient of the square of the shares remaining
before the start of time period $t$, $W_{t}^{2}$; $\beta_{3}$,
the coefficient of the price before the start of time period $t$,
$P_{t-1}$; $\beta_{4}$, the coefficient of the square of the price
before the start of time period $t$, $P_{t-1}^{2}$; $\beta_{5}$,
the coefficient of the product of the remaining shares and the price
before the start of time period $t$, $W_{t}P_{t-1}$ ; the optimal
executions; and the price path represented by the columns: (simulationSampleSize;
maxUpShares; maxDownShares; thetaImpact; sigmaStockPrice; marketImpactAvgCost;
marketImpactAvgCostComplex; totalExecutionAvgCost; intercept; remainingSharesT;
priceTMinusOne; remainingSharesTSqr; priceTMinusOneSqr; remainingSharesTtimespriceTMinusOne;
optimalExecs; prices). The other tables have the same column names
representing the same information. The table in (Figure \ref{fig:Optimal Executions-Complex})
has one additional column name: Optimal Execution $i$ which is the
optimal number of shares to execute in time period $i$. The table
in (Figure \ref{fig:Complex Executions-Costs}) has additional columns
that give the total calculation time; the time to calculate the optimal
executions using the Rsolnp call; and the time to simulate the random
prices and remaining shares given by the columns: (totalTime; optimalTime;
randomTime). All time columns are measured in seconds.

The significance of the regression coefficients in Eq: \ref{eq:Regression-Optimal-Execs}
varies across different time periods. The shares remaining to be executed,
$W_{t}$, tends to be the most significant variable. It is significant
around the $p=0.01$ and $p=0.1$ levels for the last few time periods
from the end and the significance decreases thereafter (p-values increase).
This illustrates the difficulty in predicting optimal executions as
the number of remaining time periods increases. This is an artifact
of the high noise environment, such as for trading costs and optimal
executions, and highlights the issue of making accurate estimations
in such a setting. The significance of the variables increases with
an increase in the number of simulation paths. Clearly, alternate
models that include additional variables such as spread, volume, number
of trades and so on can improve the predictive power.
\end{doublespace}
}. Unless specified otherwise all the parameter values are taken to
be the same as the values in (Bertsimas \& Lo 1998), to facilitate
a proper comparison. (Figure \ref{fig:Optimal Executions-Complex})
shows the optimal execution schedules under different levels of minimum
and maximum number of shares to execute during each time period and
different number of price paths or simulation counts. 
\end{doublespace}

\textbf{\textit{(Figure \ref{fig:Mean Variance Comparison}) compares
the average and variance of the total impact costs of our numerical
methodology with the benchmark case in (Bertsimas \& Lo 1998; also
termed the naive strategy), where the solution we get is to execute
equal amount of shares in each time period}}\footnote{\begin{doublespace}
The numbers reported are using the complex impact formulation since
it gives better performance compared to the simple impact formulation.
\end{doublespace}
}\textbf{\textit{. We report the mean and variance over a simulation
sample of 50,000 price paths. We see that the benchmark case has a
mean of around 5,262,583 which is comparable to the average execution
cost of 5,264,706 using the complex formulation}}\footnote{\begin{doublespace}
5,262,583 and 5,264,706 represent the total notional value to obtain
100,000 shares.
\end{doublespace}
}\textbf{\textit{; but the variance is significantly lower using our
methodology (769,801,363 in our case versus 1,120,457,643 in the benchmark
model). (Figure \ref{fig:Mean Variance Comparison}) also reports
the multiple of ten percentile values for the executions costs. (Figure
\ref{fig:Histogram Execution Costs}) shows the histograms of the
total costs under the two techniques (the top histogram is for the
complex formulation). In addition, our methods are more realistic
and adaptive, since the execution amounts change every-time we use
it, as the market moves and as our trading progresses. Tailoring it
to include additional state variables and capture other sources of
uncertainty is relatively straightforward.}}

Lastly to provide a better understanding of how execution costs change
with changes in the different parameters, in (Figure \ref{fig:Complex Executions-Costs})
we provide the average of the total execution costs, the simple impact
costs, the complex impact costs and the market timing costs across
different parameter values, when we are optimizing the complex formulation
of the market impact. We impose non-negativity constraints on the
execution amounts while calculating the regression coefficient; later
when we use the regression coefficients to calculate execution costs,
we remove this restriction for some iterations; this parameter is
captured as the maximum and minimum number of shares we can trade
in any given time period.

The following values of the parameters are used in the computations:
we vary the volatility of the stock price $\sigma_{\varepsilon}=\left\{ 0.125,0.25,0.30,0.35,0.40,0.45,0.50,0.55,0.60,0.65,0.70,0.75\right\} $,
the impact parameter, $\theta=\left\{ 2.5,3.0,3.5,4.0,4.5,5.0,5.5,6.0,6.5,7.0,7.5\right\} $,
the maximum and minimum number of shares we can execute in any given
time period, i.e, the liquidity, $\left\{ 6666,13332,19998,26664,3333\right\} $
and $\left\{ 0,-10,000\right\} $, and the number of simulations,
$\left\{ 50000,20000,10000,3000,2000,1000\right\} $.\textbf{ }This
gives a matrix of summary statistics with around 93 different combinations
for which we calculate the regression coefficients, the simple market
impact cost, the complex market impact cost and the total execution
cost. It is immediately obvious that increasing the impact parameter,
$\theta$, leads to an increase in the total executions costs (in
Figure \ref{fig:Complex Executions-Costs}, when $\theta$ increases
in this range, $\left\{ 2.5,3.0,3.5,4.0,4.5,5.0,5.5,6.0,6.5,7.0,7.5\right\} $,
the corresponding total average costs are $\left\{ 5132240,5160852,5187428,5216411,5240020,5269296,5290421,5319175,5349371,5368020,5401751\right\} $.
 The increase in the price volatility and the liquidity in each time
period do not show such a clear pattern and further investigation
is warranted. But we can expect that the greater uncertainty due to
higher price volatility and lack of liquidity, would force participants
to trade greater amounts earlier in the trading horizon or as liquidity
becomes available.

To calculate the regression coefficients in the quickest possible
time, it would be helpful to build a decent amount of computing infrastructure.
Since each price path can be developed independently and the only
dependence across price paths is while doing the cross-sectional regressions,
all the price paths and the optimization at each time period can be
done using parallel processing technology. If there are 20 time periods
and 1000 price paths, we would need to perform 20,000 Rsolnp optimization
calls to compute all the regression coefficients for the complex impact
formulation. This is the most time consuming portion of the algorithm
and it is highly sensitive to the initial values provided for the
routine. The calculation time increases significantly with the number
of price paths and time periods; this increase is linear with the
number of price paths but it can be more costly to perform the Rsolnp
optimizations when the number of time periods increase. To make the
calculation engine more robust we also build rudimentary intelligence,
such that in case of any interruptions the system will revert back
and resume the calculations from the last clean state that was reached.
We ran our simulations on an Intel four core windows 10 machine with
4.00 Gigabytes RAM and 2.4 Gigahertz processor speed. In the summary
statistics for each formulation we provide the time it takes to calculate
all the coefficients (Figures \ref{fig:Complex Executions-Costs};
\ref{fig:Simple-Executions-Costs}; \ref{fig:One-Step-Executions}).

To reduce the number of calculations, when we are looking to run optimal
executions across hundreds of different securities, we could create
groups of securities based on similarities in starting prices, volatilities
and other parameters and compute regression coefficients for each
group separately. Optimizing the simple impact formulation takes considerably
less time. The regression coefficients, optimal executions, execution
costs and run times are summarized in (Figures \ref{fig:Regression-Co-efficients-for-Simple},
\ref{fig:Optimal Executions-Simple}, \ref{fig:Simple-Executions-Costs}).
Another alternative to optimizing the complex impact function (instead
of the Rsolnp optimization) is to perform a one step ahead optimization.
To elaborate on this, at each step, we only look at whether the price
went up or down and execute accordingly with full foresight of only
one time period. We then run the cross-sectional regressions based
on the optimal shares with just one time period look ahead. The regression
coefficients, optimal executions, execution costs and run times are
summarized in (Figures \ref{fig:Regression-Co-efficients-for-One-Step},
\ref{fig:Optimal Executions-OneStep}, \ref{fig:One-Step-Executions}).
It would be prudent to re-calibrate the regression coefficients periodically
across all three formulations.
\begin{doublespace}

\subsection{Actual Trading Costs Attribution}
\end{doublespace}

\begin{doublespace}
(Figure \ref{fig:Trading-Cost-Distributions}) illustrates the distribution
of actual trading costs (These metrics are for live institutional
trades from a global sample measured in basis points on the $y$-axis;
the $x$-axis has the costs for two months: April and May 2015; the
size of the bubble represents the trade size) based on our attribution
methodology. (Kashyap 2015, 2016) are empirical examples of applying
the above methodology to recent market events, wherein, Mincer Zarnowitz
type regressions (Mincer \& Zarnowitz 1969) are run to establish the
accuracy of the estimates. These studies demonstrate the effectiveness
of this approach in helping us better understand and analyze real
life trading situations.
\end{doublespace}
\begin{doublespace}

\section{\label{sec:Extensions-Price-Motion}Extensions to the Benchmark Law
of Price Motion}
\end{doublespace}
\begin{doublespace}

\subsection{Law of Price Motion with Additional Source of Uncertainty }
\end{doublespace}
\begin{doublespace}

\subsubsection{Simple Formulation}
\end{doublespace}

\begin{doublespace}
The law of price motion can be changed to include an additional source
of uncertainty, $X_{t}$, which could represent changing market conditions
or private information about the security. We assume that this state
variable $X_{t}$, is serially-correlated and $\gamma$ captures its
sensitivity to the price movements, which means $\gamma$ could be
positive or negative. Incorporating this, the objective function and
the Bellman equation become,
\begin{equation}
\underset{\left\{ S_{t}\right\} }{\min}\:E_{1}\left[\sum_{t=1}^{T}\left\{ \max\left[\left(P_{t}-P_{t-1}\right),0\right]S_{t}\right\} \right]
\end{equation}
\begin{equation}
\sum_{t=1}^{T}S_{t}=\bar{S}\;,S_{t}\geq0\;,W_{1}=\bar{S},\;W_{T+1}=0\;,\;W_{t}=W_{t-1}-S_{t-1}
\end{equation}
\begin{equation}
P_{t}=P_{t-1}+\theta S_{t}+\gamma X_{t}+\varepsilon_{t}\;,\theta>0\:,E\left[\varepsilon_{t}\left|S_{t},P_{t-1}\right.\right]=0
\end{equation}
\begin{equation}
X_{t}=\rho X_{t-1}+\eta_{t}\;,\;\rho\in\left(-1,1\right)\equiv\text{AR\ensuremath{\left(1\right)} Process}
\end{equation}
\begin{equation}
\varepsilon_{t}\sim N\left(0,\sigma_{\varepsilon}^{2}\right)\equiv\text{Zero Mean IID (Independent Identically Distributed) random shock or white noise}
\end{equation}
\begin{equation}
\eta_{t}\sim N\left(0,\sigma_{\eta}^{2}\right)\equiv\text{Zero Mean IID (Independent Identically Distributed) random shock or white noise}
\end{equation}
\begin{equation}
V_{t}\left(P_{t-1},X_{t-1},W_{t}\right)=\underset{\left\{ S_{t}\right\} }{\min}\:E_{t}\left[\max\left\{ \left(P_{t}-P_{t-1}\right),0\right\} S_{t}+V_{t+1}\left(P_{t},X_{t},W_{t+1}\right)\right]
\end{equation}
By starting at the end, (time $T$) we have, 
\begin{equation}
V_{T}\left(P_{T-1},X_{T-1},W_{T}\right)=\underset{\left\{ S_{T}\right\} }{\min}\:E_{T}\left[\max\left\{ \left(P_{T}-P_{T-1}\right),0\right\} S_{T}\right]
\end{equation}
Since $W_{T+1}$ is zero, we have the optimal trade size, $S_{T}^{*}=W_{T}$
and an expression for $V_{T}$ as,
\begin{equation}
V_{T}\left(P_{T-1},X_{T-1},W_{T}\right)=E_{T}\left[\max\left\{ \left(\theta W_{T}+\varepsilon_{T}+\gamma X_{T}\right),0\right\} W_{T}\right]
\end{equation}

\end{doublespace}
\begin{prop}
\begin{doublespace}
\textcolor{red}{\label{The-number-of-additional-uncertainty-simple}}\textcolor{black}{The
number of shares to be executed in each time period follows a linear
law. $S_{T-1}=W_{T-1}/2\quad\ldots\quad S_{T-K-1}=W_{T-K-1}/\left(K+2\right)$
and the corresponding value function is
\begin{eqnarray*}
V_{T-K-1}\left(P_{T-K-2},X_{T-K-2},W_{T-K-1}\right) & = & \frac{\theta}{\left(K+2\right)}W_{T-K-1}^{2}+\alpha_{K+1}W_{T-K-1}+\beta W_{T-K-1}\frac{\phi\left(\frac{\theta W_{T-K-1}+\left(K+2\right)\alpha_{K+1}}{\left(K+2\right)\beta}\right)}{\Phi\left(\frac{\theta W_{T-K-1}+\left(K+2\right)\alpha_{K+1}}{\left(K+2\right)\beta}\right)}
\end{eqnarray*}
\[
\text{Here, }\alpha_{K+1}=\gamma\rho X_{T-K-2},\;\beta=\sqrt{\gamma^{2}\sigma_{\eta}^{2}+\sigma_{\varepsilon}^{2}}
\]
}
\end{doublespace}
\end{prop}
\begin{proof}
\begin{doublespace}
\textcolor{black}{Appendix \ref{subsec:Proof-of-Proposition-additional-simple}.}
\end{doublespace}
\end{proof}
\begin{doublespace}
The simple rule established earlier, $S_{T-1}=W_{T-1}/2$, suffices
even here, with a similar reasoning that follows from the independence
of the price impact from either the prevailing price or the size of
the unexecuted order.
\end{doublespace}
\begin{doublespace}

\subsubsection{Complex Formulation}
\end{doublespace}

\begin{doublespace}
Incorporating this additional source of uncertainty into the complex
market impact formulation, the objective function and the Bellman
equation become,
\begin{equation}
\underset{\left\{ S_{t}\right\} }{\min}\:E_{1}\left[\sum_{t=1}^{T}\left\{ \max\left[\left(P_{t}-P_{t-1}\right),0\right]W_{t}\right\} \right]
\end{equation}
\begin{equation}
\sum_{t=1}^{T}S_{t}=\bar{S}\;,S_{t}\geq0\;,W_{1}=\bar{S},\;W_{T+1}=0\;,\;W_{t}=W_{t-1}-S_{t-1}
\end{equation}
\begin{equation}
P_{t}=P_{t-1}+\theta S_{t}+\gamma X_{t}+\varepsilon_{t}\;,\theta>0\:,E\left[\varepsilon_{t}\left|S_{t},P_{t-1}\right.\right]=0
\end{equation}
\begin{equation}
X_{t}=\rho X_{t-1}+\eta_{t}\;,\;\rho\in\left(-1,1\right)\equiv\text{AR\ensuremath{\left(1\right)} Process}
\end{equation}
\begin{equation}
\varepsilon_{t}\sim N\left(0,\sigma_{\varepsilon}^{2}\right)\equiv\text{Zero Mean IID (Independent Identically Distributed) random shock or white noise}
\end{equation}
\begin{equation}
\eta_{t}\sim N\left(0,\sigma_{\eta}^{2}\right)\equiv\text{Zero Mean IID (Independent Identically Distributed) random shock or white noise}
\end{equation}
\begin{equation}
V_{t}\left(P_{t-1},X_{t-1},W_{t}\right)=\underset{\left\{ S_{t}\right\} }{\min}\:E_{t}\left[\max\left\{ \left(P_{t}-P_{t-1}\right),0\right\} W_{t}+V_{t+1}\left(P_{t},X_{t},W_{t+1}\right)\right]
\end{equation}
By starting at the end, (time $T$) we have, 
\begin{equation}
V_{T}\left(P_{T-1},X_{T-1},W_{T}\right)=\underset{\left\{ S_{T}\right\} }{\min}\:E_{T}\left[\max\left\{ \left(P_{T}-P_{T-1}\right),0\right\} W_{T}\right]
\end{equation}
Since $W_{T+1}$ is zero, we have the optimal trade size, $S_{T}^{*}=W_{T}$
and an expression for $V_{T}$ as,
\begin{equation}
V_{T}\left(P_{T-1},X_{T-1},W_{T}\right)=E_{T}\left[\max\left\{ \left(\theta W_{T}+\varepsilon_{T}+\gamma X_{T}\right),0\right\} W_{T}\right]
\end{equation}

\end{doublespace}
\begin{prop}
\begin{doublespace}
\label{The-number-of-additional-uncertainty-complex}The number of
shares to be executed in each time period and the corresponding value
function are obtained by solving, 
\end{doublespace}
\end{prop}
\begin{doublespace}
\begin{eqnarray*}
\theta W_{T-1}+\beta W_{T-1}\left\{ \frac{\theta}{\beta}\left[-\frac{\left(\frac{\theta S_{T-1}+\alpha}{\beta}\right)\phi\left(\frac{\theta S_{T-1}+\alpha}{\beta}\right)}{\Phi\left(\frac{\theta S_{T-1}+\alpha}{\beta}\right)}-\left\{ \frac{\phi\left(\frac{\theta S_{T-1}+\alpha}{\beta}\right)}{\Phi\left(\frac{\theta S_{T-1}+\alpha}{\beta}\right)}\right\} ^{2}\right]\right\} \\
-2\theta\left(W_{T-1}-S_{T-1}\right)-\alpha+\beta\left\{ -\frac{\phi\left(\frac{\theta\left(W_{T-1}-S_{T-1}\right)+\alpha}{\beta}\right)}{\Phi\left(\frac{\theta\left(W_{T-1}-S_{T-1}\right)+\alpha}{\beta}\right)}\right.\\
\left.+\frac{\theta\left(W_{T-1}-S_{T-1}\right)}{\beta}\left[\frac{\left(\frac{\theta\left(W_{T-1}-S_{T-1}\right)+\alpha}{\beta}\right)\phi\left(\frac{\theta\left(W_{T-1}-S_{T-1}\right)+\alpha}{\beta}\right)}{\Phi\left(\frac{\theta\left(W_{T-1}-S_{T-1}\right)+\alpha}{\beta}\right)}+\left\{ \frac{\phi\left(\frac{\theta\left(W_{T-1}-S_{T-1}\right)+\alpha}{\beta}\right)}{\Phi\left(\frac{\theta\left(W_{T-1}-S_{T-1}\right)+\alpha}{\beta}\right)}\right\} ^{2}\right]\right\}  & = & 0
\end{eqnarray*}

\end{doublespace}
\begin{proof}
\begin{doublespace}
Appendix \ref{subsec:Proof-of-Proposition-additional-complex}.
\end{doublespace}
\end{proof}
\begin{doublespace}
The simple rule established earlier, $S_{T-1}=W_{T-1}/2$, no longer
suffices here and we need numerical solutions at each stage of the
recursion.
\end{doublespace}
\begin{doublespace}

\subsection{\label{subsec:Linear-Percentage-Law}Linear Percentage Law of Price
Motion}
\end{doublespace}
\begin{doublespace}

\subsubsection{Simple Formulation}
\end{doublespace}

\begin{doublespace}
A law of motion based on an arithmetic random walk has a positive
probability of negative prices and it also implies that the Market
Impact has a permanent effect on the prices. The other issue is that
Market Impact as a percentage of the execution price is a decreasing
function of the price level, which is counter-factual. Hence we let
the execution price be comprised of two components, a no-impact price
$\widetilde{P_{t}}$, and the price impact $\Delta_{t}$.
\begin{equation}
P_{t}=\widetilde{P_{t}}+\Delta_{t}
\end{equation}
The no impact price is the price that would prevail in the absence
of any market impact. An observable proxy for this is the mid-point
of the bid/offer spread. This is the natural price process and we
set it to be a Geometric Brownian Motion.
\begin{equation}
\widetilde{P_{t}}=\widetilde{P}_{t-1}e^{B_{t}}
\end{equation}
\begin{equation}
B_{t}\sim N\left(\mu_{B},\sigma_{B}^{2}\right)\equiv\text{IID (Independent Identically Distributed) normal random variable}
\end{equation}
The price impact $\Delta_{t}$ captures the effect of trade size on
the transaction price including the portion of the bid/offer spread.
As a percentage of the no-impact price $\widetilde{P_{t}}$, it is
a linear function of the trade size $S_{t}$ and $X_{t}$ where as
before, $X_{t}$ is a proxy for private information or market conditions.
The parameters $\theta$ and $\gamma$ measure the sensitivity of
price impact to trade size and market conditions or private information.
\begin{equation}
\Delta_{t}=\left(\theta S_{t}+\gamma X_{t}\right)\widetilde{P}_{t}
\end{equation}
\begin{equation}
X_{t}=\rho X_{t-1}+\eta_{t}\;,\;\rho\in\left(-1,1\right)
\end{equation}
\begin{equation}
\eta_{t}\sim N\left(0,\sigma_{\eta}^{2}\right)\equiv\text{Zero Mean IID (Independent Identically Distributed) random shock or white noise}
\end{equation}
The optimization problem and Bellman equation can be written as,
\begin{equation}
\underset{\left\{ S_{t}\right\} }{\min}\:E_{1}\left[\sum_{t=1}^{T}\left\{ \max\left[\left(P_{t}-P_{t-1}\right),0\right]S_{t}\right\} \right]
\end{equation}
\begin{equation}
\sum_{t=1}^{T}S_{t}=\bar{S}\;,S_{t}\geq0\;,W_{1}=\bar{S},\;W_{T+1}=0\;,\;W_{t}=W_{t-1}-S_{t-1}
\end{equation}
\begin{equation}
V_{t}\left(P_{t-1},X_{t-1},W_{t}\right)=\underset{\left\{ S_{t}\right\} }{\min}\:E_{t}\left[\max\left\{ \left(P_{t}-P_{t-1}\right),0\right\} S_{t}+V_{t+1}\left(P_{t},X_{t},W_{t+1}\right)\right]
\end{equation}
By starting at the end, (time $T$) we have, 
\begin{equation}
V_{T}\left(P_{T-1},X_{T-1},W_{T}\right)=\underset{\left\{ S_{T}\right\} }{\min}\:E_{T}\left[\max\left\{ \left(P_{T}-P_{T-1}\right),0\right\} S_{T}\right]
\end{equation}
Since $W_{T+1}$ is zero, we have the optimal trade size, $S_{T}^{*}=W_{T}$
and an expression for $V_{T}$ as,
\begin{equation}
V_{T}\left(P_{T-1},X_{T-1},W_{T}\right)=E_{T}\left[\max\left\{ \left(\widetilde{P}_{T}\left(1+\theta W_{T}+\gamma X_{T}\right)-P_{T-1}\right),0\right\} W_{T}\right]
\end{equation}

This involves a normal log-normal mixture and solutions are known
for handling this distribution under certain circumstances (Clark
1973; Tauchen \& Pitts 1983 ; Yang 2008).
\end{doublespace}
\begin{prop}
\begin{doublespace}
\label{The-value-function-linear-percentage}The value function is
of the form, $E\left[\left.Y_{2}\right|Y_{2}>0\right]$ where, 

$Y_{2}=\left(\widetilde{P}_{T-1}W_{T}e^{B_{T}}+\theta W_{T}^{2}\widetilde{P}_{T-1}e^{B_{T}}+\gamma\rho X_{T-1}\widetilde{P}_{T-1}W_{T}e^{B_{T}}+\gamma\widetilde{P}_{T-1}W_{T}e^{B_{T}}\eta_{T}-W_{T}P_{T-1}\right)$.
This can be simplified further to,

\[
E\left[\left.\left(e^{X}Y+k\right)\right|\left(e^{X}Y+k\right)>0\right]
\]

\begin{eqnarray*}
 & = & k+e^{\left(\mu_{X}+\frac{1}{2}\sigma_{X}^{2}\right)}\left[\left\{ \frac{\Phi\left(\frac{\mu_{X}+\sigma_{X}^{2}}{\sigma_{X}}\right)}{\Phi\left(\frac{\mu_{X}}{\sigma_{X}}\right)}\right\} \left\{ \mu_{Y}\left[\Phi\left(-\left[\frac{k+\mu_{Y}}{\sigma_{Y}}\right]\right)-\Phi\left(-\frac{\mu_{Y}}{\sigma_{Y}}\right)\right]-\frac{\sigma_{Y}}{\sqrt{2\pi}}\left[e^{-\frac{1}{2}\left(\frac{k+\mu_{Y}}{\sigma_{Y}}\right)^{2}}-e^{-\frac{1}{2}\left(\frac{\mu_{Y}}{\sigma_{Y}}\right)^{2}}\right]\right\} \right.\\
 &  & +\left.\left\{ \frac{1-\Phi\left(\frac{\mu_{X}+\sigma_{X}^{2}}{\sigma_{X}}\right)}{1-\Phi\left(\frac{\mu_{X}}{\sigma_{X}}\right)}\right\} \left\{ \mu_{Y}\left[1-\Phi\left(-\left[\frac{k+\mu_{Y}}{\sigma_{Y}}\right]\right)\right]+\frac{\sigma_{Y}}{\sqrt{2\pi}}\left[e^{-\frac{1}{2}\left(\frac{k+\mu_{Y}}{\sigma_{Y}}\right)^{2}}\right]\right\} \right]
\end{eqnarray*}
\[
\text{Here, }X\sim N\left(\mu_{X},\sigma_{X}^{2}\right);Y\sim N\left(\mu_{Y},\sigma_{Y}^{2}\right);X\;\text{and }Y\text{ are independent. Also, }k<0
\]
\end{doublespace}
\end{prop}
\begin{proof}
\begin{doublespace}
Appendix \ref{subsec:Proof-of-Proposition-linear-percentage}.
\end{doublespace}
\end{proof}
\begin{doublespace}
Clearly, the approach outlined in section \ref{sec:Numerical-Framework-for}
to use least squares to approximate the conditional expectation as
a function of the state variables at each stage can be easily applied.
We can also use other numerical techniques (Miranda \& Fackler 2002)
or approximations to the error function (Chiani, Dardari \& Simon
2003).
\end{doublespace}
\begin{doublespace}

\subsubsection{Complex Formulation}
\end{doublespace}

\begin{doublespace}
The optimization problem and Bellman equation for the complex case
can be written as,
\begin{equation}
\underset{\left\{ S_{t}\right\} }{\min}\:E_{1}\left[\sum_{t=1}^{T}\left\{ \max\left[\left(P_{t}-P_{t-1}\right),0\right]W_{t}\right\} \right]
\end{equation}
\begin{equation}
\sum_{t=1}^{T}S_{t}=\bar{S}\;,S_{t}\geq0\;,W_{1}=\bar{S},\;W_{T+1}=0\;,\;W_{t}=W_{t-1}-S_{t-1}
\end{equation}
\begin{equation}
V_{t}\left(P_{t-1},X_{t-1},W_{t}\right)=\underset{\left\{ S_{t}\right\} }{\min}\:E_{t}\left[\max\left\{ \left(P_{t}-P_{t-1}\right),0\right\} W_{t}+V_{t+1}\left(P_{t},X_{t},W_{t+1}\right)\right]
\end{equation}
By starting at the end, (time $T$) we have, 
\begin{equation}
V_{T}\left(P_{T-1},X_{T-1},W_{T}\right)=\underset{\left\{ S_{T}\right\} }{\min}\:E_{T}\left[\max\left\{ \left(P_{T}-P_{T-1}\right),0\right\} W_{T}\right]
\end{equation}
Since $W_{T+1}$ is zero, we have the optimal trade size, $S_{T}^{*}=W_{T}$
and an expression for $V_{T}$ can be arrived similar to the simple
formulation in Proposition \ref{The-value-function-linear-percentage}.
\end{doublespace}

\subsection{\label{subsec:Optimization-Linear-Percentage-Impact}Optimization
under the Linear Percentage Law of Price Motion}
\begin{enumerate}
\item For the numerical algorithm discussion in Section \ref{sec:Numerical-Framework-for},
the below example illustrates how to obtain the optimal executions
when we are minimizing the complex impact function under the linear
percentage benchmark law of price motion (Section \ref{subsec:Linear-Percentage-Law}).
Minimization using the simple impact function is simpler than this
and hence not considered here. We write the objective function as,
\begin{equation}
\underset{\left\{ S_{t}\right\} }{\min}\left[\sum_{t=1}^{t=T}\left\{ \max\left(P_{t}-P_{t-1},0\right)\left(\sum_{j=t}^{j=T}S_{j}\right)\right\} \right]
\end{equation}
Note that, $P_{t}=\widetilde{P}_{t}\left(1+\theta S_{t}+\gamma X_{t}\right)$
for the linear percentage law of price motion. Also, $P_{0}=\widetilde{P}_{0}$
since $S_{0}=0$ and we take $X_{0}=0$. Here, $\sum_{t=1}^{t=T}S_{t}=W_{1}$
; $S_{t},W_{1}\geq0$ and $\theta,\gamma\in R$, that is they are
real numbers. As an example, for $T=3$, the complete objective function
will be,
\begin{equation}
\underset{\left\{ S_{1},S_{2},S_{3}\right\} }{\min}\left[\max\left(P_{1}-P_{0},0\right)\left(S_{1}+S_{2}+S_{3}\right)+\max\left(P_{2}-P_{1},0\right)\left(S_{2}+S_{3}\right)+\max\left(P_{3}-P_{2},0\right)\left(S_{3}\right)\right]\label{eq:Complete-Linear-Percentage-T-3}
\end{equation}
\begin{align}
 & =\underset{\left\{ S_{1},S_{2},S_{3}\right\} }{\min}\left[\max\left(\widetilde{P}_{1}\left(1+\theta S_{1}+\gamma X_{1}\right)-P_{0},0\right)\left(S_{1}+S_{2}+S_{3}\right)\right.\\
 & +\max\left(\widetilde{P}_{2}\left(1+\theta S_{2}+\gamma X_{2}\right)-\widetilde{P}_{1}\left(1+\theta S_{1}+\gamma X_{1}\right),0\right)\left(S_{2}+S_{3}\right)\\
 & \left.+\max\left(\widetilde{P}_{3}\left(1+\theta S_{3}+\gamma X_{3}\right)-\widetilde{P}_{2}\left(1+\theta S_{2}+\gamma X_{2}\right),0\right)\left(S_{3}\right)\right]
\end{align}
\end{enumerate}
\begin{itemize}
\begin{doublespace}
\item For the last time period, $T=3$, the optimal number of shares, $S_{3}^{*}=W_{3}$.
Since this is the last time period we execute all the shares that
are remaining at the start of the time period, $W_{3}$.
\end{doublespace}
\item When we are at time period, $T=2$, we optimize $S_{2},S_{3}$ using
the Rsolnp library such that the following function is minimized,
\begin{align}
\underset{\left\{ S_{2},S_{3}\right\} }{\min}\left[\max\left(\widetilde{P}_{2}\left(1+\theta S_{2}+\gamma X_{2}\right)-P_{1},0\right)\left(S_{2}+S_{3}\right)\right.\\
\left.+\max\left(\widetilde{P}_{3}\left(1+\theta S_{3}+\gamma X_{3}\right)-\widetilde{P}_{2}\left(1+\theta S_{2}+\gamma X_{2}\right),0\right)\left(S_{3}\right)\right]
\end{align}
Here, $\sum_{t=2}^{t=3}S_{t}=W_{2}$ ; $W_{2}$ would have a different
value on each price path or for each row in our matrix and we need
to perform this optimization exercise on each simulation path.
\item When we are at time period, $T=1$, we optimize $S_{1,}S_{2},S_{3}$
using the Rsolnp library such that the following function is minimized.
This is the same as our complete optimization objective in Eq: \ref{eq:Complete-Linear-Percentage-T-3}.
\begin{equation}
\underset{\left\{ S_{1},S_{2},S_{3}\right\} }{\min}\left[\max\left(P_{1}-P_{0},0\right)\left(S_{1}+S_{2}+S_{3}\right)+\max\left(P_{2}-P_{1},0\right)\left(S_{2}+S_{3}\right)+\max\left(P_{3}-P_{2},0\right)\left(S_{3}\right)\right]
\end{equation}
Here, $\sum_{t=1}^{t=3}S_{t}=W_{1}$ ; $W_{1}$ is the total number
of shares we start with and it would have a different value on each
price path or for each row in our matrix. We need to perform this
optimization on each simulation path.
\end{itemize}
\begin{doublespace}

\subsection{\label{subsec:Introducing-Liquidity-Constraint}Including Liquidity
Constraints}
\end{doublespace}
\begin{doublespace}

\subsubsection{Simple Formulation }
\end{doublespace}

\begin{doublespace}
A practical limitation that arises when trading is the extent of liquidity
that is available at any point in time. This becomes a restriction
on the amount of shares tradable in any given interval. Volume can
be observed and estimated with a reasonable degree of accuracy. Hence,
any measure linking volume to trading costs would be a very practical
device. There is a voluminous literature that derives theoretical
models and looks at the empirical relationship between volume and
prices (Karpoff 1986; 1987; Gallant, Rossi \& Tauchen 1992; Campbell,
Grossman \& Wang 1993; Wang 1994). We fit a specification similar
to the one in (Campbell, Grossman \& Wang 1993) wherein the price
movements can arise due to changes in future cash flows and investor
preferences or the risk aversion. The intuition for this would be
that a low return due to a price drop could be caused by an increase
in the risk aversion or bad news about future cash flows. Changes
in risk aversion cause trading volume to increase while news that
is public will already have been impounded in the price and hence
will not cause additional trading. Low returns followed by high volume
are due to increased risk aversion while low returns and low volume
are due to public knowledge of a low level of expectation of future
returns. As risk aversion increases, the group of investor still willing
to hold the stock require a greater return leading to higher future
expected returns. Bad news about future cash flows leads to lower
expected returns. This is captured as an inverse relation between
auto-correlation of returns and trading volume. The simplification
we employ combines the two sources of price changes into one, since
what can be observed is only the price return. We note that this can
be viewed as an extension of the law of price motion with an additional
source of uncertainty. Here, $O_{t}$ is the total volume traded (market
volume) in the interval $t$. The coefficient $\alpha$ can be positive
or negative, $\gamma$ is positive and $\theta$ continues to be positive.
\begin{equation}
\underset{\left\{ S_{t}\right\} }{\min}\:E_{1}\left[\sum_{t=1}^{T}\left\{ \max\left[\left(P_{t}-P_{t-1}\right),0\right]S_{t}\right\} \right]
\end{equation}
\begin{equation}
\sum_{t=1}^{T}S_{t}=\bar{S}\;,S_{t}\geq0\;,W_{1}=\bar{S},\;W_{T+1}=0\;,\;W_{t}=W_{t-1}-S_{t-1}
\end{equation}
\begin{equation}
P_{t}=\left(\alpha+1\right)P_{t-1}+\theta S_{t}P_{t-1}-\gamma\left(O_{t}-S_{t}\right)P_{t-1}+\varepsilon_{t},\;O_{t}\geq S_{t},\;\beta,\theta>0\;,\alpha\in\left(-\infty,\infty\right)\:,E\left[\varepsilon_{t}\left|S_{t},P_{t-1}\right.\right]=0
\end{equation}
\begin{equation}
O_{t}=\rho O_{t-1}+\eta_{t}\;,\;\rho\in\left(-1,1\right)\equiv\text{AR\ensuremath{\left(1\right)} Process}
\end{equation}
\begin{equation}
\varepsilon_{t}\sim N\left(0,\sigma_{\varepsilon}^{2}\right)\equiv\text{Zero Mean IID (Independent Identically Distributed) random shock or white noise}
\end{equation}
\begin{equation}
\eta_{t}\sim N\left(0,\sigma_{\eta}^{2}\right)\equiv\text{Zero Mean IID (Independent Identically Distributed) random shock or white noise}
\end{equation}
\begin{equation}
V_{t}\left(P_{t-1},O_{t-1},W_{t}\right)=\underset{\left\{ S_{t}\right\} }{\min}\:E_{t}\left[\max\left\{ \left(P_{t}-P_{t-1}\right),0\right\} S_{t}+V_{t+1}\left(P_{t},O_{t},W_{t+1}\right)\right]
\end{equation}
By starting at the end, (time $T$) we have,
\begin{equation}
V_{T}\left(P_{T-1},O_{T-1},W_{T}\right)=\underset{\left\{ S_{T}\right\} }{\min}\:E_{T}\left[\max\left\{ \left(P_{T}-P_{T-1}\right),0\right\} S_{T}\right]
\end{equation}
Since $W_{T+1}$ is zero, we have the optimal trade size, $S_{T}^{*}=W_{T}$
and an expression for $V_{T}$ as,
\end{doublespace}
\begin{prop}
\begin{doublespace}
\label{The-value-function-liquidity-constraints}The value functions
are of the form, $E\left[\left.Y\right|Y>0\right]$ where, 

$Y=\left(\alpha P_{T-1}W_{T}+\beta W_{T}^{2}P_{T-1}-\gamma\rho O_{T-1}W_{T}P_{T-1}-\gamma W_{T}P_{T-1}\eta_{T}+W_{T}\varepsilon_{T}\right)$.
For the last and last but one time periods, these can be simplified
further to,
\[
V_{T}\left(P_{T-1},O_{T-1},W_{T}\right)=\left(\sqrt{\gamma^{2}P_{T-1}^{2}\sigma_{\eta}^{2}+\sigma_{\varepsilon}^{2}}\;\right)W_{T}\psi\left(\xi W_{T}\right)\;,\;\xi W_{T}=\left(\frac{\alpha P_{T-1}+\beta W_{T}P_{T-1}-\gamma\rho O_{T-1}P_{T-1}}{\sqrt{\gamma^{2}P_{T-1}^{2}\sigma_{\eta}^{2}+\sigma_{\varepsilon}^{2}}}\right)
\]

and
\begin{eqnarray*}
V_{T-1}\left(P_{T-2},O_{T-2},W_{T-1}\right)= & \underset{\left\{ S_{T-1}\right\} }{\min} & \:E_{T-1}\left\{ \vphantom{\left(\frac{\frac{\alpha P_{T-2}}{\sqrt{\sigma_{\varepsilon}^{2}}}}{\frac{\alpha P_{T-2}}{\sqrt{\sigma_{\varepsilon}^{2}}}}\right)}S_{T-1}\left(\sqrt{\gamma^{2}P_{T-2}^{2}\sigma_{\eta}^{2}+\sigma_{\varepsilon}^{2}}\;\right)\right.
\end{eqnarray*}
\[
\left[\left(\frac{\alpha P_{T-2}+\beta S_{T-1}P_{T-2}-\gamma\rho O_{T-2}P_{T-2}}{\sqrt{\gamma^{2}P_{T-2}^{2}\sigma_{\eta}^{2}+\sigma_{\varepsilon}^{2}}}\right)+\frac{\phi\left(\frac{\alpha P_{T-2}+\beta S_{T-1}P_{T-2}-\gamma\rho O_{T-2}P_{T-2}}{\sqrt{\gamma^{2}P_{T-2}^{2}\sigma_{\eta}^{2}+\sigma_{\varepsilon}^{2}}}\right)}{\Phi\left(\frac{\alpha P_{T-2}+\beta S_{T-1}P_{T-2}-\gamma\rho O_{T-2}P_{T-2}}{\sqrt{\gamma^{2}P_{T-2}^{2}\sigma_{\eta}^{2}+\sigma_{\varepsilon}^{2}}}\right)}\right]
\]
\[
+\left(W_{T-1}-S_{T-1}\right)\left(\sqrt{\gamma^{2}\left\{ P_{T-2}\left(\alpha+1+\beta S_{T-1}-\gamma\rho O_{T-2}-\gamma\eta_{T-1}\right)+\varepsilon_{T-1}\right\} ^{2}\sigma_{\eta}^{2}+\sigma_{\varepsilon}^{2}}\;\right)
\]
\[
\left[\left(\frac{\left\{ P_{T-2}\left(\alpha+1+\beta S_{T-1}-\gamma\rho O_{T-2}-\gamma\eta_{T-1}\right)+\varepsilon_{T-1}\right\} \left\{ \alpha+\beta\left(W_{T-1}-S_{T-1}\right)-\gamma\rho^{2}O_{T-2}-\gamma\rho\eta_{T-1}\right\} }{\sqrt{\gamma^{2}\left\{ P_{T-2}\left(\alpha+1+\beta S_{T-1}-\gamma\rho O_{T-2}-\gamma\eta_{T-1}\right)+\varepsilon_{T-1}\right\} ^{2}\sigma_{\eta}^{2}+\sigma_{\varepsilon}^{2}}}\right)\right.
\]
\[
\left.+\left.\frac{\phi\left(\frac{\left\{ P_{T-2}\left(\alpha+1+\beta S_{T-1}-\gamma\rho O_{T-2}-\gamma\eta_{T-1}\right)+\varepsilon_{T-1}\right\} \left\{ \alpha+\beta\left(W_{T-1}-S_{T-1}\right)-\gamma\rho^{2}O_{T-2}-\gamma\rho\eta_{T-1}\right\} }{\sqrt{\gamma^{2}\left\{ P_{T-2}\left(\alpha+1+\beta S_{T-1}-\gamma\rho O_{T-2}-\gamma\eta_{T-1}\right)+\varepsilon_{T-1}\right\} ^{2}\sigma_{\eta}^{2}+\sigma_{\varepsilon}^{2}}}\right)}{\Phi\left(\frac{\left\{ P_{T-2}\left(\alpha+1+\beta S_{T-1}-\gamma\rho O_{T-2}-\gamma\eta_{T-1}\right)+\varepsilon_{T-1}\right\} \left\{ \alpha+\beta\left(W_{T-1}-S_{T-1}\right)-\gamma\rho^{2}O_{T-2}-\gamma\rho\eta_{T-1}\right\} }{\sqrt{\gamma^{2}\left\{ P_{T-2}\left(\alpha+1+\beta S_{T-1}-\gamma\rho O_{T-2}-\gamma\eta_{T-1}\right)+\varepsilon_{T-1}\right\} ^{2}\sigma_{\eta}^{2}+\sigma_{\varepsilon}^{2}}}\right)}\right]\right\} 
\]
Here, $\beta=\theta+\gamma$,
\end{doublespace}
\end{prop}
\begin{proof}
\begin{doublespace}
Appendix \ref{subsec:Proof-of-Proposition-liquidity-constraints}.
\end{doublespace}
\end{proof}
\begin{doublespace}
This requires numerical solutions at each stage of the recursion.
A point worth noting is that the simple rule from the earlier linear
cases, where the price impact is independent of both the prevailing
price and the size of the unexecuted order, no longer applies here.
The necessity of having to work with complicated expressions of the
sort above, highlights to us the inherent difficulty of making predictions
in a complex social system and also that our approach to estimating
Market Impact provides a realistic platform upon which further complications,
such as working with joint distributions of volume and price, can
be built. A key takeaway from this result is that volume can have
counter intuitive effects on the trading costs.
\end{doublespace}
\begin{doublespace}

\subsubsection{Complex Formulation}
\end{doublespace}

\begin{doublespace}
The optimization problem and Bellman equation for the complex case
can be written as,
\begin{equation}
\underset{\left\{ S_{t}\right\} }{\min}\:E_{1}\left[\sum_{t=1}^{T}\left\{ \max\left[\left(P_{t}-P_{t-1}\right),0\right]W_{t}\right\} \right]
\end{equation}
\begin{equation}
\sum_{t=1}^{T}S_{t}=\bar{S}\;,S_{t}\geq0\;,W_{1}=\bar{S},\;W_{T+1}=0\;,\;W_{t}=W_{t-1}-S_{t-1}
\end{equation}
\begin{equation}
P_{t}=\left(\alpha+1\right)P_{t-1}+\theta S_{t}P_{t-1}-\gamma\left(O_{t}-S_{t}\right)P_{t-1}+\varepsilon_{t},\;O_{t}\geq S_{t},\;\beta,\theta>0\;,\alpha\in\left(-\infty,\infty\right)\:,E\left[\varepsilon_{t}\left|S_{t},P_{t-1}\right.\right]=0
\end{equation}
\begin{equation}
O_{t}=\rho O_{t-1}+\eta_{t}\;,\;\rho\in\left(-1,1\right)\equiv\text{AR\ensuremath{\left(1\right)} Process}
\end{equation}
\begin{equation}
\varepsilon_{t}\sim N\left(0,\sigma_{\varepsilon}^{2}\right)\equiv\text{Zero Mean IID (Independent Identically Distributed) random shock or white noise}
\end{equation}
\begin{equation}
\eta_{t}\sim N\left(0,\sigma_{\eta}^{2}\right)\equiv\text{Zero Mean IID (Independent Identically Distributed) random shock or white noise}
\end{equation}
\begin{equation}
V_{t}\left(P_{t-1},O_{t-1},W_{t}\right)=\underset{\left\{ S_{t}\right\} }{\min}\:E_{t}\left[\max\left\{ \left(P_{t}-P_{t-1}\right),0\right\} S_{t}+V_{t+1}\left(P_{t},O_{t},W_{t+1}\right)\right]
\end{equation}
By starting at the end, (time $T$) we have,
\begin{equation}
V_{T}\left(P_{T-1},O_{T-1},W_{T}\right)=\underset{\left\{ S_{T}\right\} }{\min}\:E_{T}\left[\max\left\{ \left(P_{T}-P_{T-1}\right),0\right\} W_{T}\right]
\end{equation}
 Since $W_{T+1}$ is zero, we have the optimal trade size, $S_{T}^{*}=W_{T}$
and an expression for $V_{T}$ can be arrived similar to the simple
formulation in Proposition \ref{The-value-function-liquidity-constraints}.
\end{doublespace}
\begin{doublespace}

\subsection{Trading Costs and Price Spread Sandwich }
\end{doublespace}
\begin{doublespace}

\subsubsection{Simple Formulation\label{subsec:Simple-Formulation}}
\end{doublespace}

\begin{doublespace}
Another useful tool from a trading perspective would be a measure
that connects trading costs to the spread, which can be observed.
(Roll 1984; Stoll 1989) connect the stock price changes to the bid-offer
spread. The spread is determined due to order processing costs, adverse
information or inventory holdings costs. The covariance of price changes
are related to the covariance of the changes in spread and proportional
to the square of the spread, assuming constant spread. A modification
with time varying spread can be easily accommodated in the specifications
above with an additional source of uncertainty or the linear percentage
law of motion. Here, $Q_{t}$ is the spread at any point in time.
\begin{equation}
P_{t}=P_{t-1}+\theta S_{t}+\gamma Q_{t}+\varepsilon_{t}\;,\theta>0\:,E\left[\varepsilon_{t}\left|S_{t},P_{t-1}\right.\right]=0
\end{equation}
\begin{equation}
Q_{t}=\rho Q_{t-1}+\eta_{t}\;,\;\rho\in\left(-1,1\right)\equiv\text{AR\ensuremath{\left(1\right)} Process}
\end{equation}
\begin{equation}
\varepsilon_{t}\sim N\left(0,\sigma_{\varepsilon}^{2}\right)\equiv\text{Zero Mean IID (Independent Identically Distributed) random shock or white noise}
\end{equation}
\begin{equation}
\eta_{t}\sim N\left(0,\sigma_{\eta}^{2}\right)\equiv\text{Zero Mean IID (Independent Identically Distributed) random shock or white noise}
\end{equation}

\end{doublespace}
\begin{doublespace}

\subsubsection{Complex Formulation}
\end{doublespace}

\begin{doublespace}
This would be analogous to the case in \ref{subsec:Simple-Formulation}.
\end{doublespace}
\begin{doublespace}

\section{Conclusions and Possibilities for Future Research }
\end{doublespace}

\begin{doublespace}
We have developed a trading cost model using dynamic programming that
splits the overall price move into the market impact and market timing
components. Our models seek to separate the effect of one's actions
in a financial market trading scenario from the actions of other players.
The separation of total trading costs into the two components, one
of which is directly related to the actions of a participant holds
numerous lessons for dealing with complex systems, especially in the
social sciences, wherein reducing the complexity by splitting the
many sources of uncertainty can lead to better insights in the decision
process. Such a separation brings about a reduction in complexity
since participants can better plan their course of action based on
a greater focus on the set of controls and states that hold greater
significance for their destiny while being cognizant of the implications
of the actions of the rest of the participants.
\end{doublespace}

The above decomposition allows us deduce the zero sum game nature
of trading costs. In addition, we have developed a powerful numerical
technique that can be used under any law of motion of prices including
laws of motion with multiple sources of uncertainty. The starting
values we provide for the Rsolnp optimization call can reduce the
number of iterations it requires to find the optimal values. Hence,
a good extension to this work can be to find better starting values
for the optimal value at each time period, based on the innovations
and the other parameters.

\begin{doublespace}
To ensure that our model can be used under different situations, we
build upon the benchmark case and introduce more complex formulations
of the law of price motion, including a scenario that has multiple
sources of uncertainty and consider liquidity or volume constraints.
Relating trading costs to the spread is also easily accomplished.
Key improvements to the model and methodology would stem from adding
cases where volume and prices are not assumed to be independent. Distributions
of prices that are not normal and factor the downward skew in prices
might also provide more realistic estimates. Our model takes the prices
process as exogenous, interesting continuations can extend the separation
of impact and timing to models of the limit order book that endogenously
consider the evolution of prices. Again we stress the better insights
and understanding that results from using simpler models, but the
particulars of the securities being considered might prompt experimenting
with some of the more esoteric extensions. We have looked at only
discrete time formulations, extensions to continuous time might show
interesting theoretical behavior.

A practical way to use these models, would need to factor in the market
timing over the duration of trading. Hence, we would need to first
get an estimate of the market impact for different time intervals
and also calculate the corresponding market timings. Rather than have
a single number for the market timing for each impact estimate, it
would be more useful to have an upper bound and lower bound or the
maximum possible range of the market timing, for each particular time
duration. It can be shown that the market timing depends on the price
volatility and hence is a key time sensitive variable. Traders can
then make a decision regarding which combination of market impact
and timing they prefer, since they have more control over the market
impact, which is their vessel to navigate the turbulent seas, which
is the market timing.
\end{doublespace}
\begin{doublespace}

\section{Acknowledgements and End-notes}
\end{doublespace}
\begin{enumerate}
\begin{doublespace}
\item Numerous seminar participants, particularly at a few meetings of the
econometric society and various finance organizations, provided suggestions
to improve the paper and also discussed scenarios where the techniques
developed in the paper could be useful for electronic trading desks.
Dr. Yong Wang, Dr. Isabel Yan, Dr. Vikas Kakkar, Dr. Fred Kwan, Dr.
Costel Daniel Andonie, Dr. Guangwu Liu, Dr. Jeff Hong, Dr. Humphrey
Tung and Dr. Xu Han at the City University of Hong Kong provided encouragement
to explore and where possible apply cross disciplinary techniques.
The author would like to express his gratitude to Brad Hunt, Henry
Yegerman, Samuel Zou and Alex Gillula at Markit for many inputs during
the creation of this work. 
\end{doublespace}
\item The simulation data and related software can be made available upon
request. The views and opinions expressed in this article, along with
any mistakes, are mine alone and do not necessarily reflect the official
policy or position of either of my affiliations or any other agency.
\end{enumerate}
\begin{doublespace}

\section{References }
\end{doublespace}
\begin{enumerate}
\begin{doublespace}
\item Alfonsi, A., Fruth, A., \& Schied, A. (2010). Optimal execution strategies
in limit order books with general shape functions. Quantitative Finance,
10(2), 143-157.
\item Almgren, R., \& Chriss, N. (2001). Optimal execution of portfolio
transactions. Journal of Risk, 3, 5-40.
\item Almgren, R. F. (2003). Optimal execution with nonlinear impact functions
and trading-enhanced risk. Applied mathematical finance, 10(1), 1-18.
\item Almgren, R., Thum, C., Hauptmann, E., \& Li, H. (2005). Direct estimation
of equity market impact. Risk, 18, 5752.
\item Almgren, R., \& Lorenz, J. (2007). Adaptive arrival price. Trading,
2007(1), 59-66.
\item Avellaneda, M., \& Stoikov, S. (2008). High-frequency trading in a
limit order book. Quantitative Finance, 8(3), 217-224.
\item Bell, R. M., \& Cover, T. M. (1980). Competitive optimality of logarithmic
investment. Mathematics of Operations Research, 5(2), 161-166.
\item Bertsimas, D., \& Lo, A. W. (1998). Optimal control of execution costs.
Journal of Financial Markets, 1(1), 1-50.
\item Biais, B., Hillion, P., \& Spatt, C. (1995). An empirical analysis
of the limit order book and the order flow in the Paris Bourse. the
Journal of Finance, 50(5), 1655-1689.
\item Bodie, Z., \& Taggart, R. A. (1978). Future investment opportunities
and the value of the call provision on a bond. The Journal of Finance,
33(4), 1187-1200.
\item Bouchaud, J. P., Gefen, Y., Potters, M., \& Wyart, M. (2004). Fluctuations
and response in financial markets: the subtle nature of \textquoteleft random\textquoteright{}
price changes. Quantitative Finance, 4(2), 176-190.
\item Brown, G. W. (1951). Iterative solution of games by fictitious play.
Activity analysis of production and allocation, 13(1), 374-376.
\item Brunnermeier, M. K., \& Pedersen, L. H. (2005). Predatory trading.
The Journal of Finance, 60(4), 1825-1863.
\item Campbell, J. Y., Grossman, S. J., \& Wang, J. (1993). Trading Volume
and Serial Correlation in Stock Returns. The Quarterly Journal of
Economics, 108(4), 905-39.
\item Carlin, B. I., Lobo, M. S., \& Viswanathan, S. (2007). Episodic liquidity
crises: Cooperative and predatory trading. The Journal of Finance,
62(5), 2235-2274.
\item Chiani, M., Dardari, D., \& Simon, M. K. (2003). New exponential bounds
and approximations for the computation of error probability in fading
channels. Wireless Communications, IEEE Transactions on, 2(4), 840-845.
\item Chirinko, R. S., \& Wilson, D. J. (2008). State investment tax incentives:
A zero-sum game?. Journal of Public Economics, 92(12), 2362-2384.
\item Clark, P. K. (1973). A subordinated stochastic process model with
finite variance for speculative prices. Econometrica: journal of the
Econometric Society, 135-155.
\item Collins, B. M., \& Fabozzi, F. J. (1991). A methodology for measuring
transaction costs. Financial Analysts Journal, 47(2), 27-36.
\item Cont, R., Stoikov, S., \& Talreja, R. (2010). A stochastic model for
order book dynamics. Operations research, 58(3), 549-563.
\item Cont, R., Kukanov, A., \& Stoikov, S. (2014). The price impact of
order book events. Journal of financial econometrics, 12(1), 47-88.
\item Cont, R., \& Kukanov, A. (2017). Optimal order placement in limit
order markets. Quantitative Finance, 17(1), 21-39.
\item Crawford, V. P. (1974). Learning the optimal strategy in a zero-sum
game. Econometrica: Journal of the Econometric Society, 885-891.
\item Curato, G., Gatheral, J., \& Lillo, F. (2017). Optimal execution with
non-linear transient market impact. Quantitative Finance, 17(1), 41-54.
\item Du, B., Zhu, H., \& Zhao, J. (2016). Optimal execution in high-frequency
trading with Bayesian learning. Physica A: Statistical Mechanics and
its Applications.
\item Fama, E. F. (1970). Efficient capital markets: A review of theory
and empirical work. The journal of Finance, 25(2), 383-417.
\item Foster, F. D., \& Viswanathan, S. (1990). A theory of the interday
variations in volume, variance, and trading costs in securities markets.
The Review of Financial Studies, 3(4), 593-624.
\item Forsyth, P. A., Kennedy, J. S., Tse, S. T., \& Windcliff, H. (2012).
Optimal trade execution: a mean quadratic variation approach. Journal
of Economic Dynamics and Control, 36(12), 1971-1991.
\item Fruth, A., Schöneborn, T., \& Urusov, M. (2014). Optimal Trade Execution
and Price Manipulation in Order Books with Time-Varying Liquidity.
Mathematical Finance, 24(4), 651-695.
\item Gabaix, X., Gopikrishnan, P., Plerou, V., \& Stanley, H. E. (2006).
Institutional investors and stock market volatility. The Quarterly
Journal of Economics, 121(2), 461-504.
\item Gale, D., Kuhn, H. W., \& Tucker, A. W. (1951). Linear programming
and the theory of games. Activity analysis of production and allocation,
13, 317-335.
\item Gallant, A. R., Rossi, P. E., \& Tauchen, G. (1992). Stock prices
and volume. Review of Financial studies, 5(2), 199-242.
\item Gatheral, J., \& Schied, A. (2011). Optimal trade execution under
geometric Brownian motion in the Almgren and Chriss framework. International
Journal of Theoretical and Applied Finance, 14(03), 353-368.
\end{doublespace}
\item Ghalanos, A., Theussl, S., \& Ghalanos, M. A. (2012). Package \textquoteleft Rsolnp\textquoteright .
\begin{doublespace}
\item Gladwell, M. (2013). David and Goliath. Underdogs, Misfits, and the
Art of Battling Giants (New York: Little, Brown and Company).
\item Gujarati, D. N. (1995). Basic econometrics, 3rd. International Edition.
\item Guo, X., de Larrard, A., \& Ruan, Z. (2017). Optimal placement in
a limit order book: an analytical approach. Mathematics and Financial
Economics, 11(2), 189-213.
\item Guo, X., \& Zervos, M. (2015). Optimal execution with multiplicative
price impact. SIAM Journal on Financial Mathematics, 6(1), 281-306.
\item Hamadène, S. (2006). Mixed zero-sum stochastic differential game and
American game options. SIAM Journal on Control and Optimization, 45(2),
496-518.
\item Hamilton, J. D. (1994). Time series analysis (Vol. 2). Princeton university
press.
\item Hendershott, T., Jones, C. M., \& Menkveld, A. J. (2011). Does algorithmic
trading improve liquidity?. The Journal of Finance, 66(1), 1-33.
\end{doublespace}
\item Hill, C. W. (1990). Cooperation, opportunism, and the invisible hand:
Implications for transaction cost theory. Academy of management review,
15(3), 500-513.
\begin{doublespace}
\item Hill, J. M. (2006). Alpha as a net zero-sum game. The Journal of Portfolio
Management, 32(4), 24-32.
\item Ho, T., \& Stoll, H. R. (1980). On dealer markets under competition.
The Journal of Finance, 35(2), 259-267.
\item Ho, T., \& Stoll, H. R. (1981). Optimal dealer pricing under transactions
and return uncertainty. Journal of Financial economics, 9(1), 47-73.
\item Hopman, C. (2007). Do supply and demand drive stock prices?. Quantitative
Finance, 7(1), 37-53.
\item Huberman, G., \& Stanzl, W. (2004). Price Manipulation and Quasi-Arbitrage.
Econometrica, 72(4), 1247-1275.
\item Huberman, G., \& Stanzl, W. (2005). Optimal liquidity trading. Review
of Finance, 9(2), 165-200.
\item Jain, P. K. (2005). Financial market design and the equity premium:
Electronic versus floor trading. The Journal of Finance, 60(6), 2955-2985.
\item Karpoff, J. M. (1986). A theory of trading volume. The Journal of
Finance, 41(5), 1069-1087.
\item Karpoff, J. M. (1987). The relation between price changes and trading
volume: A survey. Journal of Financial and quantitative Analysis,
22(01), 109-126.
\item Kashyap, R. (2014). Dynamic Multi-Factor Bid\textendash Offer Adjustment
Model. The Journal of Trading, 9(3), 42-55.
\item Kashyap, R. (2015). A Tale of Two Consequences. The Journal of Trading,
10(4), 51-95.
\item Kashyap, R. (2016). Hong Kong - Shanghai Connect / Hong Kong - Beijing
Disconnect (?), Scaling the Great Wall of Chinese Securities Trading
Costs. The Journal of Trading, 11(3), 81-134.
\item Kato, T. (2014). An optimal execution problem with market impact.
Finance and Stochastics, 18(3), 695-732.
\item Klemperer, P. (2004). Auctions: theory and practice.
\item Kissell, R. (2006). The expanded implementation shortfall: Understanding
transaction cost components. The Journal of Trading, 1(3), 6-16.
\item Kissell, R., \& Malamut, R. (2006). Algorithmic decision-making framework.
The Journal of Trading, 1(1), 12-21.
\item Kyle, A. S. (1985). Continuous auctions and insider trading. Econometrica:
Journal of the Econometric Society, 1315-1335.
\item Laraki, R., \& Solan, E. (2005). The value of zero-sum stopping games
in continuous time. SIAM Journal on Control and Optimization, 43(5),
1913-1922.
\item Lillo, F., Farmer, J.D. \& Mantegna, R.N. (2003). Econophysics: Master
curve for price-impact function. Nature, 421(6919), 129\textendash 130.
\item Longstaff, F. A., \& Schwartz, E. S. (2001). Valuing American options
by simulation: a simple least-squares approach. Review of Financial
studies, 14(1), 113-147.
\item Mincer, J. A., \& Zarnowitz, V. (1969). The evaluation of economic
forecasts. In Economic Forecasts and Expectations: Analysis of Forecasting
Behavior and Performance (pp. 3-46). NBER.
\item Miranda, M. J., \& Fackler, P. L. (2002). Applied Computational Economics
and Finance.
\item Nemirovski, A., Juditsky, A., Lan, G. \& Shapiro, A. (2009). Robust
stochastic approximation approach to stochastic programming. SIAM
Journal of Optimization, 19(4), 1574-1609.
\item Obizhaeva, A. A., \& Wang, J. (2013). Optimal trading strategy and
supply/demand dynamics. Journal of Financial Markets, 16(1), 1-32.
\item Perold, A. F. (1988). The implementation shortfall: Paper versus reality.
The Journal of Portfolio Management, 14(3), 4-9.
\item Potters, M., \& Bouchaud, J. P. (2003). More statistical properties
of order books and price impact. Physica A: Statistical Mechanics
and its Applications, 324(1), 133-140.
\item Predoiu, S., Shaikhet, G., \& Shreve, S. (2011). Optimal execution
in a general one-sided limit-order book. SIAM Journal on Financial
Mathematics, 2(1), 183-212.
\item Rapoport, A. (1973). Two-person game theory. Courier Corporation.
\item Robbins, H., \& Monro, S. (1951). A Stochastic Approximation Method.
The Annals of Mathematical Statistics, 22(3), 400-407. 
\item Roll, R. (1984). A simple implicit measure of the effective bid-ask
spread in an efficient market. The Journal of Finance, 39(4), 1127-1139.
\item Schied, A., \& Schöneborn, T. (2009). Risk aversion and the dynamics
of optimal liquidation strategies in illiquid markets. Finance and
Stochastics, 13(2), 181-204.
\item Schied, A., Schöneborn, T., \& Tehranchi, M. (2010). Optimal basket
liquidation for CARA investors is deterministic. Applied Mathematical
Finance, 17(6), 471-489.
\item Schied, A. (2013). Robust strategies for optimal order execution in
the Almgren\textendash Chriss framework. Applied Mathematical Finance,
20(3), 264-286.
\item Shreve, S. E. (1988). An introduction to singular stochastic control.
In Stochastic differential systems, stochastic control theory and
applications (pp. 513-528). Springer, New York, NY.
\end{doublespace}
\item Stokey, N. L., Lucas, R. E., \& Prescott, E. C. (1989). Recursive
methods in dynamic economics. Cambridge, MA: Harvard University.
\begin{doublespace}
\item Stoll, H. R. (1989). Inferring the components of the bid-ask spread:
theory and empirical tests. The Journal of Finance, 44(1), 115-134.
\end{doublespace}
\item Swami, B. (1983). Bhagavad-Gita as it is. The Bhaktivedanta Book Trust.
Mumbai, India.
\begin{doublespace}
\item Tauchen, G. E., \& Pitts, M. (1983). The Price Variability-Volume
Relationship on Speculative Markets. Econometrica, 51(2), 485-505.
\item Treynor, J. L. (1981). What does it take to win the trading game?.
Financial Analysts Journal, 37(1), 55-60.
\item Treynor, J. L. (1994). The invisible costs of trading. The Journal
of Portfolio Management, 21(1), 71-78.
\item Turnbull, S. M. (1987). Swaps: a zero sum game?. Financial Management,
16(1), 15-21.
\item Venkataraman, K. (2001). Automated versus floor trading: An analysis
of execution costs on the Paris and New York exchanges. The Journal
of Finance, 56(4), 1445-1485.
\item Von Neumann, J., \& Morgenstern, O. (1953). Theory of games and economic
behavior. Princeton university press.
\item Von Neumann, J. (1954). A numerical method to determine optimum strategy.
Naval Research Logistics Quarterly, 1(2), 109-115.
\item Wang, J. (1994). A model of competitive stock trading volume. Journal
of political Economy, 127-168.
\item Weber, P., \& Rosenow, B. (2005). Order book approach to price impact.
Quantitative Finance, 5(4), 357-364.
\item Yang, M. (2008). Normal log-normal mixture, leptokurtosis and skewness.
Applied Economics Letters, 15(9), 737-742.
\end{doublespace}
\item Ye, Y. (1988). Interior algorithms for linear, quadratic, and linearly
constrained convex programming.
\begin{doublespace}
\item Yegerman, H. \& Gillula, A. (2014). The Use and Abuse of Implementation
Shortfall. Markit Working Paper.
\end{doublespace}
\end{enumerate}
\pagebreak{}

\part*{{\LARGE{}Supplementary Online Material }}

\section{\label{sec:Appendix-of-Figures}Appendix of Figures}

\begin{doublespace}
\begin{figure}[H]
\includegraphics[width=18.5cm]{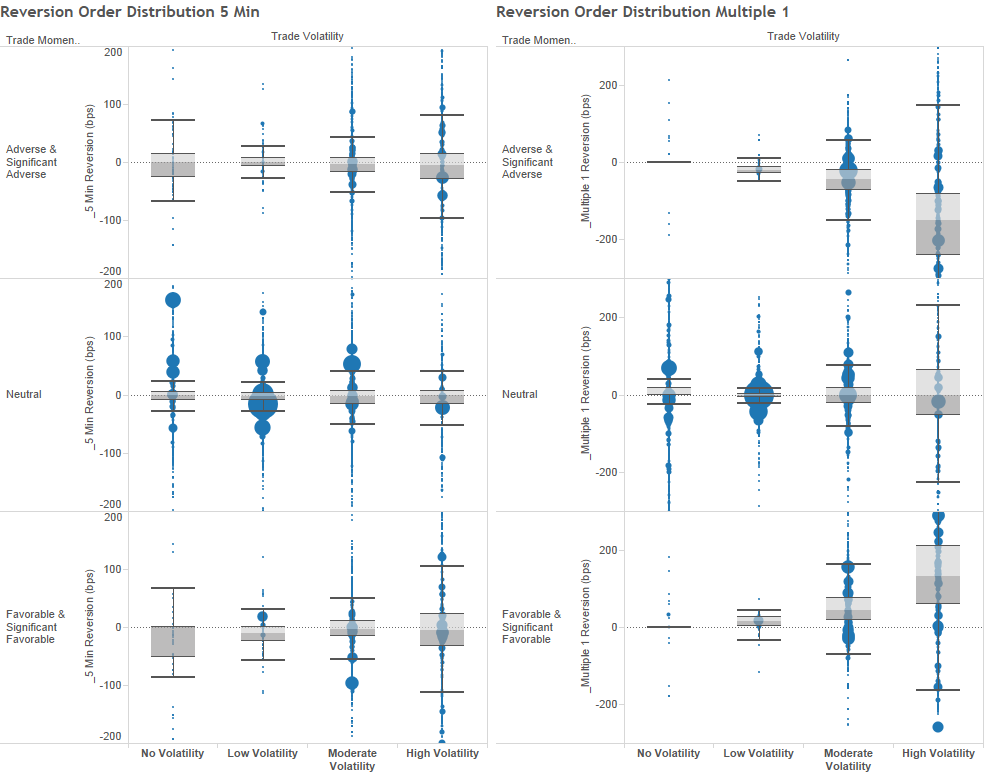}

\caption{\label{fig:Reversion-Distributions-Shorter-Horizon}Reversion Distributions
by Momentum and Volatility Environments - Shorter Horizon}
\end{figure}

\begin{figure}[H]
\includegraphics[width=18.5cm]{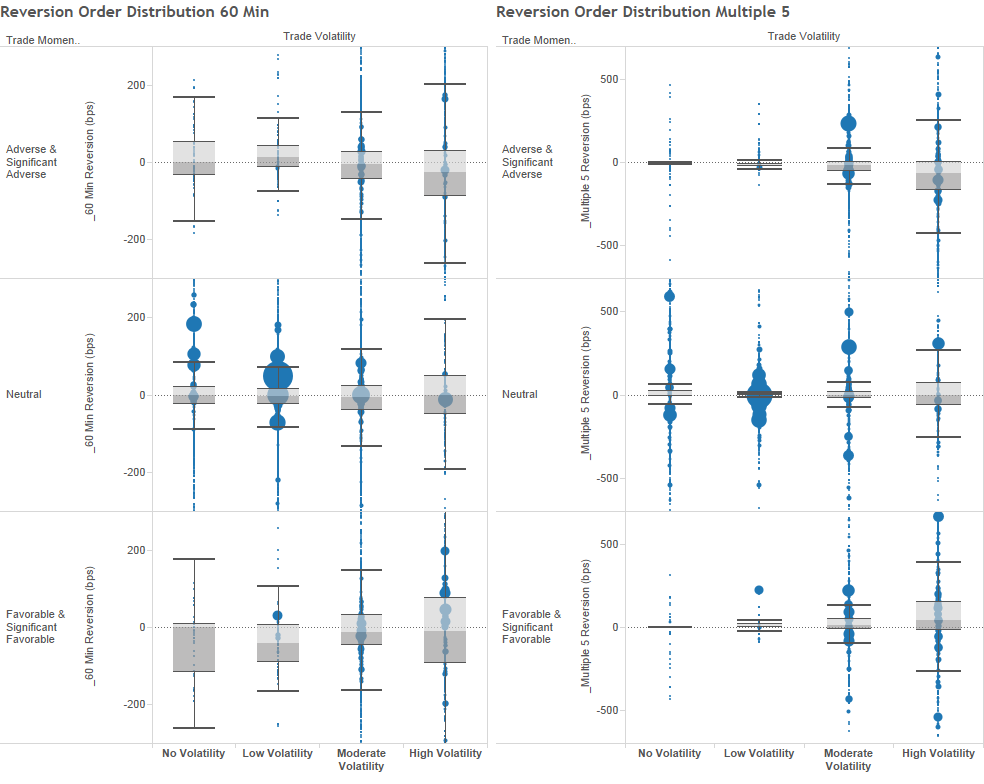}

\caption{\label{fig:Reversion-Distributions-Longer-Horizon}Reversion Distributions
by Momentum and Volatility Environments - Longer Horizon}
\end{figure}

\begin{figure}[H]
\includegraphics[width=12cm]{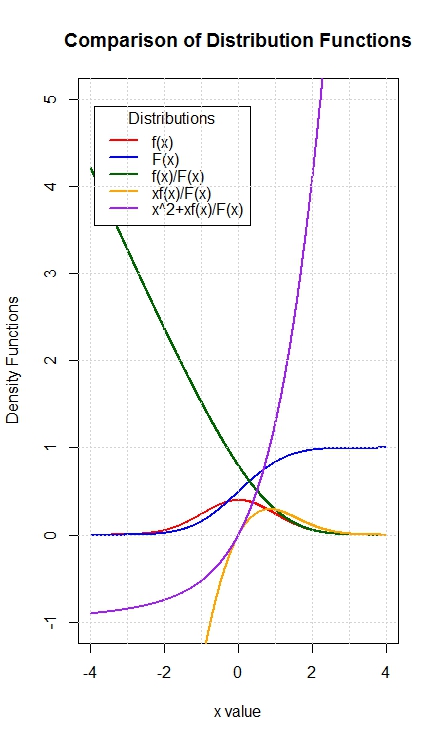}

\caption{\label{fig:Convexity-of-Distribution}Convexity of Distribution Functions}
\end{figure}

\end{doublespace}

\begin{figure}[H]
\includegraphics[width=17.5cm]{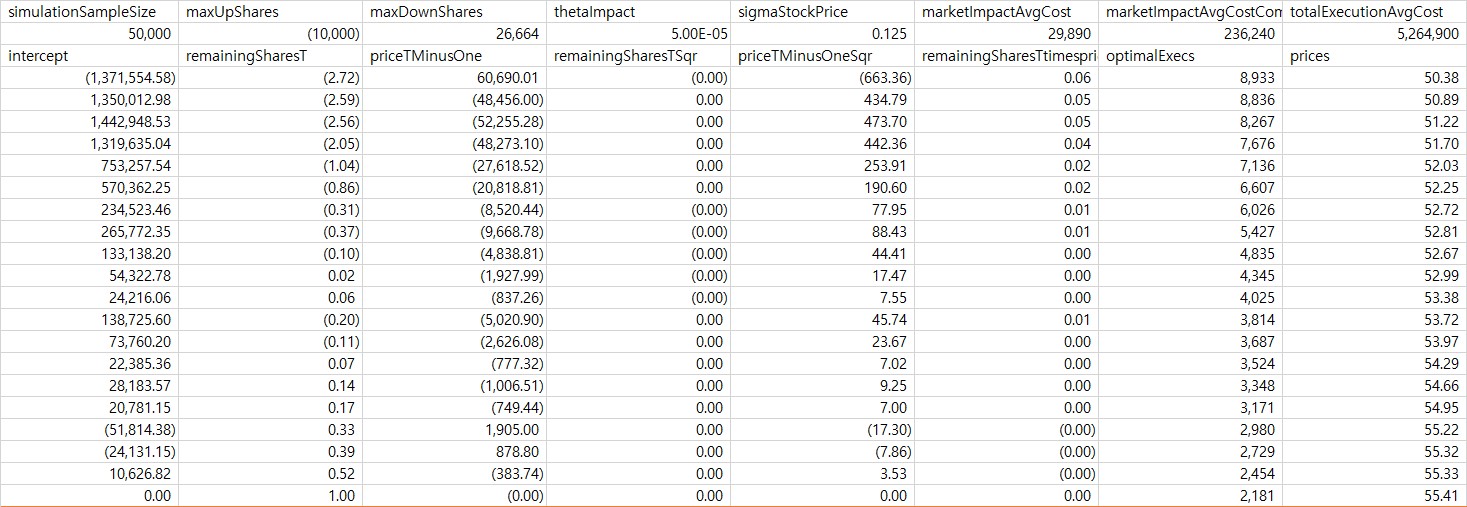}

\caption{\label{fig:Regression-Co-efficients-Complex}Complex Regression Coefficients
and Optimal Executions}
\end{figure}

\begin{figure}[H]
\includegraphics[width=17.5cm]{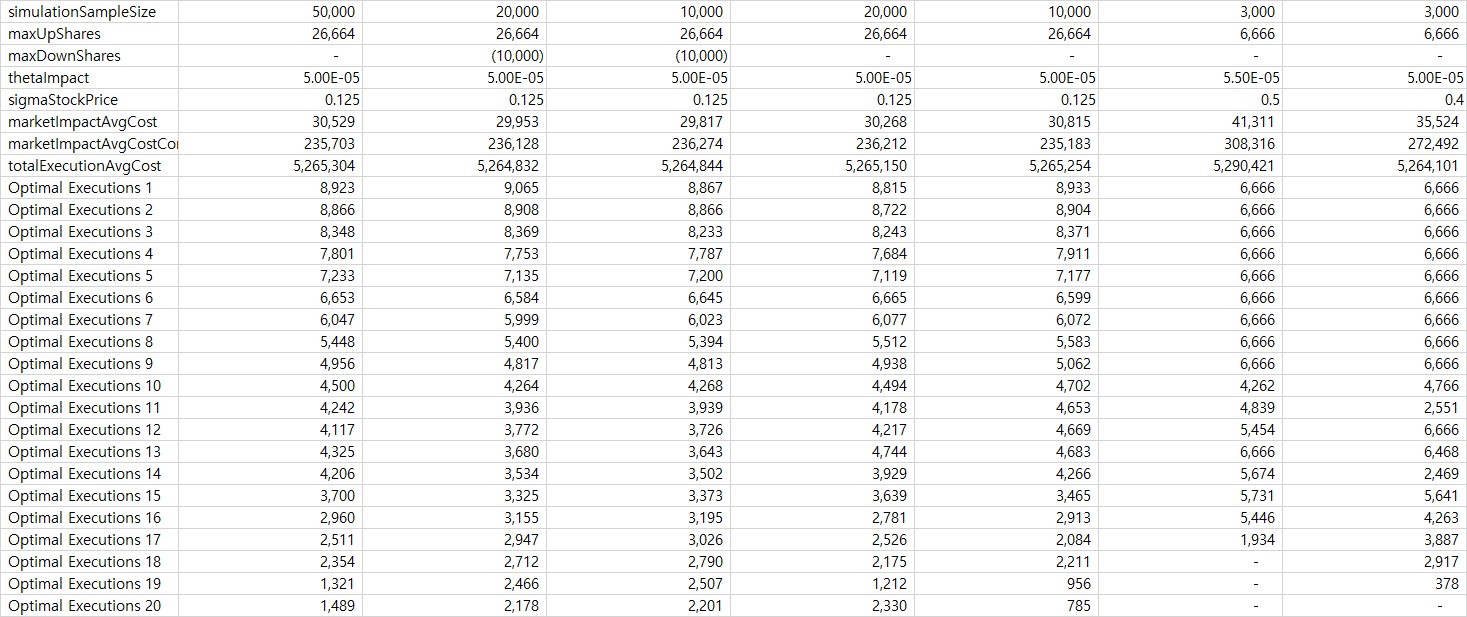}

\caption{\label{fig:Optimal Executions-Complex}Complex Optimal Executions
for Different Parameters}
\end{figure}
\begin{figure}[H]
\includegraphics[width=8cm]{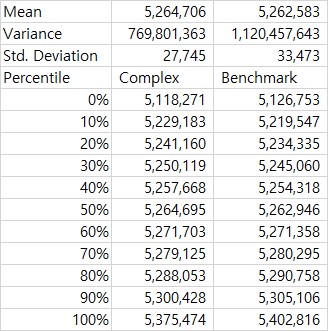}

\caption{\label{fig:Mean Variance Comparison}Mean / Variance / Percentile
Comparison of Total Execution Costs}
\end{figure}

\begin{figure}[H]
\includegraphics[width=15cm]{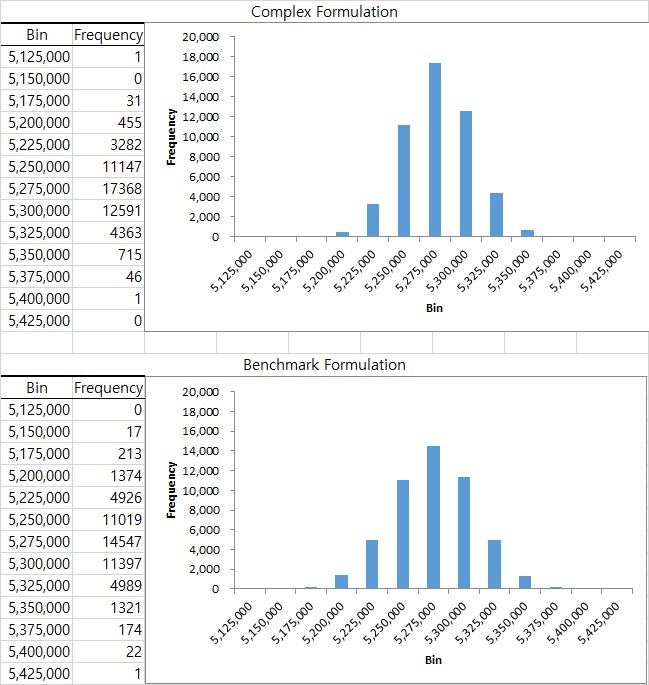}

\caption{\label{fig:Histogram Execution Costs}Histogram of Total Execution
Costs}
\end{figure}
\begin{figure}[H]
\includegraphics[width=17.5cm]{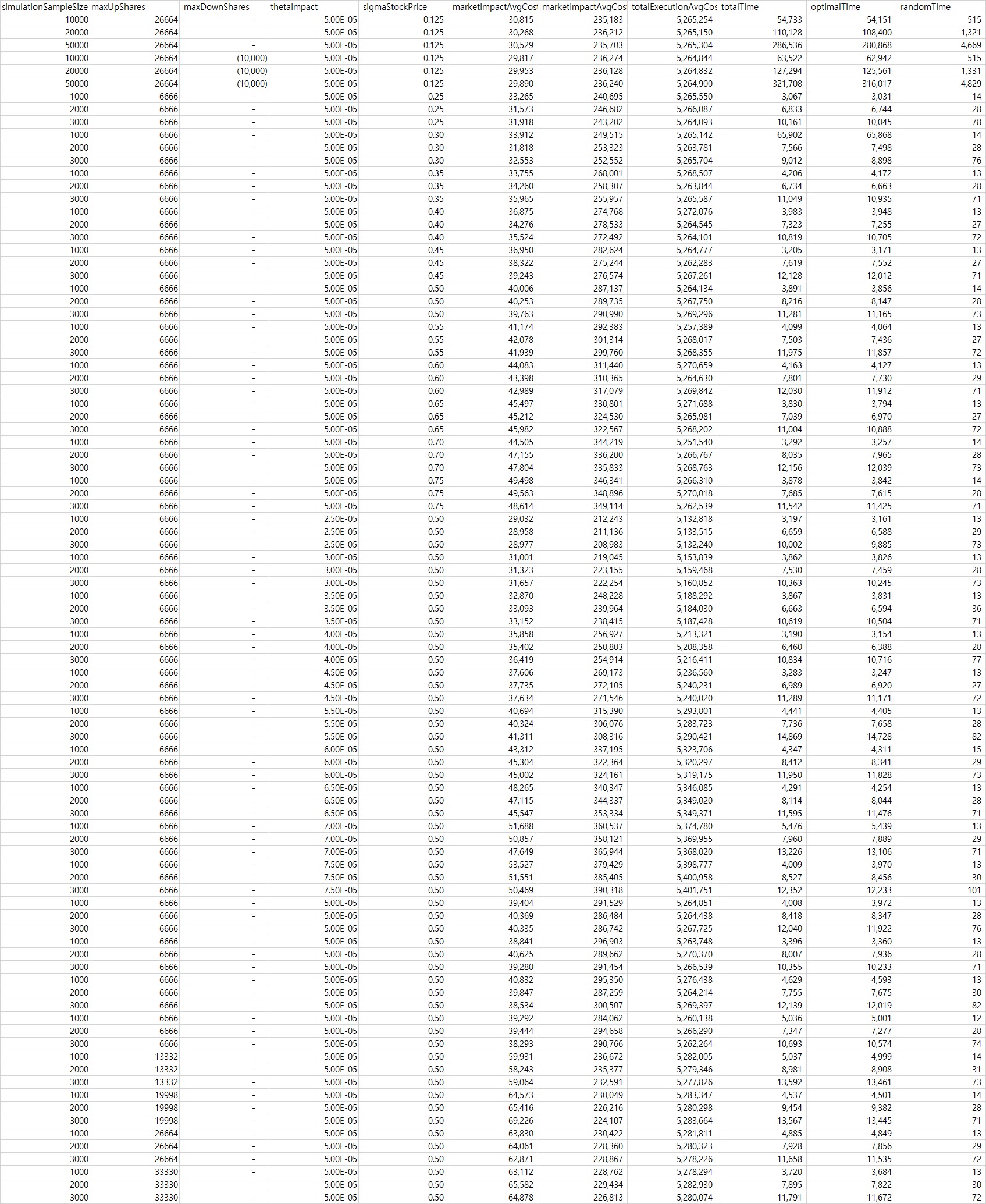}

\caption{\label{fig:Complex Executions-Costs}Complex Executions Costs for
Different Parameters}
\end{figure}

\begin{figure}[H]
\includegraphics[width=17.5cm]{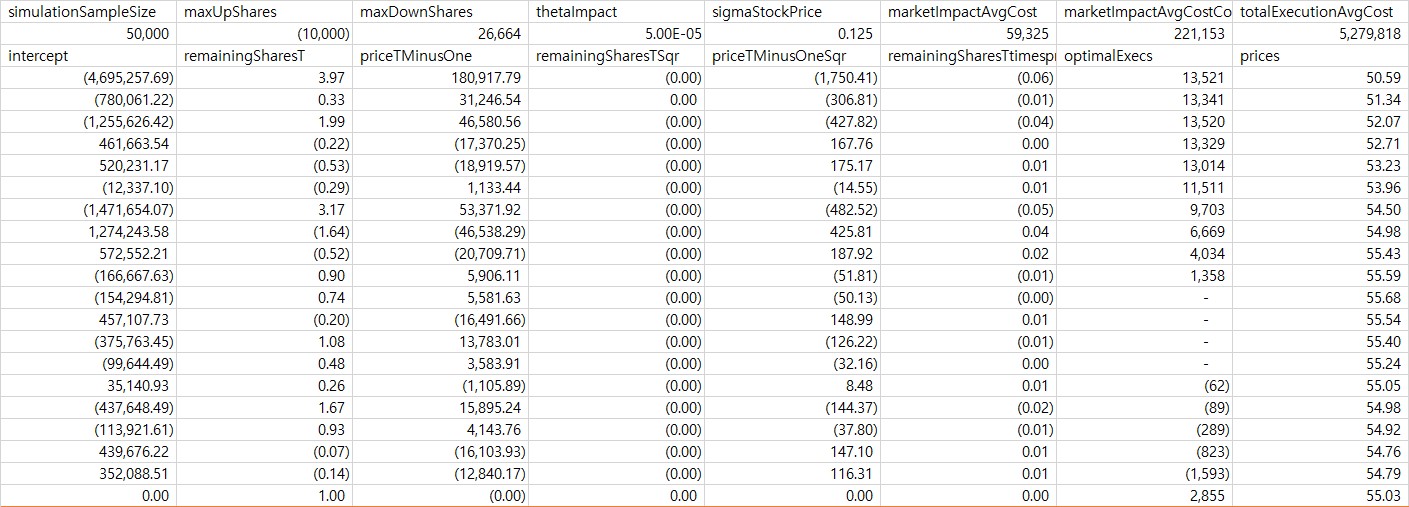}\caption{\label{fig:Regression-Co-efficients-for-Simple}Simple Regression
Coefficients and Optimal Executions}
\end{figure}

\begin{figure}[H]
\includegraphics[width=17.5cm]{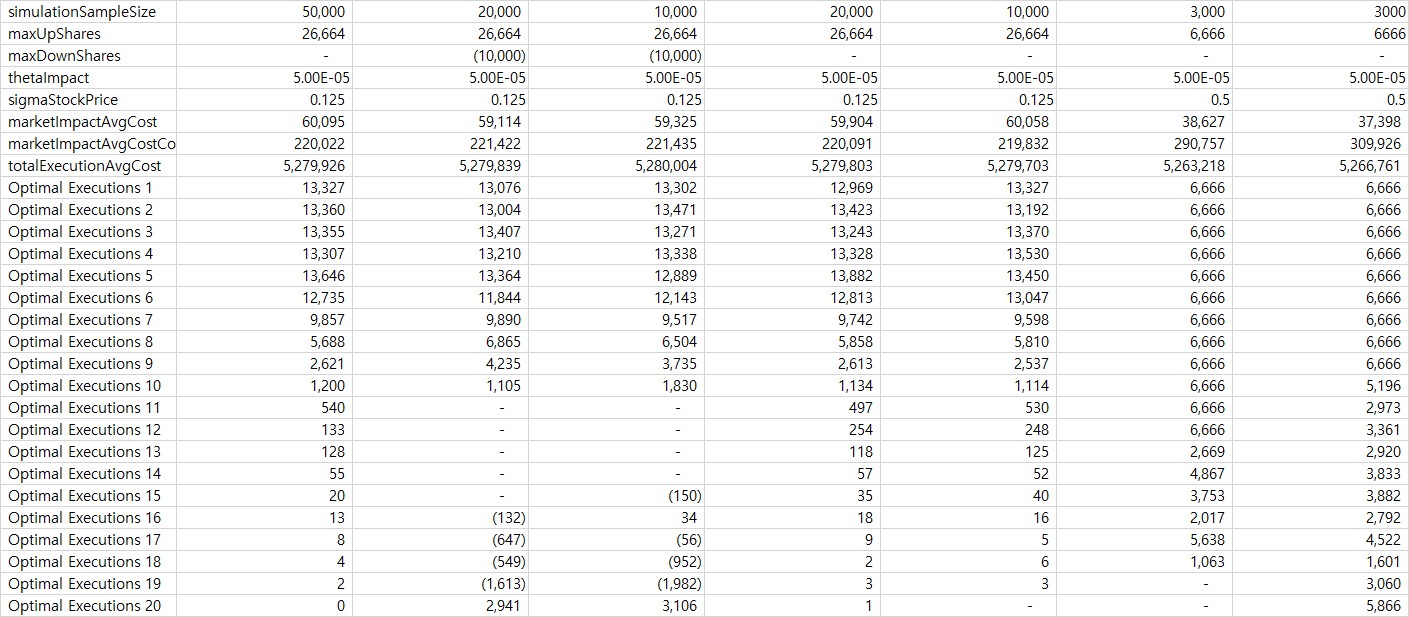}

\caption{\label{fig:Optimal Executions-Simple}Simple Optimal Executions for
Different Parameters}
\end{figure}
\begin{figure}[H]
\includegraphics[width=17.5cm]{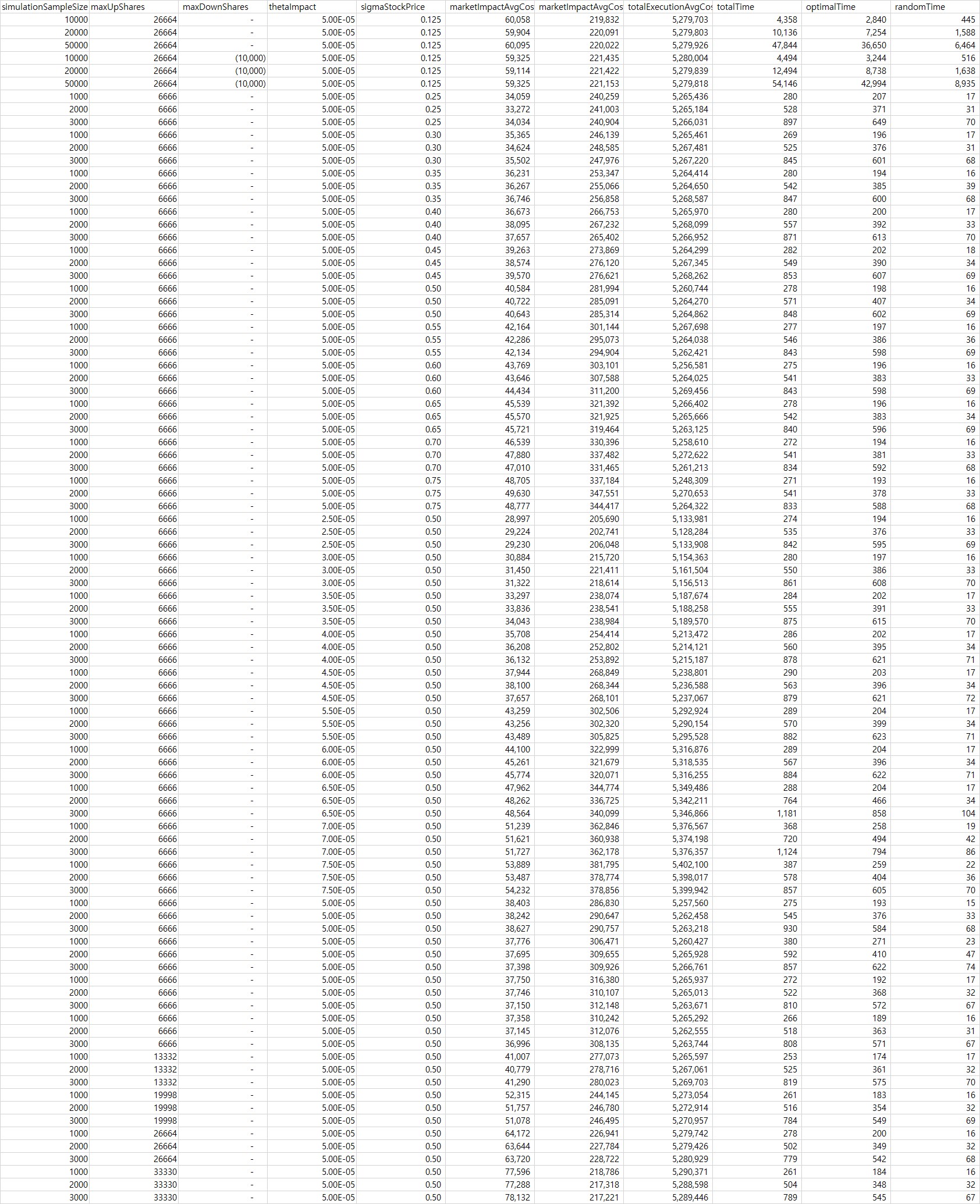}\caption{\label{fig:Simple-Executions-Costs}Simple Executions Costs for Different
Parameters}
\end{figure}

\begin{figure}[H]
\includegraphics[width=17.5cm]{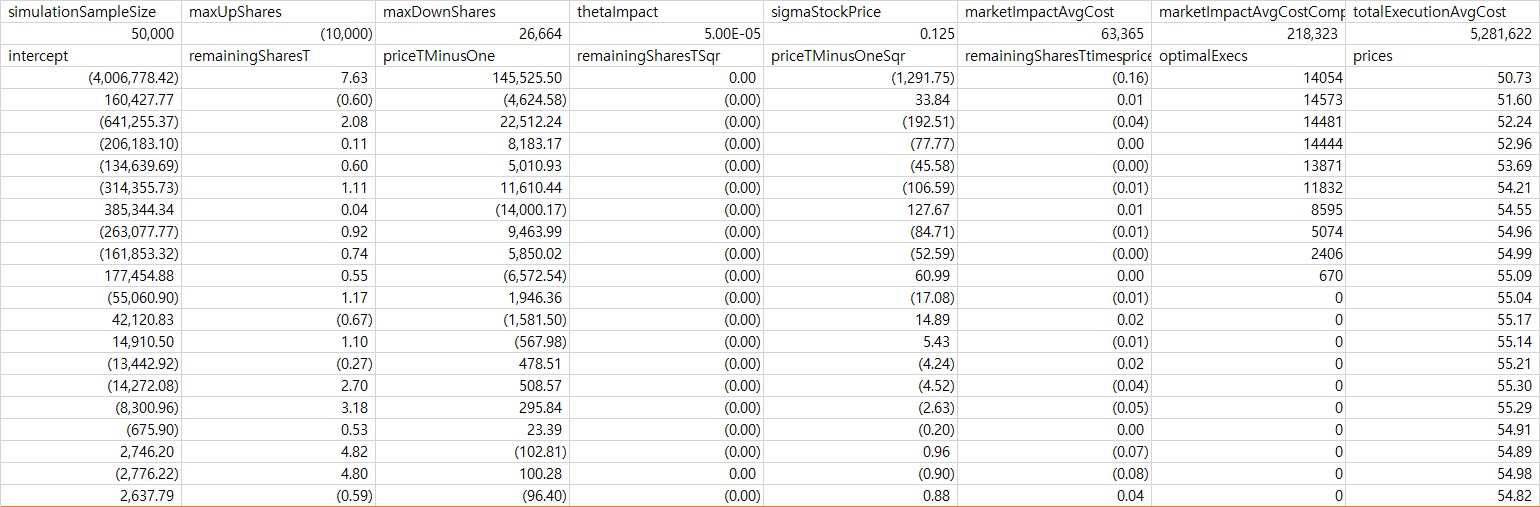}\caption{\label{fig:Regression-Co-efficients-for-One-Step}One Step Regression
Coefficients and Optimal Executions}
\end{figure}

\begin{figure}[H]
\includegraphics[width=17.5cm]{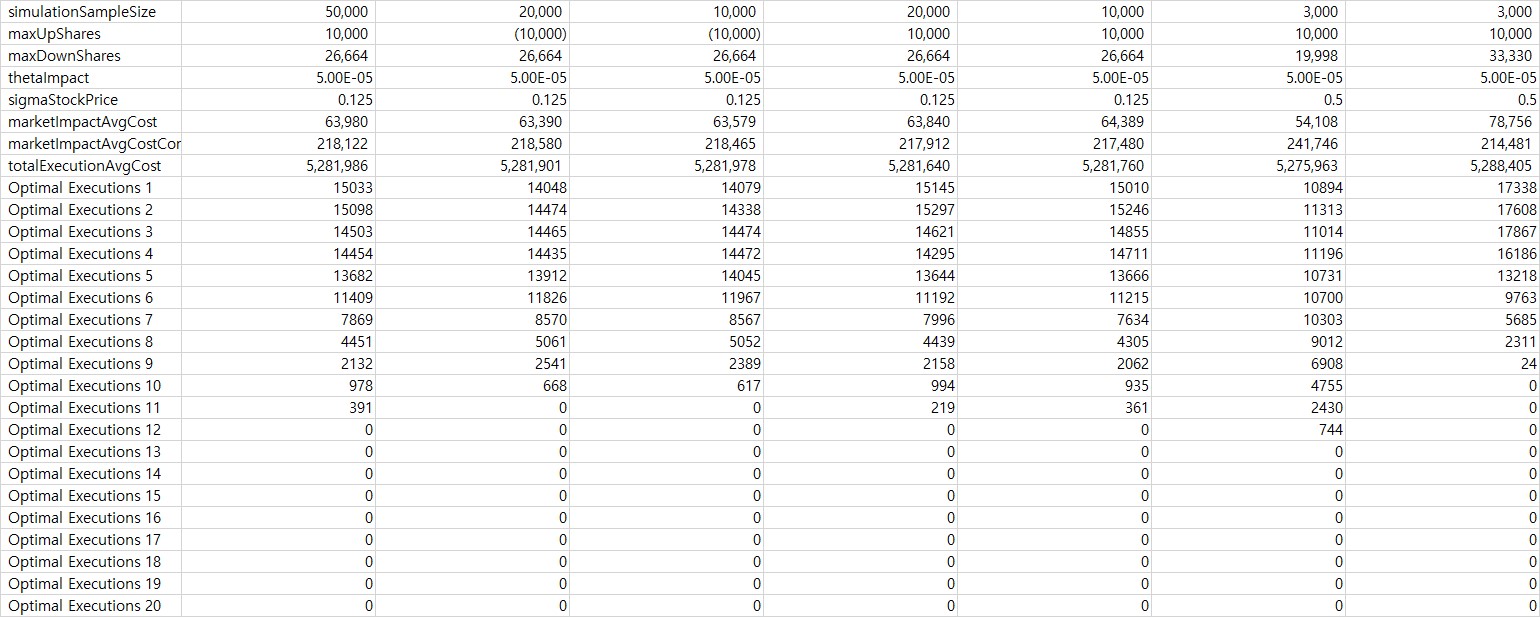}

\caption{\label{fig:Optimal Executions-OneStep}One Step Optimal Executions
for Different Parameters}
\end{figure}
\begin{figure}[H]
\includegraphics[width=17.5cm]{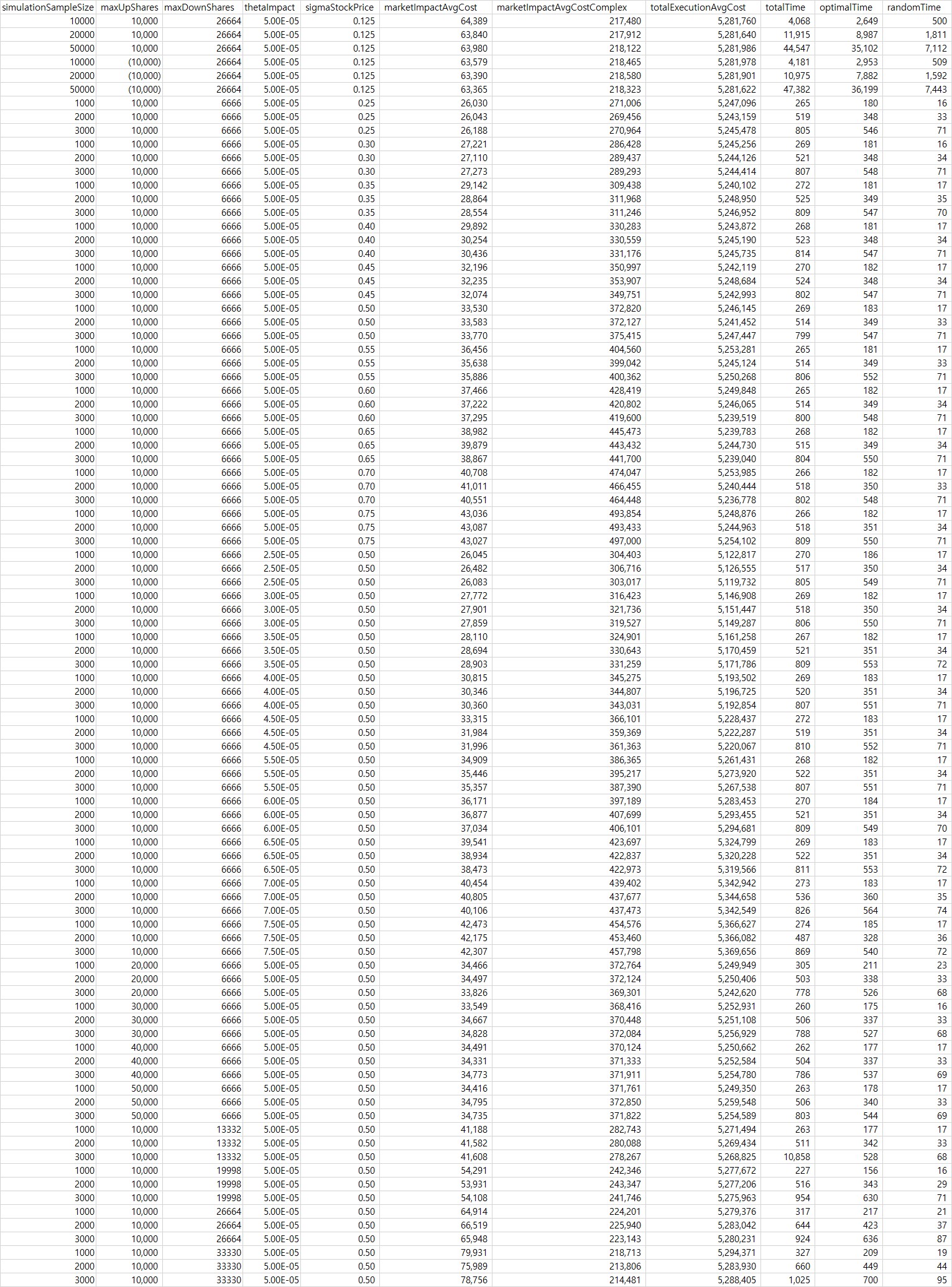}\caption{\label{fig:One-Step-Executions}One Step Executions Costs for Different
Parameters}
\end{figure}

\begin{doublespace}
\begin{figure}[H]
\includegraphics[width=18.5cm]{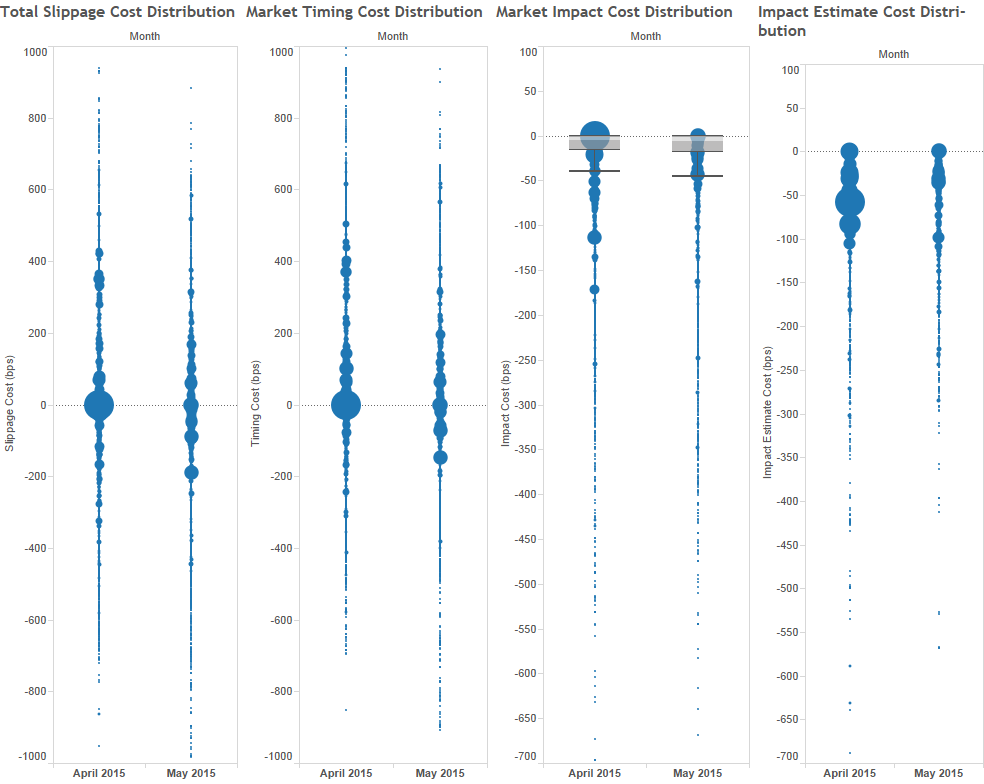}

\caption{Actual Trading Cost Distributions \label{fig:Trading-Cost-Distributions}}
\end{figure}

\end{doublespace}

\section{\label{sec:Appendix-of-Notation,}Appendix of Notation, Termonology
and Illustrative Examples}
\begin{doublespace}

\subsection{\label{subsec:Dictionary-of-Notation}Dictionary of Notation for
Optimal Trading using a Dynamic Programming Approach}
\end{doublespace}
\begin{itemize}
\begin{doublespace}
\item $\bar{S}$, the total number of shares that need to be traded.
\item $T$, the total duration of trading.
\item $N$, the number of trading intervals.
\item $\tau=T/N$, the length of each trading interval. We assume the time
intervals are of the same duration, but this can be relaxed quite
easily. In continuous time, this becomes, $N\rightarrow\infty,\tau\rightarrow0$.
\item The time then becomes divided into discrete intervals, $t_{k}=k\tau,\;k=0,...,N$.
\item For simplicity, let time be measured in unit intervals giving, $t=1,2,...,T$.
\item $S_{t}$, the number of shares acquired in period $t$ at price $P_{t}$.
\item $P_{0}$ can be any reference price or benchmark used to measure the
slippage. It is generally taken to be the arrival price or the price
at which the portfolio manager would like to complete the purchase
of the portfolio.
\item $W_{1},...,W_{T+1}$ is the trading trajectory, or a list of total
pending shares, $W_{1},...,W_{T+1}$. Here, $W_{t}$ is the number
of units that we still need to trade at time $t$. 
\end{doublespace}
\end{itemize}

\subsection{\label{subsec:Market-Impact-Simple-Example}Market Impact Simple
Formulation Examples}
\begin{enumerate}
\begin{doublespace}
\item When all the successive price moves are above their corresponding
previous price $P_{t}>P_{t-1}$, that is $\max\left[\left(P_{t}-P_{t-1}\right),0\right]=\left(P_{t}-P_{t-1}\right)$,
we have
\begin{align}
\text{Market Impact} & =\sum_{t=1}^{T}\left\{ \max\left[\left(P_{t}-P_{t-1}\right),0\right]S_{t}\right\} \\
 & =S_{1}\left(P_{1}-P_{0}\right)+S_{2}\left(P_{2}-P_{1}\right)+S_{3}\left(P_{3}-P_{2}\right)+\;...\;+S_{T}\left(P_{T}-P_{T-1}\right)
\end{align}
\begin{align}
\text{\text{Market Timing}} & =\text{Implementation Shortfall}-\text{\text{Market Impact}}\\
 & =\left(\sum_{t=1}^{T}S_{t}P_{t}\right)-\bar{S}P_{0}-S_{1}\left(P_{1}-P_{0}\right)-S_{2}\left(P_{2}-P_{1}\right)-S_{3}\left(P_{3}-P_{2}\right)-\;...\;-S_{T}\left(P_{T}-P_{T-1}\right)\\
 & =S_{1}P_{0}+S_{2}P_{1}+S_{3}P_{2}+\;...\;+S_{T}P_{T-1}-\bar{S}P_{0}\\
 & =S_{2}\left(P_{1}-P_{0}\right)+S_{3}\left(P_{2}-P_{0}\right)+\;...\;+S_{T}\left(P_{T-1}-P_{0}\right)
\end{align}
\item Some of the successive prices are below their corresponding previous
price, let us say, $\left(P_{2}<P_{1}\right)\text{ and }\left(P_{3}<P_{2}\right)$,
we have
\begin{align}
\text{Market Impact} & =\sum_{t=1}^{T}\left\{ \max\left[\left(P_{t}-P_{t-1}\right),0\right]S_{t}\right\} \\
 & =S_{1}\left(P_{1}-P_{0}\right)+S_{2}\left(0\right)+S_{3}\left(0\right)+\;...\;+S_{T}\left(P_{T}-P_{T-1}\right)
\end{align}
\begin{align}
\text{\text{Market Timing}} & =\text{Implementation Shortfall}-\text{\text{Market Impact}}\\
 & =\left(\sum_{t=1}^{T}S_{t}P_{t}\right)-\bar{S}P_{0}-S_{1}\left(P_{1}-P_{0}\right)-S_{2}\left(0\right)-S_{3}\left(0\right)-\;...\;-S_{T}\left(P_{T}-P_{T-1}\right)\\
 & =S_{2}P_{2}+S_{3}P_{3}+S_{1}P_{0}+S_{4}P_{3}+S_{5}P_{4}+\;...\;+S_{T}P_{T-1}-\bar{S}P_{0}\\
 & =S_{2}\left(P_{2}-P_{0}\right)+S_{3}\left(P_{3}-P_{0}\right)+S_{4}\left(P_{3}-P_{0}\right)+S_{5}\left(P_{4}-P_{0}\right)+\;...\;+S_{T}\left(P_{T-1}-P_{0}\right)
\end{align}
\end{doublespace}
\end{enumerate}

\subsection{\label{subsec:Market-Impact-Complex-Example}Market Impact Complex
Formulation Examples}
\begin{enumerate}
\begin{doublespace}
\item When all the successive price moves are above their corresponding
previous price $P_{t}>P_{t-1}$, that is $\max\left[\left(P_{t}-P_{t-1}\right),0\right]=\left(P_{t}-P_{t-1}\right)$,
we have
\begin{align}
\text{Market Impact} & =\sum_{t=1}^{T}\left\{ \max\left[\left(P_{t}-P_{t-1}\right),0\right]W_{t}\right\} \\
 & =W_{1}\left(P_{1}-P_{0}\right)+W_{2}\left(P_{2}-P_{1}\right)+W_{3}\left(P_{3}-P_{2}\right)+\;...\;+W_{T}\left(P_{T}-P_{T-1}\right)
\end{align}
\begin{align}
\text{\text{Market Timing}} & =\text{Implementation Shortfall}-\text{\text{Market Impact}}\\
 & =\left(\sum_{t=1}^{T}S_{t}P_{t}\right)-\bar{S}P_{0}-W_{1}\left(P_{1}-P_{0}\right)-W_{2}\left(P_{2}-P_{1}\right)-W_{3}\left(P_{3}-P_{2}\right)-\;...\;-W_{T}\left(P_{T}-P_{T-1}\right)\\
 & =\left[\sum_{t=1}^{T}\left(W_{t}-W_{t+1}\right)P_{t}\right]-W_{1}P_{0}-W_{1}\left(P_{1}-P_{0}\right)\\
 & -W_{2}\left(P_{2}-P_{1}\right)-W_{3}\left(P_{3}-P_{2}\right)-\;...\;-W_{T}\left(P_{T}-P_{T-1}\right)\\
 & =\left(W_{1}-W_{2}\right)P_{1}+\left(W_{2}-W_{3}\right)P_{2}+...+\left(W_{T}-W_{T+1}\right)P_{T}\\
 & -W_{1}P_{0}-W_{1}\left(P_{1}-P_{0}\right)-W_{2}\left(P_{2}-P_{1}\right)-W_{3}\left(P_{3}-P_{2}\right)-\;...\;-W_{T}\left(P_{T}-P_{T-1}\right)\\
 & =0
\end{align}
\item Some of the successive prices are below their corresponding previous
price, let us say, $\left(P_{2}<P_{1}\right)\text{ and }\left(P_{3}<P_{2}\right)$,
we have
\begin{align}
\text{Market Impact} & =\sum_{t=1}^{T}\left\{ \max\left[\left(P_{t}-P_{t-1}\right),0\right]W_{t}\right\} \\
 & =W_{1}\left(P_{1}-P_{0}\right)+W_{2}\left(0\right)+W_{3}\left(0\right)+\;...\;+W_{T}\left(P_{T}-P_{T-1}\right)
\end{align}
\begin{align}
\text{\text{Market Timing}} & =\text{Implementation Shortfall}-\text{\text{Market Impact}}\\
 & =\left(\sum_{t=1}^{T}S_{t}P_{t}\right)-\bar{S}P_{0}-W_{1}\left(P_{1}-P_{0}\right)-W_{2}\left(0\right)-W_{3}\left(0\right)-\;...\;-W_{T}\left(P_{T}-P_{T-1}\right)\\
 & =\left[\sum_{t=1}^{T}\left(W_{t}-W_{t+1}\right)P_{t}\right]-W_{1}P_{0}-W_{1}\left(P_{1}-P_{0}\right)\\
 & -W_{2}\left(0\right)-W_{3}\left(0\right)-\;...\;-W_{T}\left(P_{T}-P_{T-1}\right)\\
 & =\left(W_{1}-W_{2}\right)P_{1}+\left(W_{2}-W_{3}\right)P_{2}+...+\left(W_{T}-W_{T+1}\right)P_{T}\\
 & -W_{1}P_{0}-W_{1}\left(P_{1}-P_{0}\right)-W_{2}\left(0\right)-W_{3}\left(0\right)-\;...\;-W_{T}\left(P_{T}-P_{T-1}\right)\\
 & =-W_{2}P_{1}+W_{2}P_{2}-W_{3}P_{2}+W_{3}P_{3}\\
 & =W_{2}\left(P_{2}-P_{1}\right)+W_{3}\left(P_{3}-P_{2}\right)
\end{align}
\end{doublespace}
\end{enumerate}

\section{Appendix of Mathematical Proofs}
\begin{doublespace}

\subsection{\label{subsec:Proof-of-Proposition-Trading-Costs-Zero-Sum}Proof
of Theorem \ref{Trading-costs-Zero-Sum-Game}}
\end{doublespace}
\begin{lem}
\begin{doublespace}
\label{We-first-consider- I}We first consider the simple formulation,
with one interval and two market participants,
\end{doublespace}
\end{lem}
\begin{proof}
\begin{doublespace}
For the buyer we have,
\[
\text{Market Impact}=\left\{ \max\left[\left(P_{t}-P_{t-1}\right),0\right]S_{t}\right\} 
\]
\begin{eqnarray*}
\text{\text{Market Timing}} & = & \text{Implementation Shortfall}-\text{\text{Market Impact}}\\
 & = & \left(S_{t}P_{t}\right)-S_{t}P_{0}-\left\{ \max\left[\left(P_{t}-P_{t-1}\right),0\right]S_{t}\right\} 
\end{eqnarray*}
Here, we use the definition of Implementation Shortfall after adapting
it to the one interval case, 
\begin{eqnarray*}
\text{Implementation Shortfall} & = & \text{Paper Return}-\text{Real Portfolio Return}\\
 & = & \left(\sum_{t=1}^{T}S_{t}P_{t}\right)-\bar{S}P_{0}=\left(S_{t}P_{t}\right)-S_{t}P_{0}
\end{eqnarray*}
Similarly we have for the seller (noting that the drop in prices is
detrimental to the intended outcome and changing the sign accordingly),
\[
\text{Market Impact}=\left\{ \max\left[\left(P_{t-1}-P_{t}\right),0\right]S_{t}\right\} 
\]
\begin{eqnarray*}
\text{\text{Market Timing}} & = & -\text{Implementation Shortfall}-\text{\text{Market Impact}}\\
 &  & S_{t}P_{0}-\left(S_{t}P_{t}\right)-\left\{ \max\left[\left(P_{t-1}-P_{t}\right),0\right]S_{t}\right\} 
\end{eqnarray*}
If $\left(P_{t}>P_{t-1}\right)$, 

For the buyer, 
\[
\text{Market Impact}=\left(P_{t}-P_{t-1}\right)S_{t}
\]
\begin{eqnarray*}
\text{\text{Market Timing}} & = & \left(S_{t}P_{t}\right)-S_{t}P_{0}-\left(P_{t}-P_{t-1}\right)S_{t}=\left(P_{t-1}-P_{0}\right)S_{t}
\end{eqnarray*}
For the seller, 
\[
\text{Market Impact}=0
\]
\begin{eqnarray*}
\text{\text{Market Timing}} & = & S_{t}P_{0}-\left(S_{t}P_{t}\right)-0=\left(P_{0}-P_{t}\right)S_{t}
\end{eqnarray*}
Sum of the impact and timing across both the participants, 
\[
\text{Total Market Impact}=\left(P_{t}-P_{t-1}\right)S_{t}
\]
\begin{eqnarray*}
\text{Total \text{Market Timing}} & = & \left(P_{t-1}-P_{0}\right)S_{t}+\left(P_{0}-P_{t}\right)S_{t}=\left(P_{t-1}-P_{t}\right)S_{t}
\end{eqnarray*}
\[
\text{Total Market Impact}+\text{Total \text{Market Timing}}=0
\]
If $\left(P_{t}<P_{t-1}\right)$, 

For the buyer, 
\[
\text{Market Impact}=0
\]
\begin{eqnarray*}
\text{\text{Market Timing}} & = & \left(S_{t}P_{t}\right)-S_{t}P_{0}-0=\left(P_{t}-P_{0}\right)S_{t}
\end{eqnarray*}
For the seller, 
\[
\text{Market Impact}=\left(P_{t-1}-P_{t}\right)S_{t}
\]
\begin{eqnarray*}
\text{\text{Market Timing}} & = & -\text{Implementation Shortfall}-\text{\text{Market Impact}}\\
 &  & S_{t}P_{0}-\left(S_{t}P_{t}\right)-\left(P_{t-1}-P_{t}\right)S_{t}=S_{t}P_{0}-S_{t}P_{t-1}
\end{eqnarray*}
Sum of the impact and timing across both the participants, 
\[
\text{Total Market Impact}=\left(P_{t-1}-P_{t}\right)S_{t}
\]
\begin{eqnarray*}
\text{Total \text{Market Timing}} & = & \left(P_{0}-P_{t-1}\right)S_{t}+\left(P_{t}-P_{0}\right)S_{t}=\left(P_{t}-P_{t-1}\right)S_{t}
\end{eqnarray*}
\[
\text{Total Market Impact}+\text{Total \text{Market Timing}}=0
\]

It should be clear that this holds for all non-zero positive values
of prices and number of shares which can include zero, that is $\forall P_{t\in\left\{ t=0,1,2,...,T\right\} }\in\left(0,\infty\right)$
and $\forall S_{t}\in\left[0,\infty\right)$
\end{doublespace}
\end{proof}
\begin{lem}
\begin{doublespace}
\label{We-next-consider-II}We next consider the simple formulation,
with multiple intervals and multiple participants.
\end{doublespace}
\end{lem}
\begin{enumerate}
\begin{doublespace}
\item We argue that this scenario with multiple intervals and multiple participants
can be reduced to an amalgamation of the above case (Lemma \ref{We-first-consider- I})
with a single interval and two participants.
\item Our definition of an interval is such that during each interval, only
one exchange happens between buyer and seller for a total of two participants,
with the sum of market impact and market timing being equal to zero.
\item To convince us that such an interval exists, we reason as follows:
when multiple exchanges happen during an interval, we split the interval
into sub-intervals such that only one exchange happens in each interval.
If multiple exchanges happen simultaneously, they can be viewed as
one exchange by combining all the buy trades on one side against the
sell trades on the other side.
\item The sum of many such individual intervals also have the same property
(by mathematical induction), wherein the sum of market impact and
market timing equals zero, which follows from (Lemma \ref{We-first-consider- I}).
It should be clear that this holds for all non-zero positive values
of prices and number of shares which can include zero, that is $\forall P_{t\in\left\{ t=0,1,2,...,T\right\} }\in\left(0,\infty\right)$
and $\forall S_{t}\in\left[0,\infty\right)$ across all the intervals
considered.
\end{doublespace}
\end{enumerate}
\begin{lem}
\begin{doublespace}
\label{Lastly,-we-consider-III}Lastly, we consider the complex formulation
with multiple intervals and multiple participants.
\end{doublespace}
\end{lem}
\begin{enumerate}
\begin{doublespace}
\item We argue that the complex formulation scenario with multiple intervals
and multiple participants can be reduced to an amalgamation of the
above two cases (Lemma \ref{We-first-consider- I}, \ref{We-next-consider-II}).
\item Any shares unexecuted by the end of a certain time interval will need
to be executed before the end of the total time duration available
for trading, since we note that by assumption, there will be no unexecuted
shares once the total time duration is completed.
\item We apply Lemma \ref{We-first-consider- I} to the sum of impact and
timing for the shares executed at the last time interval, making this
sum zero. We then consider the last interval and the interval before
that together and apply Lemma \ref{We-next-consider-II} to these
two intervals, which gives the sum of impact and timing across both
these intervals as zero. 
\item We can then include additional intervals towards the beginning of
the trading duration and deduce that the sum of impact and timing
across the new interval and the already aggregated intervals is zero
using mathematical induction. It should be clear that this holds for
all non-zero positive values of prices and number of shares which
can include zero, that is $\forall P_{t\in\left\{ t=0,1,2,...,T\right\} }\in\left(0,\infty\right)$
and $\forall S_{t}\in\left[0,\infty\right)$ across all the intervals
considered.
\item Hence, by considering the shares executed in the last interval and
successively including the intervals before that, we get that the
corresponding sum of market impact and market timing equals zero.
The result is that we have at the end of the total trading duration
after aggregating across all the individual intervals, the sum of
total market impact and total market timing being equal to zero. This
completest the proof of Theorem \ref{Trading-costs-Zero-Sum-Game}.
\end{doublespace}
\end{enumerate}
\begin{doublespace}

\subsection{\label{subsec:Proof-of-Proposition-Convexity}Proof of Proposition
\ref{The-value-function-convexity}}
\end{doublespace}
\begin{proof}
\begin{doublespace}
We have from the value function for the last time period.

\[
V_{T}\left(P_{T-1},W_{T}\right)=E_{T}\left[\max\left\{ \left(\theta W_{T}+\varepsilon_{T}\right),0\right\} W_{T}\right]
\]

\[
V_{T}\left(P_{T-1},W_{T}\right)=E_{T}\left[\max\left\{ \left(\theta W_{T}^{2}+W_{T}\varepsilon_{T}\right),0\right\} \right]
\]
\[
V_{T}\left(P_{T-1},W_{T}\right)=E_{T}\left[\left.\left(\theta W_{T}^{2}+W_{T}\varepsilon_{T}\right)\right|\left(\theta W_{T}^{2}+W_{T}\varepsilon_{T}\right)>0\right]
\]
\[
\left\{ \because\;E\left[max\left(X,c\right)\right]=E\left[X\left|X>c\right.\right]Pr\left[X>c\right]+E\left[c\left|X\text{\ensuremath{\le}}c\right.\right]Pr\left[X\text{\ensuremath{\le}}c\right]\;\right\} 
\]
This is of the form, $E\left[\left.Y\right|Y>0\right]$ where, $Y=\left(\theta W_{T}^{2}+W_{T}\varepsilon_{T}\right)$.
We then need to calculate, 
\[
E_{T}\left[\left.\left(\theta W_{T}^{2}+W_{T}\sigma_{\varepsilon}Z\right)\right|Z>\left(-\frac{\theta W_{T}}{\sigma_{\varepsilon}}\right)\right]\quad,where\;Z\sim N\left(0,1\right)
\]
\[
\left[\because\;Y\sim N\left(\theta W_{T}^{2},W_{T}^{2}\sigma_{\varepsilon}^{2}\right)\equiv Y\sim N\left(\mu,\sigma^{2}\right)\Rightarrow Y=\mu+\sigma Z\;;\;Y>0\Rightarrow Z>-\mu/\sigma\right]
\]
We have for every standard normal distribution, $Z$, and for every
$u,$ $Pr\left[Z>\text{\textminus}u\right]=Pr\left[Z<u\right]=\mathbf{\Phi}\left(u\right)$.
Here, $\phi$ and $\mathbf{\Phi}$ are the standard normal PDF and
CDF, respectively. 
\begin{eqnarray*}
E\left[\left.Z\right|Z>-u\right] & = & \frac{1}{\mathbf{\Phi}\left(u\right)}\left[\int_{-u}^{\infty}t\phi\left(t\right)dt\right]\\
 & = & \frac{1}{\mathbf{\Phi}\left(u\right)}\left[\left.-\phi\left(t\right)\right|_{-u}^{\infty}\right]=\frac{\phi\left(u\right)}{\mathbf{\Phi}\left(u\right)}
\end{eqnarray*}
\[
\left[\because\int t\phi\left(t\right)dt=\int t{\displaystyle {\frac{1}{\sqrt{2\pi}}}e^{-{\frac{1}{2}}t^{2}}}dt=\int\left\{ \frac{d{\displaystyle {-\frac{1}{\sqrt{2\pi}}}e^{-{\frac{1}{2}}t^{2}}}}{dt}\right\} dt\right]
\]
Hence we have, 
\begin{eqnarray*}
E\left[\left.Y\right|Y>0\right] & = & \mu+\sigma E\left[\left.Z\right|Z>\left(-\frac{\mu}{\sigma}\right)\right]\\
 & = & \mu+\frac{\sigma\phi\left(\mu/\sigma\right)}{\mathbf{\Phi}\left(\mu/\sigma\right)}
\end{eqnarray*}
Setting, $\psi\left(u\right)=u+\phi\left(u\right)/\Phi\left(u\right)$,
\[
E\left[\left.Y\right|Y>0\right]=\sigma\psi\left(\mu/\sigma\right)
\]
\begin{eqnarray*}
V_{T}\left(P_{T-1},W_{T}\right) & = & E_{T}\left[\left.\left(\theta W_{T}^{2}+W_{T}\varepsilon_{T}\right)\right|\left(\theta W_{T}^{2}+W_{T}\varepsilon_{T}\right)>0\right]\\
 & = & W_{T}\sigma_{\varepsilon}\left[\frac{\theta W_{T}}{\sigma_{\varepsilon}}+\frac{\phi\left(\frac{\theta W_{T}}{\sigma_{\varepsilon}}\right)}{\Phi\left(\frac{\theta W_{T}}{\sigma_{\varepsilon}}\right)}\right]=\sigma_{\varepsilon}W_{T}\psi\left(\xi W_{T}\right)\;,\;\xi=\frac{\theta}{\sigma_{\varepsilon}}
\end{eqnarray*}
 In the next to last period, $T-1$, the Bellman equation is,
\[
V_{T-1}\left(P_{T-2},W_{T-1}\right)=\underset{\left\{ S_{T-1}\right\} }{\min}\:E_{T-1}\left[\max\left\{ \left(P_{T-1}-P_{T-2}\right),0\right\} S_{T-1}+V_{T}\left(P_{T-1},W_{T}\right)\right]
\]
\[
=\underset{\left\{ S_{T-1}\right\} }{\min}\:E_{T-1}\left[\max\left\{ \left(\theta S_{T-1}+\varepsilon_{T-1}\right),0\right\} S_{T-1}+V_{T}\left(P_{T-2}+\theta S_{T-1}+\varepsilon_{T-1},W_{T-1}-S_{T-1}\right)\right]
\]
\[
=\underset{\left\{ S_{T-1}\right\} }{\min}\left\{ S_{T-1}\sigma_{\varepsilon}\left[\frac{\theta S_{T-1}}{\sigma_{\varepsilon}}+\frac{\phi\left(\frac{\theta S_{T-1}}{\sigma_{\varepsilon}}\right)}{\Phi\left(\frac{\theta S_{T-1}}{\sigma_{\varepsilon}}\right)}\right]+\left(W_{T-1}-S_{T-1}\right)\sigma_{\varepsilon}\left[\frac{\theta\left(W_{T-1}-S_{T-1}\right)}{\sigma_{\varepsilon}}+\frac{\phi\left(\frac{\theta\left(W_{T-1}-S_{T-1}\right)}{\sigma_{\varepsilon}}\right)}{\Phi\left(\frac{\theta\left(W_{T-1}-S_{T-1}\right)}{\sigma_{\varepsilon}}\right)}\right]\right\} 
\]
\[
=\underset{\left\{ S_{T-1}\right\} }{\min}\left[S_{T-1}\sigma_{\varepsilon}\psi\left(\xi S_{T-1}\right)+\left(W_{T-1}-S_{T-1}\right)\sigma_{\varepsilon}\psi\left\{ \xi\left(W_{T-1}-S_{T-1}\right)\right\} \right]
\]
We show this to be a convex function with a unique minimum. Let us
start with, 
\[
G\left(x\right)=\frac{\phi\left(x\right)}{\Phi\left(x\right)}\quad\left|\;\forall x>0\right.
\]
\[
\frac{\partial G\left(x\right)}{\partial x}=\frac{\phi'\left(x\right)}{\Phi\left(x\right)}-\left[\frac{\phi\left(x\right)}{\Phi\left(x\right)}\right]^{2}
\]
\[
\frac{\partial G\left(x\right)}{\partial x}=-\frac{x\phi\left(x\right)}{\Phi\left(x\right)}-\left[\frac{\phi\left(x\right)}{\Phi\left(x\right)}\right]^{2}
\]
\[
\left[\because\;\frac{\partial\phi\left(x\right)}{\partial\left(x\right)}=-x\phi\left(x\right)\;;\;\frac{\partial\Phi\left(x\right)}{\partial\left(x\right)}=\phi\left(x\right)\right]
\]
\[
\frac{\partial G\left(x\right)}{\partial x}=\frac{\phi\left(x\right)}{\Phi\left(x\right)}\left[-x-\frac{\phi\left(x\right)}{\Phi\left(x\right)}\right]<0\quad\left|\;\forall x>0\right.
\]
\begin{eqnarray*}
\frac{\partial^{2}G\left(x\right)}{\partial x^{2}} & = & -\frac{\phi\left(x\right)}{\Phi\left(x\right)}-x\left\{ -\frac{x\phi\left(x\right)}{\Phi\left(x\right)}-\left[\frac{\phi\left(x\right)}{\Phi\left(x\right)}\right]^{2}\right\} -\frac{2\phi\left(x\right)\phi'\left(x\right)}{\Phi^{2}\left(x\right)}+\frac{2\phi^{3}\left(x\right)}{\Phi^{3}\left(x\right)}
\end{eqnarray*}
\begin{eqnarray*}
 & = & \frac{-\phi\left(x\right)\Phi^{2}\left(x\right)+x^{2}\phi\left(x\right)\Phi^{2}\left(x\right)+x\phi^{2}\left(x\right)\Phi\left(x\right)+2x\phi^{2}\left(x\right)\Phi\left(x\right)+2\phi^{3}\left(x\right)}{\Phi^{3}\left(x\right)}
\end{eqnarray*}
\begin{eqnarray*}
 & = & \frac{\phi\left(x\right)\left\{ -\Phi^{2}\left(x\right)+x^{2}\Phi^{2}\left(x\right)+3x\phi\left(x\right)\Phi\left(x\right)+2\phi^{2}\left(x\right)\right\} }{\Phi^{3}\left(x\right)}
\end{eqnarray*}
Consider the following, $\forall x>0$
\[
\text{Let }K\left(x\right)=3x\phi\left(x\right)\Phi\left(x\right)+2\phi^{2}\left(x\right)+\left(x^{2}-1\right)\Phi^{2}\left(x\right)
\]
For $x\geq1,K\left(x\right)>0$. Also, 
\[
K\left(0\right)=\frac{1}{\pi}-\frac{1}{4}>0
\]
\[
\left[\because\;\phi\left(x\right)=\frac{e^{-\frac{1}{2}x^{2}}}{\sqrt{2\pi}}\;;\;\phi\left(0\right)=\frac{1}{\sqrt{2\pi}}\;;\;\Phi\left(x\right)=\frac{1}{2}\left\{ 1+\text{erf\ensuremath{\left(\frac{x}{\sqrt{2}}\right)}}\right\} \;;\;\Phi\left(0\right)=\frac{1}{2}\right]
\]
For $x\in\left(0,1\right)$,
\[
K'\left(x\right)=3\phi\left(x\right)\Phi\left(x\right)-3x^{2}\phi\left(x\right)\Phi\left(x\right)+3x\phi^{2}\left(x\right)-4x\phi^{2}\left(x\right)+\left(x^{2}-1\right)2\Phi\left(x\right)\phi\left(x\right)+2x\Phi^{2}\left(x\right)
\]
\[
K'\left(x\right)=x\left\{ 2\Phi^{2}\left(x\right)-\phi^{2}\left(x\right)\right\} +\left\{ 1-x^{2}\right\} \phi\left(x\right)\Phi\left(x\right)\geq xL\left(x\right)
\]
where, $L\left(x\right)=\left\{ 2\Phi^{2}\left(x\right)-\phi^{2}\left(x\right)\right\} $.
Further, $\Phi\left(x\right)\geq\frac{1}{2}$ and $\phi^{2}\left(x\right)\leq\frac{1}{2\pi}$
$\Rightarrow L\left(x\right)\geq\frac{1}{2}-\frac{1}{2\pi}>0$. $K'\left(x\right)>0\Rightarrow K\left(x\right)$
is increasing. Hence, $K\left(x\right)>K\left(0\right)>0\;\left|\;\forall u\in\left(0,1\right)\right.$.
This gives, $K\left(x\right)>0$ and $\frac{\partial^{2}G\left(x\right)}{\partial x^{2}}>0\;\left|\;\forall x\in\left(0,\infty\right)\right.$.
It is worth noting the following asymptotic properties
\[
\left[\because\;\underset{x\rightarrow0+}{\lim}\frac{\partial^{2}G\left(x\right)}{\partial x^{2}}>0\;;\;\underset{x\rightarrow\infty}{\lim}\frac{\partial^{2}G\left(x\right)}{\partial x^{2}}=0\;;\;\frac{\partial G\left(x\right)}{\partial x}<0\quad\left|\;\forall x>0\right.\right]
\]

Next we show that $f\left(a-x\right)$ is convex, given $f''\left(x\right)>0\;;\;x>0$
\[
\text{Let }y=a-x
\]
\begin{eqnarray*}
\frac{\partial f\left(y\right)}{\partial x} & = & \frac{\partial f\left(y\right)}{\partial y}\frac{\partial\left(a-x\right)}{\partial x}\;;\;a>x
\end{eqnarray*}
\[
=\left(-1\right)f'\left(y\right)
\]
\begin{eqnarray*}
\frac{\partial^{2}f\left(y\right)}{\partial x^{2}} & = & \left(-1\right)\frac{\partial f'\left(y\right)}{\partial y}\frac{\partial\left(a-x\right)}{\partial x}
\end{eqnarray*}
\[
=f''\left(y\right)
\]
\[
>0\;\left[\because\;f''\left(y\right)>0\;\left|\;\forall y>0\right.\right]
\]
We can similarly show that $f\left(bx\right)$ is convex if $f\left(x\right)$
is convex (in our case, $b>0$, but the result holds $\forall b$).
Next we derive conditions when $x^{2}+xf\left(x\right)$ is convex,
given $f''\left(x\right)>0\;;\;x>0$ 
\begin{eqnarray*}
\text{Let }g\left(x\right) & = & x^{2}+xf\left(x\right)\\
g'\left(x\right) & = & 2x+xf'\left(x\right)+f\left(x\right)\\
g''\left(x\right) & = & 2+xf''\left(x\right)+2f'\left(x\right)
\end{eqnarray*}
\[
g''\left(x\right)>0\;\text{if}\;f'\left(x\right)>0\;\text{or if }\;2+xf''\left(x\right)>\left|2f'\left(x\right)\right|
\]
Finally,
\[
\text{Let }Q\left(x\right)=x^{2}+x\frac{\phi\left(x\right)}{\Phi\left(x\right)}
\]
\begin{eqnarray*}
\frac{\partial Q\left(x\right)}{\partial x} & = & 2x+x\left[-\frac{x\phi\left(x\right)}{\Phi\left(x\right)}-\left\{ \frac{\phi\left(x\right)}{\Phi\left(x\right)}\right\} ^{2}\right]+\frac{\phi\left(x\right)}{\Phi\left(x\right)}
\end{eqnarray*}
\begin{eqnarray*}
\left.\frac{\partial Q\left(x\right)}{\partial x}\right|_{x=0} & = & \frac{\phi\left(0\right)}{\Phi\left(0\right)}>0
\end{eqnarray*}
\begin{eqnarray*}
\frac{\partial^{2}Q\left(x\right)}{\partial x^{2}} & = & 2+x\phi\left(x\right)\left[\frac{-\Phi^{2}\left(x\right)+x^{2}\Phi^{2}\left(x\right)+3x\phi\left(x\right)\Phi\left(x\right)+2\phi^{2}\left(x\right)}{\Phi^{3}\left(x\right)}\right]+2\left[-x\frac{\phi\left(x\right)}{\Phi\left(x\right)}-\left\{ \frac{\phi\left(x\right)}{\Phi\left(x\right)}\right\} ^{2}\right]
\end{eqnarray*}
\begin{eqnarray*}
 & = & 2+\phi\left(x\right)\left[\frac{x^{3}\Phi^{2}\left(x\right)+3x^{2}\phi\left(x\right)\Phi\left(x\right)+2x\phi^{2}\left(x\right)-3x\Phi^{2}\left(x\right)-2\phi\left(x\right)\Phi\left(x\right)}{\Phi^{3}\left(x\right)}\right]
\end{eqnarray*}
\begin{eqnarray*}
 & = & \left[\frac{2\Phi^{3}\left(x\right)+x^{3}\Phi^{2}\left(x\right)\phi\left(x\right)+3x^{2}\phi^{2}\left(x\right)\Phi\left(x\right)+2x\phi^{3}\left(x\right)-3x\Phi^{2}\left(x\right)\phi\left(x\right)-2\phi^{2}\left(x\right)\Phi\left(x\right)}{\Phi^{3}\left(x\right)}\right]
\end{eqnarray*}
\[
\text{Let, }K\left(x\right)=2\Phi^{3}\left(x\right)+2x\phi^{3}\left(x\right)+x^{3}\Phi^{2}\left(x\right)\phi\left(x\right)+3x^{2}\phi^{2}\left(x\right)\Phi\left(x\right)-3x\Phi^{2}\left(x\right)\phi\left(x\right)-2\phi^{2}\left(x\right)\Phi\left(x\right)
\]
We need to show that $K\left(x\right)>0\:\left|\forall x>0\right.$.
First we note that,
\[
K\left(0\right)=\frac{1}{4}-\frac{1}{2\pi}>0
\]
We can write this as,
\[
K\left(x\right)=2x\phi^{3}\left(x\right)+x^{3}\Phi^{2}\left(x\right)\phi\left(x\right)+3x^{2}\phi^{2}\left(x\right)\Phi\left(x\right)+\Phi\left(x\right)\left[2\Phi^{2}\left(x\right)-3x\Phi\left(x\right)\phi\left(x\right)-2\phi^{2}\left(x\right)\right]
\]
We then need to show, 
\[
L\left(x\right)=\left[2\Phi^{2}\left(x\right)-3x\Phi\left(x\right)\phi\left(x\right)-2\phi^{2}\left(x\right)\right]>0\left|\forall x>0\right.
\]
\[
L\left(0\right)=\left[2\Phi^{2}\left(0\right)-2\phi^{2}\left(0\right)\right]=\frac{1}{2}-\frac{1}{\pi}>0
\]
\begin{eqnarray*}
\frac{\partial L\left(x\right)}{\partial x} & = & 4\Phi\left(x\right)\phi\left(x\right)-3\Phi\left(x\right)\phi\left(x\right)-3x\phi^{2}\left(x\right)+3x^{2}\Phi\left(x\right)\phi\left(x\right)+4x\phi^{2}\left(x\right)\\
 & = & \Phi\left(x\right)\phi\left(x\right)+3x^{2}\Phi\left(x\right)\phi\left(x\right)+x\phi^{2}\left(x\right)>0\left|\forall x\geq0\right.
\end{eqnarray*}
Therefore $L(x)$ is an increasing function on the interval $[0,\infty)$.
Its minimum must be at $L(0)>0$, proving $L(x)>0$ and $\frac{\partial^{2}Q\left(x\right)}{\partial x^{2}}>0\;\left|\;\forall x\in\left(0,\infty\right)\right.$.
It is worth noting the following asymptotic properties and the graphical
results shown in the main text, 
\[
\left[\because\;\underset{x\rightarrow0+}{\lim}\frac{\partial^{2}Q\left(x\right)}{\partial x^{2}}>0\;;\;\underset{x\rightarrow\infty}{\lim}\frac{\partial^{2}Q\left(x\right)}{\partial x^{2}}>0\;;\;\frac{\partial Q\left(x\right)}{\partial x}>0\quad\left|\;\forall x>0\right.\right]
\]
\end{doublespace}
\end{proof}
\begin{doublespace}

\subsection{\label{subsec:Proof-of-Proposition-benchmark-simple}Proof of Proposition
\ref{The-number-of-benchmark-simple}}
\end{doublespace}
\begin{proof}
\begin{doublespace}
Consider,
\begin{eqnarray*}
V_{T-1}\left(P_{T-2},W_{T-1}\right) & = & \underset{\left\{ S_{T-1}\right\} }{\min}\left[S_{T-1}\sigma_{\varepsilon}\psi\left(\xi S_{T-1}\right)+\left(W_{T-1}-S_{T-1}\right)\sigma_{\varepsilon}\psi\left\{ \xi\left(W_{T-1}-S_{T-1}\right)\right\} \right]\\
 &  & \text{Here, }\psi\left(u\right)=u+\phi\left(u\right)/\Phi\left(u\right)\;,\;\xi=\frac{\theta}{\sigma_{\varepsilon}},
\end{eqnarray*}
First Order Conditions (FOC) give,
\[
\frac{\partial}{\partial S_{T-1}}\left[S_{T-1}\sigma_{\varepsilon}\psi\left(\xi S_{T-1}\right)+\left(W_{T-1}-S_{T-1}\right)\sigma_{\varepsilon}\psi\left\{ \xi\left(W_{T-1}-S_{T-1}\right)\right\} \right]=0
\]
\[
\xi S_{T-1}\psi'\left(\xi S_{T-1}\right)+\psi\left(\xi S_{T-1}\right)-\xi\left(W_{T-1}-S_{T-1}\right)\psi'\left(\xi\left\{ W_{T-1}-S_{T-1}\right\} \right)-\psi\left(\xi\left\{ W_{T-1}-S_{T-1}\right\} \right)=0
\]
\begin{eqnarray*}
\xi^{2}S_{T-1}\left\{ 1-\frac{\xi S_{T-1}\phi\left(\xi S_{T-1}\right)}{\Phi\left(\xi S_{T-1}\right)}-\left[\frac{\phi\left(\xi S_{T-1}\right)}{\Phi\left(\xi S_{T-1}\right)}\right]^{2}\right\} +\left\{ \xi S_{T-1}+\frac{\phi\left(\xi S_{T-1}\right)}{\Phi\left(\xi S_{T-1}\right)}\right\} \\
+\xi^{2}\left(W_{T-1}-S_{T-1}\right)\left\{ -1+\frac{\xi\left(W_{T-1}-S_{T-1}\right)\phi\left(\xi\left\{ W_{T-1}-S_{T-1}\right\} \right)}{\Phi\left(\xi\left\{ W_{T-1}-S_{T-1}\right\} \right)}+\left[\frac{\phi\left(\xi\left\{ W_{T-1}-S_{T-1}\right\} \right)}{\Phi\left(\xi\left\{ W_{T-1}-S_{T-1}\right\} \right)}\right]^{2}\right\} \\
-\left\{ \xi\left\{ W_{T-1}-S_{T-1}\right\} +\frac{\phi\left(\xi\left\{ W_{T-1}-S_{T-1}\right\} \right)}{\Phi\left(\xi\left\{ W_{T-1}-S_{T-1}\right\} \right)}\right\}  & = & 0
\end{eqnarray*}
\[
\left[\because\;\psi'\left(u\right)=1-\frac{u\phi\left(u\right)}{\Phi\left(u\right)}-\left[\frac{\phi\left(u\right)}{\Phi\left(u\right)}\right]^{2}\text{ and }\psi\left(u\right)=u+\frac{\phi\left(u\right)}{\Phi\left(u\right)}\right]
\]
\begin{eqnarray*}
S_{T-1}+\frac{1}{\xi}S_{T-1}+\frac{1}{\xi^{2}}\frac{\phi\left(\xi S_{T-1}\right)}{\Phi\left(\xi S_{T-1}\right)}+\frac{\xi\left(W_{T-1}-S_{T-1}\right)^{2}\phi\left(\xi\left\{ W_{T-1}-S_{T-1}\right\} \right)}{\Phi\left(\xi\left\{ W_{T-1}-S_{T-1}\right\} \right)}\\
+\left(W_{T-1}-S_{T-1}\right)\left[\frac{\phi\left(\xi\left\{ W_{T-1}-S_{T-1}\right\} \right)}{\Phi\left(\xi\left\{ W_{T-1}-S_{T-1}\right\} \right)}\right]^{2} & =\\
\left(W_{T-1}-S_{T-1}\right)+\frac{1}{\xi}\left\{ W_{T-1}-S_{T-1}\right\} +\frac{1}{\xi^{2}}\frac{\phi\left(\xi\left\{ W_{T-1}-S_{T-1}\right\} \right)}{\Phi\left(\xi\left\{ W_{T-1}-S_{T-1}\right\} \right)}\\
+\frac{\xi S_{T-1}^{2}\phi\left(\xi S_{T-1}\right)}{\Phi\left(\xi S_{T-1}\right)}+S_{T-1}\left[\frac{\phi\left(\xi S_{T-1}\right)}{\Phi\left(\xi S_{T-1}\right)}\right]^{2}
\end{eqnarray*}
Setting $S_{T-1}=W_{T-1}/2$ gives $RHS=LHS$. We have, 
\begin{eqnarray*}
V_{T-1}\left(P_{T-2},W_{T-1}\right) & = & S_{T-1}\sigma_{\varepsilon}\left[\frac{\theta S_{T-1}}{\sigma_{\varepsilon}}+\frac{\phi\left(\frac{\theta S_{T-1}}{\sigma_{\varepsilon}}\right)}{\Phi\left(\frac{\theta S_{T-1}}{\sigma_{\varepsilon}}\right)}\right]\\
 &  & +\left(W_{T-1}-S_{T-1}\right)\sigma_{\varepsilon}\left[\frac{\theta\left(W_{T-1}-S_{T-1}\right)}{\sigma_{\varepsilon}}+\frac{\phi\left(\frac{\theta\left(W_{T-1}-S_{T-1}\right)}{\sigma_{\varepsilon}}\right)}{\Phi\left(\frac{\theta\left(W_{T-1}-S_{T-1}\right)}{\sigma_{\varepsilon}}\right)}\right]
\end{eqnarray*}
\[
V_{T-1}\left(P_{T-2},W_{T-1}\right)=W_{T-1}\sigma_{\varepsilon}\left[\frac{\theta W_{T-1}}{2\sigma_{\varepsilon}}+\frac{\phi\left(\frac{\theta W_{T-1}}{2\sigma_{\varepsilon}}\right)}{\Phi\left(\frac{\theta W_{T-1}}{2\sigma_{\varepsilon}}\right)}\right]
\]
Absent closed form solutions, numerical techniques using $\xi_{1}>0$
can be tried. We can also set $S_{T-1}\approx\omega_{1}\left(W_{T-1}\right)$
using a well behaved (continuous and differentiable) function, $\omega_{1}$.
But the former approach is simpler and lends itself easily to numerical
solutions that we will attempt in the more complex laws of motion
to follow. 
\[
S_{T-1}\approx\xi_{1}W_{T-1}
\]
or with additional terms including non-linear regressions as,
\[
S_{T-1}\approx\xi_{0}+\xi_{1}W_{T-1}+\xi_{2}\left(W_{T-1}\right)^{2}\text{ OR }S_{T-1}\approx\xi_{0}\left(W_{T-1}\right)^{\xi_{1}}
\]
\[
V_{T-1}\left(P_{T-2},W_{T-1}\right)=\left[\sigma_{\varepsilon}\xi_{1}W_{T-1}\psi\left(\xi\xi_{1}W_{T-1}\right)+\left\{ W_{T-1}-\xi_{1}W_{T-1}\right\} \sigma_{\varepsilon}\psi\left\{ \xi\left(W_{T-1}-\xi_{1}W_{T-1}\right)\right\} \right]
\]
\[
=\sigma_{\varepsilon}W_{T-1}\left[\xi_{1}\psi\left(\xi\xi_{1}W_{T-1}\right)+\left(1-\xi_{1}\right)\psi\left\{ \xi W_{T-1}\left(1-\xi_{1}\right)\right\} \right]
\]
\[
=\psi_{1}\left(W_{T-1}\right)\;,\text{here, }\psi_{1}\text{ is a convex function}.
\]
Continuing the recursion,
\[
V_{T-2}\left(P_{T-3},W_{T-2}\right)=\underset{\left\{ S_{T-2}\right\} }{\min}\:E_{T-2}\left[\max\left\{ \left(P_{T-2}-P_{T-3}\right),0\right\} S_{T-2}+V_{T-1}\left(P_{T-2},W_{T-1}\right)\right]
\]
\[
=\underset{\left\{ S_{T-2}\right\} }{\min}\:E_{T-2}\left[\max\left\{ \left(\theta S_{T-2}+\varepsilon_{T-2}\right),0\right\} S_{T-2}+V_{T-1}\left(P_{T-3}+\theta S_{T-2}+\varepsilon_{T-3},W_{T-2}-S_{T-2}\right)\right]
\]
\[
=\underset{\left\{ S_{T-2}\right\} }{\min}\left[S_{T-2}\sigma_{\varepsilon}\left\{ \frac{\theta S_{T-2}}{\sigma_{\varepsilon}}+\frac{\phi\left(\frac{\theta S_{T-2}}{\sigma_{\varepsilon}}\right)}{\Phi\left(\frac{\theta S_{T-2}}{\sigma_{\varepsilon}}\right)}\right\} +\left(W_{T-2}-S_{T-2}\right)\sigma_{\varepsilon}\left\{ \frac{\theta\left(W_{T-2}-S_{T-2}\right)}{2\sigma_{\varepsilon}}+\frac{\phi\left[\frac{\theta\left(W_{T-2}-S_{T-2}\right)}{2\sigma_{\varepsilon}}\right]}{\Phi\left[\frac{\theta\left(W_{T-2}-S_{T-2}\right)}{2\sigma_{\varepsilon}}\right]}\right\} \right]
\]
First Order Conditions (FOC) give,
\begin{eqnarray*}
3\theta S_{T-2}-\theta W_{T-2}+\frac{\sigma_{\varepsilon}\phi\left(\frac{\theta S_{T-2}}{\sigma_{\varepsilon}}\right)}{\Phi\left(\frac{\theta S_{T-2}}{\sigma_{\varepsilon}}\right)}-\frac{\theta^{2}S_{T-2}^{2}\phi\left(\frac{\theta S_{T-2}}{\sigma_{\varepsilon}}\right)}{\sigma_{\varepsilon}\Phi\left(\frac{\theta S_{T-2}}{\sigma_{\varepsilon}}\right)}-\theta S_{T-2}\left[\frac{\phi\left(\frac{\theta S_{T-2}}{\sigma_{\varepsilon}}\right)}{\Phi\left(\frac{\theta S_{T-2}}{\sigma_{\varepsilon}}\right)}\right]^{2}\\
-\frac{\sigma_{\varepsilon}\phi\left[\frac{\theta\left(W_{T-2}-S_{T-2}\right)}{2\sigma_{\varepsilon}}\right]}{\Phi\left[\frac{\theta\left(W_{T-2}-S_{T-2}\right)}{2\sigma_{\varepsilon}}\right]}+\frac{\theta^{2}\left(W_{T-2}-S_{T-2}\right)^{2}}{4\sigma_{\varepsilon}}\frac{\phi\left[\frac{\theta\left(W_{T-2}-S_{T-2}\right)}{2\sigma_{\varepsilon}}\right]}{\Phi\left[\frac{\theta\left(W_{T-2}-S_{T-2}\right)}{2\sigma_{\varepsilon}}\right]}+\frac{\theta\left(W_{T-2}-S_{T-2}\right)}{2}\left\{ \frac{\phi\left[\frac{\theta\left(W_{T-2}-S_{T-2}\right)}{2\sigma_{\varepsilon}}\right]}{\Phi\left[\frac{\theta\left(W_{T-2}-S_{T-2}\right)}{2\sigma_{\varepsilon}}\right]}\right\} ^{2} & = & 0
\end{eqnarray*}
This gives, $S_{T-2}=W_{T-2}/3$ and the corresponding value function
as, 
\[
V_{T-2}\left(P_{T-3},W_{T-2}\right)=\frac{W_{T-2}}{3}\sigma_{\varepsilon}\left\{ \frac{\theta}{\sigma_{\varepsilon}}\frac{W_{T-2}}{3}+\frac{\phi\left(\frac{\theta}{\sigma_{\varepsilon}}\frac{W_{T-2}}{3}\right)}{\Phi\left(\frac{\theta}{\sigma_{\varepsilon}}\frac{W_{T-2}}{3}\right)}\right\} +\left(\frac{2W_{T-2}}{3}\right)\sigma_{\varepsilon}\left\{ \frac{\theta}{2\sigma_{\varepsilon}}\frac{2W_{T-2}}{3}+\frac{\phi\left[\frac{\theta}{2\sigma_{\varepsilon}}\frac{2W_{T-2}}{3}\right]}{\Phi\left[\frac{\theta}{2\sigma_{\varepsilon}}\frac{2W_{T-2}}{3}\right]}\right\} 
\]
\[
V_{T-2}\left(P_{T-3},W_{T-2}\right)=\sigma_{\varepsilon}W_{T-2}\left[\frac{\theta W_{T-2}}{3\sigma_{\varepsilon}}+\frac{\phi\left(\frac{\theta W_{T-2}}{3\sigma_{\varepsilon}}\right)}{\Phi\left(\frac{\theta W_{T-2}}{3\sigma_{\varepsilon}}\right)}\right]
\]

We show the general case using induction. Let the value function hold
for $T-K$
\[
V_{T-K}\left(P_{T-K-1},W_{T-K}\right)=\sigma_{\varepsilon}W_{T-K}\left[\frac{\theta W_{T-K}}{\left(K+1\right)\sigma_{\varepsilon}}+\frac{\phi\left(\frac{\theta W_{T-K}}{\left(K+1\right)\sigma_{\varepsilon}}\right)}{\Phi\left(\frac{\theta W_{T-K}}{\left(K+1\right)\sigma_{\varepsilon}}\right)}\right]
\]
Continuing the recursion,
\begin{eqnarray*}
V_{T-K-1}\left(P_{T-K-2},W_{T-K-1}\right) & = & \underset{\left\{ S_{T-K-1}\right\} }{\min}\:E_{T-K-1}\left[\max\left\{ \left(P_{T-K-1}-P_{T-K-2}\right),0\right\} S_{T-K-1}\right.\\
 &  & +\left.V_{T-K}\left(P_{T-K-1},W_{T-K}\right)\right]
\end{eqnarray*}
\begin{eqnarray*}
 & = & \underset{\left\{ S_{T-K-1}\right\} }{\min}\:E_{T-K-1}\left[\max\left\{ \left(\theta S_{T-K-1}+\varepsilon_{T-K-1}\right),0\right\} S_{T-K-1}\right.\\
 &  & +\left.V_{T-K}\left(P_{T-K-2}+\theta S_{T-K-1}+\varepsilon_{T-K-2},W_{T-K-1}-S_{T-K-1}\right)\right]
\end{eqnarray*}
\begin{eqnarray*}
 & = & \underset{\left\{ S_{T-K-1}\right\} }{\min}\left[S_{T-K-1}\sigma_{\varepsilon}\left\{ \frac{\theta S_{T-K-1}}{\sigma_{\varepsilon}}+\frac{\phi\left(\frac{\theta S_{T-K-1}}{\sigma_{\varepsilon}}\right)}{\Phi\left(\frac{\theta S_{T-K-1}}{\sigma_{\varepsilon}}\right)}\right\} \right.\\
 &  & +\left.\left(W_{T-K-1}-S_{T-K-1}\right)\sigma_{\varepsilon}\left\{ \frac{\theta\left(W_{T-K-1}-S_{T-K-1}\right)}{\left(K+1\right)\sigma_{\varepsilon}}+\frac{\phi\left[\frac{\theta\left(W_{T-K-1}-S_{T-K-1}\right)}{\left(K+1\right)\sigma_{\varepsilon}}\right]}{\Phi\left[\frac{\theta\left(W_{T-K-1}-S_{T-K-1}\right)}{\left(K+1\right)\sigma_{\varepsilon}}\right]}\right\} \right]
\end{eqnarray*}
\begin{eqnarray*}
 & = & \underset{\left\{ S_{T-K-1}\right\} }{\min}\left[S_{T-K-1}\left\{ \xi S_{T-K-1}+\frac{\phi\left(\xi S_{T-K-1}\right)}{\Phi\left(\xi S_{T-K-1}\right)}\right\} \right.\\
 &  & +\left.\left(W_{T-K-1}-S_{T-K-1}\right)\left\{ \frac{\xi\left(W_{T-K-1}-S_{T-K-1}\right)}{\left(K+1\right)}+\frac{\phi\left[\frac{\xi\left(W_{T-K-1}-S_{T-K-1}\right)}{\left(K+1\right)}\right]}{\Phi\left[\frac{\xi\left(W_{T-K-1}-S_{T-K-1}\right)}{\left(K+1\right)}\right]}\right\} \right]
\end{eqnarray*}
First Order Conditions (FOC) give,
\begin{eqnarray*}
\left\{ \xi S_{T-K-1}+\frac{\phi\left(\xi S_{T-K-1}\right)}{\Phi\left(\xi S_{T-K-1}\right)}\right\} \\
+S_{T-K-1}\left\{ \xi-\xi^{2}S_{T-K-1}\frac{\phi\left(\xi S_{T-K-1}\right)}{\Phi\left(\xi S_{T-K-1}\right)}-\xi\left[\frac{\phi\left(\xi S_{T-K-1}\right)}{\Phi\left(\xi S_{T-K-1}\right)}\right]^{2}\right\} \\
-\left\{ \frac{\xi\left(W_{T-K-1}-S_{T-K-1}\right)}{\left(K+1\right)}+\frac{\phi\left[\frac{\xi\left(W_{T-K-1}-S_{T-K-1}\right)}{\left(K+1\right)}\right]}{\Phi\left[\frac{\xi\left(W_{T-K-1}-S_{T-K-1}\right)}{\left(K+1\right)}\right]}\right\} \\
+\left(W_{T-K-1}-S_{T-K-1}\right)\left\{ -\frac{\xi}{\left(K+1\right)}+\frac{\xi^{2}\left(W_{T-K-1}-S_{T-K-1}\right)}{\left(K+1\right)^{2}}\frac{\phi\left[\frac{\xi\left(W_{T-K-1}-S_{T-K-1}\right)}{\left(K+1\right)}\right]}{\Phi\left[\frac{\xi\left(W_{T-K-1}-S_{T-K-1}\right)}{\left(K+1\right)}\right]}\right.\\
+\left.\frac{\xi}{\left(K+1\right)}\left[\frac{\phi\left[\frac{\xi\left(W_{T-K-1}-S_{T-K-1}\right)}{\left(K+1\right)}\right]}{\Phi\left[\frac{\xi\left(W_{T-K-1}-S_{T-K-1}\right)}{\left(K+1\right)}\right]}\right]^{2}\right\}  & = & 0
\end{eqnarray*}
\begin{eqnarray*}
2\xi S_{T-K-1}+\frac{\phi\left(\xi S_{T-K-1}\right)}{\Phi\left(\xi S_{T-K-1}\right)}+\frac{\xi^{2}\left(W_{T-K-1}-S_{T-K-1}\right)^{2}}{\left(K+1\right)^{2}}\frac{\phi\left[\frac{\xi\left(W_{T-K-1}-S_{T-K-1}\right)}{\left(K+1\right)}\right]}{\Phi\left[\frac{\xi\left(W_{T-K-1}-S_{T-K-1}\right)}{\left(K+1\right)}\right]}\\
+\frac{\xi\left(W_{T-K-1}-S_{T-K-1}\right)}{\left(K+1\right)}\left[\frac{\phi\left[\frac{\xi\left(W_{T-K-1}-S_{T-K-1}\right)}{\left(K+1\right)}\right]}{\Phi\left[\frac{\xi\left(W_{T-K-1}-S_{T-K-1}\right)}{\left(K+1\right)}\right]}\right]^{2} & =\\
\xi^{2}S_{T-K-1}^{2}\frac{\phi\left(\xi S_{T-K-1}\right)}{\Phi\left(\xi S_{T-K-1}\right)}+\xi S_{T-K-1}\left[\frac{\phi\left(\xi S_{T-K-1}\right)}{\Phi\left(\xi S_{T-K-1}\right)}\right]^{2}\\
+\frac{2\xi\left(W_{T-K-1}-S_{T-K-1}\right)}{\left(K+1\right)}+\frac{\phi\left[\frac{\xi\left(W_{T-K-1}-S_{T-K-1}\right)}{\left(K+1\right)}\right]}{\Phi\left[\frac{\xi\left(W_{T-K-1}-S_{T-K-1}\right)}{\left(K+1\right)}\right]}
\end{eqnarray*}
This gives, $S_{T-K-1}=W_{T-K-1}/\left(K+2\right)$ and the corresponding
value function is, 
\begin{eqnarray*}
V_{T-K-1}\left(P_{T-K-2},W_{T-K-1}\right) & = & \frac{W_{T-K-1}}{\left(K+2\right)}\sigma_{\varepsilon}\left\{ \frac{\theta}{\sigma_{\varepsilon}}\frac{W_{T-K-1}}{\left(K+2\right)}+\frac{\phi\left(\frac{\theta}{\sigma_{\varepsilon}}\frac{W_{T-K-1}}{\left(K+2\right)}\right)}{\Phi\left(\frac{\theta}{\sigma_{\varepsilon}}\frac{W_{T-K-1}}{\left(K+2\right)}\right)}\right\} \\
 &  & +\left(W_{T-K-1}-\frac{W_{T-K-1}}{\left(K+2\right)}\right)\sigma_{\varepsilon}\left\{ \frac{\theta}{\left(K+1\right)\sigma_{\varepsilon}}\left(W_{T-K-1}-\frac{W_{T-K-1}}{\left(K+2\right)}\right)\right.\\
 &  & +\left.\frac{\phi\left[\frac{\theta}{\left(K+1\right)\sigma_{\varepsilon}}\left(W_{T-K-1}-\frac{W_{T-K-1}}{\left(K+2\right)}\right)\right]}{\Phi\left[\frac{\theta}{\left(K+1\right)\sigma_{\varepsilon}}\left(W_{T-K-1}-\frac{W_{T-K-1}}{\left(K+2\right)}\right)\right]}\right\} 
\end{eqnarray*}
\[
V_{T-K-1}\left(P_{T-K-2},W_{T-K-1}\right)=\sigma_{\varepsilon}W_{T-K-1}\left[\frac{\theta}{\sigma_{\varepsilon}}\frac{W_{T-K-1}}{\left(K+2\right)}+\frac{\phi\left(\frac{\theta}{\sigma_{\varepsilon}}\frac{W_{T-K-1}}{\left(K+2\right)}\right)}{\Phi\left(\frac{\theta}{\sigma_{\varepsilon}}\frac{W_{T-K-1}}{\left(K+2\right)}\right)}\right]
\]
This completes the induction. 
\end{doublespace}
\end{proof}
\begin{doublespace}

\subsection{\label{subsec:Proof-of-Proposition-benchmark-complex}Proof of Proposition
\ref{The-number-of-benchmark-complex}}
\end{doublespace}
\begin{proof}
\begin{doublespace}
Consider, 

\[
V_{T}\left(P_{T-1},W_{T}\right)=E_{T}\left[\max\left\{ \left(\theta W_{T}+\varepsilon_{T}\right),0\right\} W_{T}\right]
\]

\[
V_{T}\left(P_{T-1},W_{T}\right)=E_{T}\left[\max\left\{ \left(\theta W_{T}^{2}+W_{T}\varepsilon_{T}\right),0\right\} \right]
\]
\[
V_{T}\left(P_{T-1},W_{T}\right)=E_{T}\left[\left.\left(\theta W_{T}^{2}+W_{T}\varepsilon_{T}\right)\right|\left(\theta W_{T}^{2}+W_{T}\varepsilon_{T}\right)>0\right]
\]
\[
\left\{ \because\;E\left[max\left(X,c\right)\right]=E\left[X\left|X>c\right.\right]Pr\left[X>c\right]+E\left[c\left|X\text{\ensuremath{\le}}c\right.\right]Pr\left[X\text{\ensuremath{\le}}c\right]\;\right\} 
\]
This is of the form, $E\left[\left.Y\right|Y>0\right]$ where, $Y=\left(\theta W_{T}^{2}+W_{T}\varepsilon_{T}\right)$.
We then need to calculate, 
\[
E_{T}\left[\left.\left(\theta W_{T}^{2}+W_{T}\sigma_{\varepsilon}Z\right)\right|Z>\left(-\frac{\theta W_{T}}{\sigma_{\varepsilon}}\right)\right]\quad,where\;Z\sim N\left(0,1\right)
\]
\[
\left[\because\;Y\sim N\left(\theta W_{T}^{2},W_{T}^{2}\sigma_{\varepsilon}^{2}\right)\equiv Y\sim N\left(\mu,\sigma^{2}\right)\Rightarrow Y=\mu+\sigma Z\;;\;Y>0\Rightarrow Z>-\mu/\sigma\right]
\]
We have for every standard normal distribution, $Z$, and for every
$u,$ $Pr\left[Z>\text{\textminus}u\right]=Pr\left[Z<u\right]=\mathbf{\Phi}\left(u\right)$.
Here, $\phi$ and $\mathbf{\Phi}$ are the standard normal PDF and
CDF, respectively. 
\begin{eqnarray*}
E\left[\left.Z\right|Z>-u\right] & = & \frac{1}{\mathbf{\Phi}\left(u\right)}\left[\int_{-u}^{\infty}t\phi\left(t\right)dt\right]\\
 & = & \frac{1}{\mathbf{\Phi}\left(u\right)}\left[\left.-\phi\left(t\right)\right|_{-u}^{\infty}\right]=\frac{\phi\left(u\right)}{\mathbf{\Phi}\left(u\right)}
\end{eqnarray*}
Hence we have, 
\begin{eqnarray*}
E\left[\left.Y\right|Y>0\right] & = & \mu+\sigma E\left[\left.Z\right|Z>\left(-\frac{\mu}{\sigma}\right)\right]\\
 & = & \mu+\frac{\sigma\phi\left(\mu/\sigma\right)}{\mathbf{\Phi}\left(\mu/\sigma\right)}
\end{eqnarray*}
Setting, $\psi\left(u\right)=u+\phi\left(u\right)/\Phi\left(u\right)$,
\[
E\left[\left.Y\right|Y>0\right]=\sigma\psi\left(\mu/\sigma\right)
\]
\begin{eqnarray*}
V_{T}\left(P_{T-1},W_{T}\right) & = & E_{T}\left[\left.\left(\theta W_{T}^{2}+W_{T}\varepsilon_{T}\right)\right|\left(\theta W_{T}^{2}+W_{T}\varepsilon_{T}\right)>0\right]\\
 & = & W_{T}\sigma_{\varepsilon}\left[\frac{\theta W_{T}}{\sigma_{\varepsilon}}+\frac{\phi\left(\frac{\theta W_{T}}{\sigma_{\varepsilon}}\right)}{\Phi\left(\frac{\theta W_{T}}{\sigma_{\varepsilon}}\right)}\right]=\sigma_{\varepsilon}W_{T}\psi\left(\xi W_{T}\right)\;,\;\xi=\frac{\theta}{\sigma_{\varepsilon}}
\end{eqnarray*}
 In the next to last period, $T-1$, the Bellman equation is,
\[
V_{T-1}\left(P_{T-2},W_{T-1}\right)=\underset{\left\{ S_{T-1}\right\} }{\min}\:E_{T-1}\left[\max\left\{ \left(P_{T-1}-P_{T-2}\right),0\right\} W_{T-1}+V_{T}\left(P_{T-1},W_{T}\right)\right]
\]
\[
=\underset{\left\{ S_{T-1}\right\} }{\min}\:E_{T-1}\left[\max\left\{ \left(\theta S_{T-1}+\varepsilon_{T-1}\right),0\right\} W_{T-1}+V_{T}\left(P_{T-2}+\theta S_{T-1}+\varepsilon_{T-1},W_{T-1}-S_{T-1}\right)\right]
\]
\[
=\underset{\left\{ S_{T-1}\right\} }{\min}\left\{ W_{T-1}\sigma_{\varepsilon}\left[\frac{\theta S_{T-1}}{\sigma_{\varepsilon}}+\frac{\phi\left(\frac{\theta S_{T-1}}{\sigma_{\varepsilon}}\right)}{\Phi\left(\frac{\theta S_{T-1}}{\sigma_{\varepsilon}}\right)}\right]+\left(W_{T-1}-S_{T-1}\right)\sigma_{\varepsilon}\left[\frac{\theta\left(W_{T-1}-S_{T-1}\right)}{\sigma_{\varepsilon}}+\frac{\phi\left(\frac{\theta\left(W_{T-1}-S_{T-1}\right)}{\sigma_{\varepsilon}}\right)}{\Phi\left(\frac{\theta\left(W_{T-1}-S_{T-1}\right)}{\sigma_{\varepsilon}}\right)}\right]\right\} 
\]

\[
V_{T-1}\left(P_{T-2},W_{T-1}\right)=\underset{\left\{ S_{T-1}\right\} }{\min}\left[W_{T-1}\sigma_{\varepsilon}\psi\left(\xi S_{T-1}\right)+\left(W_{T-1}-S_{T-1}\right)\sigma_{\varepsilon}\psi\left\{ \xi\left(W_{T-1}-S_{T-1}\right)\right\} \right]
\]
\[
\text{Here, }\psi\left(u\right)=u+\phi\left(u\right)/\Phi\left(u\right)\;;\;\xi=\frac{\theta}{\sigma_{\varepsilon}}\;;\;\text{Note that, }W_{T-1}=S_{T-1}+W_{T}
\]
We can show the above expression to be a convex function with a unique
minimum, since it is the sum of the portion shown to be convex earlier
(Proposition \ref{The-value-function-convexity}, Appendix \ref{subsec:Proof-of-Proposition-Convexity}),
another convex function and a linear component.

First Order Conditions (FOC) give,
\[
\frac{\partial}{\partial S_{T-1}}\left[W_{T-1}\sigma_{\varepsilon}\psi\left(\xi S_{T-1}\right)+\left(W_{T-1}-S_{T-1}\right)\sigma_{\varepsilon}\psi\left\{ \xi\left(W_{T-1}-S_{T-1}\right)\right\} \right]=0
\]
\[
\xi W_{T-1}\psi'\left(\xi S_{T-1}\right)-\xi\left(W_{T-1}-S_{T-1}\right)\psi'\left(\xi\left\{ W_{T-1}-S_{T-1}\right\} \right)-\psi\left(\xi\left\{ W_{T-1}-S_{T-1}\right\} \right)=0
\]
\begin{eqnarray*}
\left(\xi W_{T-1}\right)\left\{ 1-\frac{\xi S_{T-1}\phi\left(\xi S_{T-1}\right)}{\Phi\left(\xi S_{T-1}\right)}-\left[\frac{\phi\left(\xi S_{T-1}\right)}{\Phi\left(\xi S_{T-1}\right)}\right]^{2}\right\} \\
-\xi\left(W_{T-1}-S_{T-1}\right)\left\{ 1-\frac{\xi\left(W_{T-1}-S_{T-1}\right)\phi\left(\xi\left\{ W_{T-1}-S_{T-1}\right\} \right)}{\Phi\left(\xi\left\{ W_{T-1}-S_{T-1}\right\} \right)}-\left[\frac{\phi\left(\xi\left\{ W_{T-1}-S_{T-1}\right\} \right)}{\Phi\left(\xi\left\{ W_{T-1}-S_{T-1}\right\} \right)}\right]^{2}\right\} \\
-\left\{ \xi\left\{ W_{T-1}-S_{T-1}\right\} +\frac{\phi\left(\xi\left\{ W_{T-1}-S_{T-1}\right\} \right)}{\Phi\left(\xi\left\{ W_{T-1}-S_{T-1}\right\} \right)}\right\}  & = & 0
\end{eqnarray*}
\[
\left[\because\;\psi'\left(u\right)=1-\frac{u\phi\left(u\right)}{\Phi\left(u\right)}-\left[\frac{\phi\left(u\right)}{\Phi\left(u\right)}\right]^{2}\text{ and }\psi\left(u\right)=u+\frac{\phi\left(u\right)}{\Phi\left(u\right)}\right]
\]
\begin{eqnarray*}
W_{T-1}+\frac{\xi\left(W_{T-1}-S_{T-1}\right)^{2}\phi\left(\xi\left\{ W_{T-1}-S_{T-1}\right\} \right)}{\Phi\left(\xi\left\{ W_{T-1}-S_{T-1}\right\} \right)}+\left(W_{T-1}-S_{T-1}\right)\left[\frac{\phi\left(\xi\left\{ W_{T-1}-S_{T-1}\right\} \right)}{\Phi\left(\xi\left\{ W_{T-1}-S_{T-1}\right\} \right)}\right]^{2} &  & =\\
\left(W_{T-1}-S_{T-1}\right)+\left\{ W_{T-1}-S_{T-1}\right\} +\frac{1}{\xi}\frac{\phi\left(\xi\left\{ W_{T-1}-S_{T-1}\right\} \right)}{\Phi\left(\xi\left\{ W_{T-1}-S_{T-1}\right\} \right)}+\frac{\xi W_{T-1}S_{T-1}\phi\left(\xi S_{T-1}\right)}{\Phi\left(\xi S_{T-1}\right)}+W_{T-1}\left[\frac{\phi\left(\xi S_{T-1}\right)}{\Phi\left(\xi S_{T-1}\right)}\right]^{2}
\end{eqnarray*}
\end{doublespace}
\end{proof}
\begin{doublespace}

\subsection{\label{subsec:Proof-of-Proposition-additional-simple}Proof of Proposition
\ref{The-number-of-additional-uncertainty-simple}}
\end{doublespace}
\begin{proof}
\begin{doublespace}
Consider,
\[
V_{T}\left(P_{T-1},X_{T-1},W_{T}\right)=E_{T}\left[\max\left\{ \left(\theta W_{T}+\varepsilon_{T}+\gamma X_{T}\right),0\right\} W_{T}\right]
\]

\[
V_{T}\left(P_{T-1},X_{T-1},W_{T}\right)=E_{T}\left[\max\left\{ \left(\theta W_{T}^{2}+W_{T}\varepsilon_{T}+\gamma\rho W_{T}X_{T-1}+\gamma W_{T}\eta_{T}\right),0\right\} \right]
\]
\[
=E_{T}\left[\left.\left(\theta W_{T}^{2}+W_{T}\varepsilon_{T}+\gamma\rho W_{T}X_{T-1}+\gamma W_{T}\eta_{T}\right)\right|\left(\theta W_{T}^{2}+W_{T}\varepsilon_{T}+\gamma\rho W_{T}X_{T-1}+\gamma W_{T}\eta_{T}\right)>0\right]
\]
\[
\left\{ \because\;E\left[max\left(X,c\right)\right]=E\left[X\left|X>c\right.\right]Pr\left[X>c\right]+E\left[c\left|X\text{\ensuremath{\le}}c\right.\right]Pr\left[X\text{\ensuremath{\le}}c\right]\;\right\} 
\]
This is of the form, $E\left[\left.Y\right|Y>0\right]$ where, $Y=\left(\theta W_{T}^{2}+W_{T}\varepsilon_{T}+\gamma\rho W_{T}X_{T-1}+\gamma W_{T}\eta_{T}\right)$.
We then need to calculate, 
\[
E_{T}\left[\left.\left\{ \theta W_{T}^{2}+\gamma\rho W_{T}X_{T-1}+W_{T}\left(\sqrt{\gamma^{2}\sigma_{\eta}^{2}+\sigma_{\varepsilon}^{2}}\;\right)Z\right\} \right|Z>\left(-\frac{\theta W_{T}+\gamma\rho X_{T-1}}{\sqrt{\gamma^{2}\sigma_{\eta}^{2}+\sigma_{\varepsilon}^{2}}}\right)\right]\quad,where\;Z\sim N\left(0,1\right)
\]
\[
\left[\because\;X\sim N(\mu_{X},\sigma_{X}^{2})\;;\;Y\sim N(\mu_{Y},\sigma_{Y}^{2})\;;\;U=X+Y\Rightarrow U\sim N(\mu_{X}+\mu_{Y},\sigma_{X}^{2}+\sigma_{Y}^{2})\right]
\]
\[
\left[\because\;Y\sim N\left(\theta W_{T}^{2}+\gamma\rho W_{T}X_{T-1},W_{T}^{2}\left\{ \gamma^{2}\sigma_{\eta}^{2}+\sigma_{\varepsilon}^{2}\right\} \right)\equiv Y\sim N\left(\mu,\sigma^{2}\right)\Rightarrow Y=\mu+\sigma Z\;;\;Y>0\Rightarrow Z>-\mu/\sigma\right]
\]
We have for every standard normal distribution, $Z$, and for every
$u,$ $Pr\left[Z>\text{\textminus}u\right]=Pr\left[Z<u\right]=\mathbf{\Phi}\left(u\right)$.
Here, $\phi$ and $\mathbf{\Phi}$ are the standard normal PDF and
CDF, respectively. 
\begin{eqnarray*}
E\left[\left.Z\right|Z>-u\right] & = & \frac{1}{\mathbf{\Phi}\left(u\right)}\left[\int_{-u}^{\infty}t\phi\left(t\right)dt\right]\\
 & = & \frac{1}{\mathbf{\Phi}\left(u\right)}\left[\left.-\phi\left(t\right)\right|_{-u}^{\infty}\right]=\frac{\phi\left(u\right)}{\mathbf{\Phi}\left(u\right)}
\end{eqnarray*}
Hence we have, 
\begin{eqnarray*}
E\left[\left.Y\right|Y>0\right] & = & \mu+\sigma E\left[\left.Z\right|Z>\left(-\frac{\mu}{\sigma}\right)\right]\\
 & = & \mu+\frac{\sigma\phi\left(\mu/\sigma\right)}{\mathbf{\Phi}\left(\mu/\sigma\right)}
\end{eqnarray*}
Setting, $\psi\left(u\right)=u+\phi\left(u\right)/\Phi\left(u\right)$,
\[
E\left[\left.Y\right|Y>0\right]=\sigma\psi\left(\mu/\sigma\right)
\]
\[
V_{T}\left(P_{T-1},X_{T-1},W_{T}\right)=E_{T}\left[\left.\left(\theta W_{T}^{2}+W_{T}\varepsilon_{T}+\gamma\rho W_{T}X_{T-1}+\gamma W_{T}\eta_{T}\right)\right|\left(\theta W_{T}^{2}+W_{T}\varepsilon_{T}+\gamma\rho W_{T}X_{T-1}+\gamma W_{T}\eta_{T}\right)>0\right]
\]
\[
=W_{T}\left(\sqrt{\gamma^{2}\sigma_{\eta}^{2}+\sigma_{\varepsilon}^{2}}\;\right)\left[\frac{\theta W_{T}+\gamma\rho X_{T-1}}{\sqrt{\gamma^{2}\sigma_{\eta}^{2}+\sigma_{\varepsilon}^{2}}}+\frac{\phi\left(\frac{\theta W_{T}+\gamma\rho X_{T-1}}{\sqrt{\gamma^{2}\sigma_{\eta}^{2}+\sigma_{\varepsilon}^{2}}}\right)}{\Phi\left(\frac{\theta W_{T}+\gamma\rho X_{T-1}}{\sqrt{\gamma^{2}\sigma_{\eta}^{2}+\sigma_{\varepsilon}^{2}}}\right)}\right]
\]
\[
=\left(\sqrt{\gamma^{2}\sigma_{\eta}^{2}+\sigma_{\varepsilon}^{2}}\;\right)W_{T}\psi\left(\xi W_{T}\right)\;,\;\xi W_{T}=\frac{\theta W_{T}+\gamma\rho X_{T-1}}{\sqrt{\gamma^{2}\sigma_{\eta}^{2}+\sigma_{\varepsilon}^{2}}}
\]
 In the next to last period, $T-1$, the Bellman equation is,
\[
V_{T-1}\left(P_{T-2},X_{T-2},W_{T-1}\right)=\underset{\left\{ S_{T-1}\right\} }{\min}\:E_{T-1}\left[\max\left\{ \left(P_{T-1}-P_{T-2}\right),0\right\} S_{T-1}+V_{T}\left(P_{T-1},X_{T-1},W_{T}\right)\right]
\]
\begin{eqnarray*}
 & = & \underset{\left\{ S_{T-1}\right\} }{\min}\:E_{T-1}\left[\max\left\{ \left(\theta S_{T-1}^{2}+S_{T-1}\varepsilon_{T-1}+\gamma\rho S_{T-1}X_{T-2}+\gamma S_{T-1}\eta_{T-1}\right),0\right\} \right.\\
 &  & \left.\quad\qquad+\quad V_{T}\left(P_{T-2}+\theta S_{T-1}+\varepsilon_{T-1}+\gamma\rho X_{T-2}+\gamma\eta_{T-1},\rho X_{T-2}+\eta_{T-1},W_{T-1}-S_{T-1}\right)\right]
\end{eqnarray*}
\begin{eqnarray*}
 & = & \underset{\left\{ S_{T-1}\right\} }{\min}\left\{ S_{T-1}\left(\sqrt{\gamma^{2}\sigma_{\eta}^{2}+\sigma_{\varepsilon}^{2}}\;\right)\left[\frac{\theta S_{T-1}+\gamma\rho X_{T-2}}{\sqrt{\gamma^{2}\sigma_{\eta}^{2}+\sigma_{\varepsilon}^{2}}}+\frac{\phi\left(\frac{\theta S_{T-1}+\gamma\rho X_{T-2}}{\sqrt{\gamma^{2}\sigma_{\eta}^{2}+\sigma_{\varepsilon}^{2}}}\right)}{\Phi\left(\frac{\theta S_{T-1}+\gamma\rho X_{T-2}}{\sqrt{\gamma^{2}\sigma_{\eta}^{2}+\sigma_{\varepsilon}^{2}}}\right)}\right]\right.\\
 &  & \left.+\left\{ W_{T-1}-S_{T-1}\right\} \left(\sqrt{\gamma^{2}\sigma_{\eta}^{2}+\sigma_{\varepsilon}^{2}}\;\right)\left[\frac{\theta\left(W_{T-1}-S_{T-1}\right)+\gamma\rho X_{T-2}}{\sqrt{\gamma^{2}\sigma_{\eta}^{2}+\sigma_{\varepsilon}^{2}}}+\frac{\phi\left(\frac{\theta\left(W_{T-1}-S_{T-1}\right)+\gamma\rho X_{T-2}}{\sqrt{\gamma^{2}\sigma_{\eta}^{2}+\sigma_{\varepsilon}^{2}}}\right)}{\Phi\left(\frac{\theta\left(W_{T-1}-S_{T-1}\right)+\gamma\rho X_{T-2}}{\sqrt{\gamma^{2}\sigma_{\eta}^{2}+\sigma_{\varepsilon}^{2}}}\right)}\right]\right\} 
\end{eqnarray*}
\[
=\underset{\left\{ S_{T-1}\right\} }{\min}\left[S_{T-1}\left(\sqrt{\gamma^{2}\sigma_{\eta}^{2}+\sigma_{\varepsilon}^{2}}\;\right)\psi\left(\xi_{1}S_{T-1}\right)+\left(W_{T-1}-S_{T-1}\right)\left(\sqrt{\gamma^{2}\sigma_{\eta}^{2}+\sigma_{\varepsilon}^{2}}\;\right)\psi\left\{ \xi_{1}\left(W_{T-1}-S_{T-1}\right)\right\} \right]
\]
\[
\text{Here, }\;\xi_{1}S_{T-1}=\frac{\theta S_{T-1}+\gamma\rho X_{T-2}}{\sqrt{\gamma^{2}\sigma_{\eta}^{2}+\sigma_{\varepsilon}^{2}}}\;\;\text{Also, let }\alpha_{1}=\gamma\rho X_{T-2},\;\beta=\sqrt{\gamma^{2}\sigma_{\eta}^{2}+\sigma_{\varepsilon}^{2}}
\]
\begin{eqnarray*}
 & = & \underset{\left\{ S_{T-1}\right\} }{\min}\left\{ \left[\theta S_{T-1}^{2}+\alpha_{1}S_{T-1}+\beta S_{T-1}\frac{\phi\left(\frac{\theta S_{T-1}+\alpha}{\beta}\right)}{\Phi\left(\frac{\theta S_{T-1}+\alpha}{\beta}\right)}\right]\right.\\
 &  & \left.+\left[\theta\left(W_{T-1}-S_{T-1}\right)^{2}+\alpha_{1}\left(W_{T-1}-S_{T-1}\right)+\beta\left(W_{T-1}-S_{T-1}\right)\frac{\phi\left(\frac{\theta\left(W_{T-1}-S_{T-1}\right)+\alpha_{1}}{\beta}\right)}{\Phi\left(\frac{\theta\left(W_{T-1}-S_{T-1}\right)+\alpha_{1}}{\beta}\right)}\right]\right\} 
\end{eqnarray*}
This is a convex function and numerical solutions can be obtained
at each stage of the recursion or taking First Order Conditions give,
\begin{eqnarray*}
2\theta S_{T-1}+\alpha_{1}+\beta\left\{ \frac{\phi\left(\frac{\theta S_{T-1}+\alpha_{1}}{\beta}\right)}{\Phi\left(\frac{\theta S_{T-1}+\alpha_{1}}{\beta}\right)}+\frac{\theta S_{T-1}}{\beta}\left[-\frac{\left(\frac{\theta S_{T-1}+\alpha_{1}}{\beta}\right)\phi\left(\frac{\theta S_{T-1}+\alpha_{1}}{\beta}\right)}{\Phi\left(\frac{\theta S_{T-1}+\alpha_{1}}{\beta}\right)}-\left\{ \frac{\phi\left(\frac{\theta S_{T-1}+\alpha_{1}}{\beta}\right)}{\Phi\left(\frac{\theta S_{T-1}+\alpha_{1}}{\beta}\right)}\right\} ^{2}\right]\right\} \\
-2\theta\left(W_{T-1}-S_{T-1}\right)-\alpha_{1}+\beta\left\{ -\frac{\phi\left(\frac{\theta\left(W_{T-1}-S_{T-1}\right)+\alpha_{1}}{\beta}\right)}{\Phi\left(\frac{\theta\left(W_{T-1}-S_{T-1}\right)+\alpha_{1}}{\beta}\right)}\right.\\
\left.+\frac{\theta\left(W_{T-1}-S_{T-1}\right)}{\beta}\left[\frac{\left(\frac{\theta\left(W_{T-1}-S_{T-1}\right)+\alpha_{1}}{\beta}\right)\phi\left(\frac{\theta\left(W_{T-1}-S_{T-1}\right)+\alpha_{1}}{\beta}\right)}{\Phi\left(\frac{\theta\left(W_{T-1}-S_{T-1}\right)+\alpha_{1}}{\beta}\right)}+\left\{ \frac{\phi\left(\frac{\theta\left(W_{T-1}-S_{T-1}\right)+\alpha_{1}}{\beta}\right)}{\Phi\left(\frac{\theta\left(W_{T-1}-S_{T-1}\right)+\alpha_{1}}{\beta}\right)}\right\} ^{2}\right]\right\}  & = & 0
\end{eqnarray*}
$S_{T-1}=W_{T-1}/2$ solves this, giving the value function,
\begin{eqnarray*}
V_{T-1}\left(P_{T-2},X_{T-2},W_{T-1}\right) & = & \frac{\theta}{2}W_{T-1}^{2}+\alpha_{1}W_{T-1}+\beta W_{T-1}\frac{\phi\left(\frac{\theta W_{T-1}+2\alpha_{1}}{2\beta}\right)}{\Phi\left(\frac{\theta W_{T-1}+2\alpha_{1}}{2\beta}\right)}
\end{eqnarray*}
Continuing the recursion,
\[
V_{T-2}\left(P_{T-3},X_{T-3},W_{T-2}\right)=\underset{\left\{ S_{T-2}\right\} }{\min}\:E_{T-2}\left[\max\left\{ \left(P_{T-2}-P_{T-3}\right),0\right\} S_{T-2}+V_{T-1}\left(P_{T-2},X_{T-2},W_{T-1}\right)\right]
\]
\begin{eqnarray*}
 & = & \underset{\left\{ S_{T-2}\right\} }{\min}\:E_{T-2}\left[\max\left\{ \left(\theta S_{T-2}^{2}+S_{T-2}\varepsilon_{T-2}+\gamma\rho S_{T-2}X_{T-3}+\gamma S_{T-2}\eta_{T-2}\right),0\right\} \right.\\
 &  & \left.\quad\qquad+\quad V_{T-1}\left(P_{T-3}+\theta S_{T-2}+\varepsilon_{T-2}+\gamma\rho X_{T-3}+\gamma\eta_{T-2},\rho X_{T-3}+\eta_{T-2},W_{T-2}-S_{T-2}\right)\right]
\end{eqnarray*}
\begin{eqnarray*}
 & = & \underset{\left\{ S_{T-2}\right\} }{\min}\left\{ S_{T-2}\left(\sqrt{\gamma^{2}\sigma_{\eta}^{2}+\sigma_{\varepsilon}^{2}}\;\right)\left[\frac{\theta S_{T-2}+\gamma\rho X_{T-3}}{\sqrt{\gamma^{2}\sigma_{\eta}^{2}+\sigma_{\varepsilon}^{2}}}+\frac{\phi\left(\frac{\theta S_{T-2}+\gamma\rho X_{T-3}}{\sqrt{\gamma^{2}\sigma_{\eta}^{2}+\sigma_{\varepsilon}^{2}}}\right)}{\Phi\left(\frac{\theta S_{T-2}+\gamma\rho X_{T-3}}{\sqrt{\gamma^{2}\sigma_{\eta}^{2}+\sigma_{\varepsilon}^{2}}}\right)}\right]\right.\\
 &  & \left.+\left\{ W_{T-2}-S_{T-2}\right\} \left(\sqrt{\gamma^{2}\sigma_{\eta}^{2}+\sigma_{\varepsilon}^{2}}\;\right)\left[\frac{\theta\left\{ W_{T-2}-S_{T-2}\right\} +2\gamma\rho X_{T-3}}{2\sqrt{\gamma^{2}\sigma_{\eta}^{2}+\sigma_{\varepsilon}^{2}}}+\frac{\phi\left(\frac{\theta\left\{ W_{T-2}-S_{T-2}\right\} +2\gamma\rho X_{T-3}}{2\sqrt{\gamma^{2}\sigma_{\eta}^{2}+\sigma_{\varepsilon}^{2}}}\right)}{\Phi\left(\frac{\theta\left\{ W_{T-2}-S_{T-2}\right\} +2\gamma\rho X_{T-3}}{2\sqrt{\gamma^{2}\sigma_{\eta}^{2}+\sigma_{\varepsilon}^{2}}}\right)}\right]\right\} 
\end{eqnarray*}
\[
=\underset{\left\{ S_{T-2}\right\} }{\min}\left[S_{T-2}\left(\sqrt{\gamma^{2}\sigma_{\eta}^{2}+\sigma_{\varepsilon}^{2}}\;\right)\psi\left(\xi_{2}S_{T-2}\right)+\left(W_{T-2}-S_{T-2}\right)\left(\sqrt{\gamma^{2}\sigma_{\eta}^{2}+\sigma_{\varepsilon}^{2}}\;\right)\psi\left\{ \xi_{2}\left(W_{T-2}-S_{T-2}\right)\right\} \right]
\]
\[
\text{Here, }\;\xi_{2}S_{T-2}=\frac{\theta S_{T-2}+2\gamma\rho X_{T-3}}{2\sqrt{\gamma^{2}\sigma_{\eta}^{2}+\sigma_{\varepsilon}^{2}}}\;\;\text{Also, let }\alpha_{2}=\gamma\rho X_{T-3},\;\beta=\sqrt{\gamma^{2}\sigma_{\eta}^{2}+\sigma_{\varepsilon}^{2}}
\]
\begin{eqnarray*}
 & = & \underset{\left\{ S_{T-2}\right\} }{\min}\left\{ \left[\theta S_{T-2}^{2}+\alpha_{2}S_{T-2}+\beta S_{T-2}\frac{\phi\left(\frac{\theta S_{T-2}+\alpha_{2}}{\beta}\right)}{\Phi\left(\frac{\theta S_{T-2}+\alpha_{2}}{\beta}\right)}\right]\right.\\
 &  & \left.+\left[\frac{\theta}{2}\left(W_{T-2}-S_{T-2}\right)^{2}+\alpha_{2}\left(W_{T-2}-S_{T-2}\right)+\beta\left(W_{T-2}-S_{T-2}\right)\frac{\phi\left(\frac{\theta\left(W_{T-2}-S_{T-2}\right)+2\alpha_{2}}{2\beta}\right)}{\Phi\left(\frac{\theta\left(W_{T-2}-S_{T-2}\right)+2\alpha_{2}}{2\beta}\right)}\right]\right\} 
\end{eqnarray*}
This is a convex function and numerical solutions can be obtained
at each stage of the recursion or taking First Order Conditions give,
\begin{eqnarray*}
2\theta S_{T-2}+\alpha_{2}+\beta\left\{ \frac{\phi\left(\frac{\theta S_{T-2}+\alpha_{2}}{\beta}\right)}{\Phi\left(\frac{\theta S_{T-2}+\alpha_{2}}{\beta}\right)}+\frac{\theta S_{T-2}}{\beta}\left[-\frac{\left(\frac{\theta S_{T-2}+\alpha_{2}}{\beta}\right)\phi\left(\frac{\theta S_{T-2}+\alpha_{2}}{\beta}\right)}{\Phi\left(\frac{\theta S_{T-2}+\alpha_{2}}{\beta}\right)}-\left\{ \frac{\phi\left(\frac{\theta S_{T-2}+\alpha_{2}}{\beta}\right)}{\Phi\left(\frac{\theta S_{T-2}+\alpha_{2}}{\beta}\right)}\right\} ^{2}\right]\right\} \\
-\theta\left(W_{T-2}-S_{T-2}\right)-\alpha_{2}+\beta\left\{ -\frac{\phi\left(\frac{\theta\left(W_{T-2}-S_{T-2}\right)+2\alpha_{2}}{2\beta}\right)}{\Phi\left(\frac{\theta\left(W_{T-2}-S_{T-2}\right)+2\alpha}{2\beta}\right)}\right.\\
\left.+\frac{\theta\left(W_{T-2}-S_{T-2}\right)}{2\beta}\left[\frac{\left(\frac{\theta\left(W_{T-2}-S_{T-2}\right)+2\alpha_{2}}{2\beta}\right)\phi\left(\frac{\theta\left(W_{T-2}-S_{T-2}\right)+2\alpha_{2}}{2\beta}\right)}{\Phi\left(\frac{\theta\left(W_{T-2}-S_{T-2}\right)+2\alpha_{2}}{2\beta}\right)}+\left\{ \frac{\phi\left(\frac{\theta\left(W_{T-2}-S_{T-2}\right)+2\alpha_{2}}{2\beta}\right)}{\Phi\left(\frac{\theta\left(W_{T-2}-S_{T-2}\right)+2\alpha_{2}}{2\beta}\right)}\right\} ^{2}\right]\right\}  & = & 0
\end{eqnarray*}
$S_{T-1}=W_{T-1}/3$ solves this, giving the value function,
\begin{eqnarray*}
V_{T-2}\left(P_{T-3},X_{T-3},W_{T-2}\right) & = & \frac{\theta}{3}W_{T-2}^{2}+\alpha_{2}W_{T-2}+\beta W_{T-2}\frac{\phi\left(\frac{\theta W_{T-2}+3\alpha_{2}}{3\beta}\right)}{\Phi\left(\frac{\theta W_{T-2}+3\alpha_{2}}{3\beta}\right)}
\end{eqnarray*}

We show the general case using induction. Let the value function hold
for $T-K$
\begin{eqnarray*}
V_{T-K}\left(P_{T-K-1},X_{T-K-1},W_{T-K}\right) & = & \frac{\theta}{\left(K+1\right)}W_{T-K}^{2}+\alpha_{K}W_{T-K}+\beta W_{T-K}\frac{\phi\left(\frac{\theta W_{T-K}+\left(K+1\right)\alpha_{K}}{\left(K+1\right)\beta}\right)}{\Phi\left(\frac{\theta W_{T-K}+\left(K+1\right)\alpha_{K}}{\left(K+1\right)\beta}\right)}
\end{eqnarray*}
\begin{eqnarray*}
V_{T-K-1}\left(P_{T-K-2},X_{T-K-2},W_{T-K-1}\right) & = & \underset{\left\{ S_{T-K-1}\right\} }{\min}\:E_{T-K-1}\left[\max\left\{ \left(P_{T-K-1}-P_{T-K-2}\right),0\right\} S_{T-K-1}\right.\\
 &  & \left.+V_{T-K}\left(P_{T-K-1},X_{T-K-1},W_{T-K}\right)\right]
\end{eqnarray*}
\begin{eqnarray*}
 & = & \underset{\left\{ S_{T-K-1}\right\} }{\min}\:E_{T-K-1}\left[\max\left\{ \left(\theta S_{T-K-1}^{2}+S_{T-K-1}\varepsilon_{T-K-1}+\gamma\rho S_{T-K-1}X_{T-K-2}+\gamma S_{T-K-1}\eta_{T-K-1}\right),0\right\} \right.
\end{eqnarray*}
\begin{eqnarray*}
\left.+V_{T-K}\left(P_{T-K-2}+\theta S_{T-K-1}+\varepsilon_{T-K-1}+\gamma\rho X_{T-K-2}+\gamma\eta_{T-K-1},\rho X_{T-K-2}+\eta_{T-K-1},W_{T-K-1}-S_{T-K-1}\right)\right]
\end{eqnarray*}
\begin{eqnarray*}
 & = & \underset{\left\{ S_{T-K-1}\right\} }{\min}\left\{ S_{T-K-1}\left(\sqrt{\gamma^{2}\sigma_{\eta}^{2}+\sigma_{\varepsilon}^{2}}\;\right)\left[\frac{\theta S_{T-K-1}+\gamma\rho X_{T-K-2}}{\sqrt{\gamma^{2}\sigma_{\eta}^{2}+\sigma_{\varepsilon}^{2}}}+\frac{\phi\left(\frac{\theta S_{T-K-1}+\gamma\rho X_{T-K-2}}{\sqrt{\gamma^{2}\sigma_{\eta}^{2}+\sigma_{\varepsilon}^{2}}}\right)}{\Phi\left(\frac{\theta S_{T-K-1}+\gamma\rho X_{T-K-2}}{\sqrt{\gamma^{2}\sigma_{\eta}^{2}+\sigma_{\varepsilon}^{2}}}\right)}\right]\right.\\
 &  & +\left\{ W_{T-K-1}-S_{T-K-1}\right\} \left(\sqrt{\gamma^{2}\sigma_{\eta}^{2}+\sigma_{\varepsilon}^{2}}\;\right)\\
 &  & \left.\left[\frac{\theta\left\{ W_{T-K-1}-S_{T-K-1}\right\} +\left(K+1\right)\gamma\rho X_{T-3}}{\left(K+1\right)\sqrt{\gamma^{2}\sigma_{\eta}^{2}+\sigma_{\varepsilon}^{2}}}+\frac{\phi\left(\frac{\theta\left\{ W_{T-K-1}-S_{T-K-1}\right\} +\left(K+1\right)\gamma\rho X_{T-K-2}}{\left(K+1\right)\sqrt{\gamma^{2}\sigma_{\eta}^{2}+\sigma_{\varepsilon}^{2}}}\right)}{\Phi\left(\frac{\theta\left\{ W_{T-K-1}-S_{T-K-1}\right\} +\left(K+1\right)\gamma\rho X_{T-K-2}}{\left(K+1\right)\sqrt{\gamma^{2}\sigma_{\eta}^{2}+\sigma_{\varepsilon}^{2}}}\right)}\right]\right\} 
\end{eqnarray*}
\begin{eqnarray*}
 & = & \underset{\left\{ S_{T-K-1}\right\} }{\min}\left[S_{T-K-1}\left(\sqrt{\gamma^{2}\sigma_{\eta}^{2}+\sigma_{\varepsilon}^{2}}\;\right)\psi\left(\xi_{K+1}S_{T-K-1}\right)\right.\\
 &  & \left.+\left(W_{T-K-1}-S_{T-K-1}\right)\left(\sqrt{\gamma^{2}\sigma_{\eta}^{2}+\sigma_{\varepsilon}^{2}}\;\right)\psi\left\{ \xi_{K+1}\left(W_{T-K-1}-S_{T-K-1}\right)\right\} \right]
\end{eqnarray*}
\[
\text{Here, }\;\xi_{K+1}S_{T-K-1}=\frac{\theta S_{T-K-1}+2\gamma\rho X_{T-K-2}}{2\sqrt{\gamma^{2}\sigma_{\eta}^{2}+\sigma_{\varepsilon}^{2}}}\;\;\text{Also, let }\alpha_{K+1}=\gamma\rho X_{T-K-2},\;\beta=\sqrt{\gamma^{2}\sigma_{\eta}^{2}+\sigma_{\varepsilon}^{2}}
\]
\begin{eqnarray*}
 & = & \underset{\left\{ S_{T-K-1}\right\} }{\min}\left\{ \left[\theta S_{T-K-1}^{2}+\alpha_{K+1}S_{T-K-1}+\beta S_{T-K-1}\frac{\phi\left(\frac{\theta S_{T-K-1}+\alpha_{K+1}}{\beta}\right)}{\Phi\left(\frac{\theta S_{T-K-1}+\alpha_{K+1}}{\beta}\right)}\right]\right.\\
 &  & +\left[\frac{\theta}{\left(K+1\right)}\left(W_{T-K-1}-S_{T-K-1}\right)^{2}+\alpha_{K+1}\left(W_{T-K-1}-S_{T-K-1}\right)\right]\\
 &  & \left.+\left[\beta\left(W_{T-K-1}-S_{T-K-1}\right)\frac{\phi\left(\frac{\theta\left(W_{T-K-1}-S_{T-K-1}\right)+\left(K+1\right)\alpha_{K+1}}{\left(K+1\right)\beta}\right)}{\Phi\left(\frac{\theta\left(W_{T-K-1}-S_{T-K-1}\right)+\left(K+1\right)\alpha_{K+1}}{\left(K+1\right)\beta}\right)}\right]\right\} 
\end{eqnarray*}
This is a convex function and numerical solutions can be obtained
at each stage of the recursion or taking First Order Conditions give,
\begin{eqnarray*}
2\theta S_{T-K-1}+\alpha_{K+1}+\beta\left\{ \frac{\phi\left(\frac{\theta S_{T-K-1}+\alpha_{K+1}}{\beta}\right)}{\Phi\left(\frac{\theta S_{T-K-1}+\alpha_{K+1}}{\beta}\right)}\right.\\
\left.+\frac{\theta S_{T-K-1}}{\beta}\left[-\frac{\left(\frac{\theta S_{T-K-1}+\alpha_{K+1}}{\beta}\right)\phi\left(\frac{\theta S_{T-K-1}+\alpha_{K+1}}{\beta}\right)}{\Phi\left(\frac{\theta S_{T-K-1}+\alpha_{K+1}}{\beta}\right)}-\left\{ \frac{\phi\left(\frac{\theta S_{T-K-1}+\alpha_{K+1}}{\beta}\right)}{\Phi\left(\frac{\theta S_{T-K-1}+\alpha_{K+1}}{\beta}\right)}\right\} ^{2}\right]\right\} \\
-\frac{2\theta}{\left(K+1\right)}\left(W_{T-K-1}-S_{T-K-1}\right)-\alpha_{K+1}+\beta\left\{ -\frac{\phi\left(\frac{\theta\left(W_{T-K-1}-S_{T-K-1}\right)+\left(K+1\right)\alpha_{K+1}}{\left(K+1\right)\beta}\right)}{\Phi\left(\frac{\theta\left(W_{T-K-1}-S_{T-K-1}\right)+\left(K+1\right)\alpha_{K+1}}{\left(K+1\right)\beta}\right)}\right.\\
+\frac{\theta\left(W_{T-K-1}-S_{T-K-1}\right)}{\left(K+1\right)\beta}\left[\frac{\left(\frac{\theta\left(W_{T-K-1}-S_{T-K-1}\right)+\left(K+1\right)\alpha_{K+1}}{\left(K+1\right)\beta}\right)\phi\left(\frac{\theta\left(W_{T-K-1}-S_{T-K-1}\right)+\left(K+1\right)\alpha_{K+1}}{\left(K+1\right)\beta}\right)}{\Phi\left(\frac{\theta\left(W_{T-K-1}-S_{T-K-1}\right)+\left(K+1\right)\alpha_{K+1}}{\left(K+1\right)\beta}\right)}\right.\\
\left.\left.+\left\{ \frac{\phi\left(\frac{\theta\left(W_{T-K-1}-S_{T-K-1}\right)+\left(K+1\right)\alpha_{K+1}}{\left(K+1\right)\beta}\right)}{\Phi\left(\frac{\theta\left(W_{T-K-1}-S_{T-K-1}\right)+\left(K+1\right)\alpha_{K+1}}{\left(K+1\right)\beta}\right)}\right\} ^{2}\right]\right\}  & = & 0
\end{eqnarray*}
$S_{T-K-1}=W_{T-K-1}/\left(K+2\right)$ solves this, giving the value
function and completing the induction.
\begin{eqnarray*}
V_{T-K-1}\left(P_{T-K-2},X_{T-K-2},W_{T-K-1}\right) & = & \frac{\theta}{\left(K+2\right)}W_{T-K-1}^{2}+\alpha_{K+1}W_{T-K-1}+\beta W_{T-K-1}\frac{\phi\left(\frac{\theta W_{T-K-1}+\left(K+2\right)\alpha_{K+1}}{\left(K+2\right)\beta}\right)}{\Phi\left(\frac{\theta W_{T-K-1}+\left(K+2\right)\alpha_{K+1}}{\left(K+2\right)\beta}\right)}
\end{eqnarray*}
\end{doublespace}
\end{proof}
\begin{doublespace}

\subsection{\label{subsec:Proof-of-Proposition-additional-complex}Proof of Proposition
\ref{The-number-of-additional-uncertainty-complex}}
\end{doublespace}
\begin{proof}
\begin{doublespace}
Consider,
\[
V_{T}\left(P_{T-1},X_{T-1},W_{T}\right)=E_{T}\left[\max\left\{ \left(\theta W_{T}+\varepsilon_{T}+\gamma X_{T}\right),0\right\} W_{T}\right]
\]

\[
V_{T}\left(P_{T-1},X_{T-1},W_{T}\right)=E_{T}\left[\max\left\{ \left(\theta W_{T}^{2}+W_{T}\varepsilon_{T}+\gamma\rho W_{T}X_{T-1}+\gamma W_{T}\eta_{T}\right),0\right\} \right]
\]
\[
=E_{T}\left[\left.\left(\theta W_{T}^{2}+W_{T}\varepsilon_{T}+\gamma\rho W_{T}X_{T-1}+\gamma W_{T}\eta_{T}\right)\right|\left(\theta W_{T}^{2}+W_{T}\varepsilon_{T}+\gamma\rho W_{T}X_{T-1}+\gamma W_{T}\eta_{T}\right)>0\right]
\]
\[
\left\{ \because\;E\left[max\left(X,c\right)\right]=E\left[X\left|X>c\right.\right]Pr\left[X>c\right]+E\left[c\left|X\text{\ensuremath{\le}}c\right.\right]Pr\left[X\text{\ensuremath{\le}}c\right]\;\right\} 
\]
This is of the form, $E\left[\left.Y\right|Y>0\right]$ where, $Y=\left(\theta W_{T}^{2}+W_{T}\varepsilon_{T}+\gamma\rho W_{T}X_{T-1}+\gamma W_{T}\eta_{T}\right)$.
We then need to calculate, 
\[
E_{T}\left[\left.\left\{ \theta W_{T}^{2}+\gamma\rho W_{T}X_{T-1}+W_{T}\left(\sqrt{\gamma^{2}\sigma_{\eta}^{2}+\sigma_{\varepsilon}^{2}}\;\right)Z\right\} \right|Z>\left(-\frac{\theta W_{T}+\gamma\rho X_{T-1}}{\sqrt{\gamma^{2}\sigma_{\eta}^{2}+\sigma_{\varepsilon}^{2}}}\right)\right]\quad,where\;Z\sim N\left(0,1\right)
\]
\[
\left[\because\;X\sim N(\mu_{X},\sigma_{X}^{2})\;;\;Y\sim N(\mu_{Y},\sigma_{Y}^{2})\;;\;U=X+Y\Rightarrow U\sim N(\mu_{X}+\mu_{Y},\sigma_{X}^{2}+\sigma_{Y}^{2})\right]
\]
\[
\left[\because\;Y\sim N\left(\theta W_{T}^{2}+\gamma\rho W_{T}X_{T-1},W_{T}^{2}\left\{ \gamma^{2}\sigma_{\eta}^{2}+\sigma_{\varepsilon}^{2}\right\} \right)\equiv Y\sim N\left(\mu,\sigma^{2}\right)\Rightarrow Y=\mu+\sigma Z\;;\;Y>0\Rightarrow Z>-\mu/\sigma\right]
\]
We have for every standard normal distribution, $Z$, and for every
$u,$ $Pr\left[Z>\text{\textminus}u\right]=Pr\left[Z<u\right]=\mathbf{\Phi}\left(u\right)$.
Here, $\phi$ and $\mathbf{\Phi}$ are the standard normal PDF and
CDF, respectively. 
\begin{eqnarray*}
E\left[\left.Z\right|Z>-u\right] & = & \frac{1}{\mathbf{\Phi}\left(u\right)}\left[\int_{-u}^{\infty}t\phi\left(t\right)dt\right]\\
 & = & \frac{1}{\mathbf{\Phi}\left(u\right)}\left[\left.-\phi\left(t\right)\right|_{-u}^{\infty}\right]=\frac{\phi\left(u\right)}{\mathbf{\Phi}\left(u\right)}
\end{eqnarray*}
Hence we have, 
\begin{eqnarray*}
E\left[\left.Y\right|Y>0\right] & = & \mu+\sigma E\left[\left.Z\right|Z>\left(-\frac{\mu}{\sigma}\right)\right]\\
 & = & \mu+\frac{\sigma\phi\left(\mu/\sigma\right)}{\mathbf{\Phi}\left(\mu/\sigma\right)}
\end{eqnarray*}
Setting, $\psi\left(u\right)=u+\phi\left(u\right)/\Phi\left(u\right)$,
\[
E\left[\left.Y\right|Y>0\right]=\sigma\psi\left(\mu/\sigma\right)
\]
\[
V_{T}\left(P_{T-1},X_{T-1},W_{T}\right)=E_{T}\left[\left.\left(\theta W_{T}^{2}+W_{T}\varepsilon_{T}+\gamma\rho W_{T}X_{T-1}+\gamma W_{T}\eta_{T}\right)\right|\left(\theta W_{T}^{2}+W_{T}\varepsilon_{T}+\gamma\rho W_{T}X_{T-1}+\gamma W_{T}\eta_{T}\right)>0\right]
\]
\[
=W_{T}\left(\sqrt{\gamma^{2}\sigma_{\eta}^{2}+\sigma_{\varepsilon}^{2}}\;\right)\left[\frac{\theta W_{T}+\gamma\rho X_{T-1}}{\sqrt{\gamma^{2}\sigma_{\eta}^{2}+\sigma_{\varepsilon}^{2}}}+\frac{\phi\left(\frac{\theta W_{T}+\gamma\rho X_{T-1}}{\sqrt{\gamma^{2}\sigma_{\eta}^{2}+\sigma_{\varepsilon}^{2}}}\right)}{\Phi\left(\frac{\theta W_{T}+\gamma\rho X_{T-1}}{\sqrt{\gamma^{2}\sigma_{\eta}^{2}+\sigma_{\varepsilon}^{2}}}\right)}\right]
\]
\[
=\left(\sqrt{\gamma^{2}\sigma_{\eta}^{2}+\sigma_{\varepsilon}^{2}}\;\right)W_{T}\psi\left(\xi W_{T}\right)\;,\;\xi W_{T}=\frac{\theta W_{T}+\gamma\rho X_{T-1}}{\sqrt{\gamma^{2}\sigma_{\eta}^{2}+\sigma_{\varepsilon}^{2}}}
\]
 In the next to last period, $T-1$, the Bellman equation is,
\[
V_{T-1}\left(P_{T-2},X_{T-2},W_{T-1}\right)=\underset{\left\{ S_{T-1}\right\} }{\min}\:E_{T-1}\left[\max\left\{ \left(P_{T-1}-P_{T-2}\right),0\right\} W_{T-1}+V_{T}\left(P_{T-1},X_{T-1},W_{T}\right)\right]
\]
\begin{eqnarray*}
 & = & \underset{\left\{ S_{T-1}\right\} }{\min}\:E_{T-1}\left[\max\left\{ \left(\theta W_{T-1}S_{T-1}+W_{T-1}\varepsilon_{T-1}+\gamma\rho W_{T-1}X_{T-2}+\gamma W_{T-1}\eta_{T-1}\right),0\right\} \right.\\
 &  & \left.\quad\qquad+\quad V_{T}\left(P_{T-2}+\theta S_{T-1}+\varepsilon_{T-1}+\gamma\rho X_{T-2}+\gamma\eta_{T-1},\rho X_{T-2}+\eta_{T-1},W_{T-1}-S_{T-1}\right)\right]
\end{eqnarray*}
\begin{eqnarray*}
 & = & \underset{\left\{ S_{T-1}\right\} }{\min}\left\{ W_{T-1}\left(\sqrt{\gamma^{2}\sigma_{\eta}^{2}+\sigma_{\varepsilon}^{2}}\;\right)\left[\frac{\theta S_{T-1}+\gamma\rho X_{T-2}}{\sqrt{\gamma^{2}\sigma_{\eta}^{2}+\sigma_{\varepsilon}^{2}}}+\frac{\phi\left(\frac{\theta S_{T-1}+\gamma\rho X_{T-2}}{\sqrt{\gamma^{2}\sigma_{\eta}^{2}+\sigma_{\varepsilon}^{2}}}\right)}{\Phi\left(\frac{\theta S_{T-1}+\gamma\rho X_{T-2}}{\sqrt{\gamma^{2}\sigma_{\eta}^{2}+\sigma_{\varepsilon}^{2}}}\right)}\right]\right.\\
 &  & \left.+\left\{ W_{T-1}-S_{T-1}\right\} \left(\sqrt{\gamma^{2}\sigma_{\eta}^{2}+\sigma_{\varepsilon}^{2}}\;\right)\left[\frac{\theta\left(W_{T-1}-S_{T-1}\right)+\gamma\rho X_{T-2}}{\sqrt{\gamma^{2}\sigma_{\eta}^{2}+\sigma_{\varepsilon}^{2}}}+\frac{\phi\left(\frac{\theta\left(W_{T-1}-S_{T-1}\right)+\gamma\rho X_{T-2}}{\sqrt{\gamma^{2}\sigma_{\eta}^{2}+\sigma_{\varepsilon}^{2}}}\right)}{\Phi\left(\frac{\theta\left(W_{T-1}-S_{T-1}\right)+\gamma\rho X_{T-2}}{\sqrt{\gamma^{2}\sigma_{\eta}^{2}+\sigma_{\varepsilon}^{2}}}\right)}\right]\right\} 
\end{eqnarray*}
\[
=\underset{\left\{ S_{T-1}\right\} }{\min}\left[W_{T-1}\left(\sqrt{\gamma^{2}\sigma_{\eta}^{2}+\sigma_{\varepsilon}^{2}}\;\right)\psi\left(\xi_{1}S_{T-1}\right)+\left(W_{T-1}-S_{T-1}\right)\left(\sqrt{\gamma^{2}\sigma_{\eta}^{2}+\sigma_{\varepsilon}^{2}}\;\right)\psi\left\{ \xi_{1}\left(W_{T-1}-S_{T-1}\right)\right\} \right]
\]
\[
\text{Here, }\;\xi_{1}S_{T-1}=\frac{\theta S_{T-1}+\gamma\rho X_{T-2}}{\sqrt{\gamma^{2}\sigma_{\eta}^{2}+\sigma_{\varepsilon}^{2}}}\;\;\text{Also, let }\alpha=\gamma\rho X_{T-2},\;\beta=\sqrt{\gamma^{2}\sigma_{\eta}^{2}+\sigma_{\varepsilon}^{2}}
\]
\begin{eqnarray*}
 & = & \underset{\left\{ S_{T-1}\right\} }{\min}\left\{ \left[\theta W_{T-1}S_{T-1}+\alpha W_{T-1}+\beta W_{T-1}\frac{\phi\left(\frac{\theta S_{T-1}+\alpha}{\beta}\right)}{\Phi\left(\frac{\theta S_{T-1}+\alpha}{\beta}\right)}\right]\right.\\
 &  & \left.+\left[\theta\left(W_{T-1}-S_{T-1}\right)^{2}+\alpha\left(W_{T-1}-S_{T-1}\right)+\beta\left(W_{T-1}-S_{T-1}\right)\frac{\phi\left(\frac{\theta\left(W_{T-1}-S_{T-1}\right)+\alpha}{\beta}\right)}{\Phi\left(\frac{\theta\left(W_{T-1}-S_{T-1}\right)+\alpha}{\beta}\right)}\right]\right\} 
\end{eqnarray*}
This is a convex function and taking First Order Conditions give,
\begin{eqnarray*}
\theta W_{T-1}+\beta W_{T-1}\left\{ \frac{\theta}{\beta}\left[-\frac{\left(\frac{\theta S_{T-1}+\alpha}{\beta}\right)\phi\left(\frac{\theta S_{T-1}+\alpha}{\beta}\right)}{\Phi\left(\frac{\theta S_{T-1}+\alpha}{\beta}\right)}-\left\{ \frac{\phi\left(\frac{\theta S_{T-1}+\alpha}{\beta}\right)}{\Phi\left(\frac{\theta S_{T-1}+\alpha}{\beta}\right)}\right\} ^{2}\right]\right\} \\
-2\theta\left(W_{T-1}-S_{T-1}\right)-\alpha+\beta\left\{ -\frac{\phi\left(\frac{\theta\left(W_{T-1}-S_{T-1}\right)+\alpha}{\beta}\right)}{\Phi\left(\frac{\theta\left(W_{T-1}-S_{T-1}\right)+\alpha}{\beta}\right)}\right.\\
\left.+\frac{\theta\left(W_{T-1}-S_{T-1}\right)}{\beta}\left[\frac{\left(\frac{\theta\left(W_{T-1}-S_{T-1}\right)+\alpha}{\beta}\right)\phi\left(\frac{\theta\left(W_{T-1}-S_{T-1}\right)+\alpha}{\beta}\right)}{\Phi\left(\frac{\theta\left(W_{T-1}-S_{T-1}\right)+\alpha}{\beta}\right)}+\left\{ \frac{\phi\left(\frac{\theta\left(W_{T-1}-S_{T-1}\right)+\alpha}{\beta}\right)}{\Phi\left(\frac{\theta\left(W_{T-1}-S_{T-1}\right)+\alpha}{\beta}\right)}\right\} ^{2}\right]\right\}  & = & 0
\end{eqnarray*}
\end{doublespace}
\end{proof}
\begin{doublespace}

\subsection{\label{subsec:Proof-of-Proposition-linear-percentage}Proof of Proposition
\ref{The-value-function-linear-percentage}}
\end{doublespace}
\begin{proof}
\begin{doublespace}
Consider,

\[
V_{T}\left(P_{T-1},X_{T-1},W_{T}\right)=E_{T}\left[\max\left\{ \left(\widetilde{P}_{T}\left(1+\theta W_{T}+\gamma X_{T}\right)-P_{T-1}\right),0\right\} W_{T}\right]
\]

\[
V_{T}\left(P_{T-1},X_{T-1},W_{T}\right)=E_{T}\left[\max\left\{ \left(\widetilde{P}_{T-1}e^{B_{T}}\left[1+\theta W_{T}+\gamma\left(\rho X_{T-1}+\eta_{T}\right)\right]-P_{T-1}\right),0\right\} W_{T}\right]
\]
\[
=E_{T}\left[\max\left\{ \widetilde{P}_{T-1}W_{T}e^{B_{T}}+\theta W_{T}^{2}\widetilde{P}_{T-1}e^{B_{T}}+\gamma\rho X_{T-1}\widetilde{P}_{T-1}W_{T}e^{B_{T}}+\gamma\widetilde{P}_{T-1}W_{T}e^{B_{T}}\eta_{T}-W_{T}P_{T-1},0\right\} \right]
\]
\begin{eqnarray*}
 & = & E_{T}\left[\left.\left(\widetilde{P}_{T-1}W_{T}e^{B_{T}}+\theta W_{T}^{2}\widetilde{P}_{T-1}e^{B_{T}}+\gamma\rho X_{T-1}\widetilde{P}_{T-1}W_{T}e^{B_{T}}+\gamma\widetilde{P}_{T-1}W_{T}e^{B_{T}}\eta_{T}-W_{T}P_{T-1}\right)\right|\right.\\
 &  & \left.\left(\widetilde{P}_{T-1}W_{T}e^{B_{T}}+\theta W_{T}^{2}\widetilde{P}_{T-1}e^{B_{T}}+\gamma\rho X_{T-1}\widetilde{P}_{T-1}W_{T}e^{B_{T}}+\gamma\widetilde{P}_{T-1}W_{T}e^{B_{T}}\eta_{T}-W_{T}P_{T-1}\right)>0\right]
\end{eqnarray*}
\[
\left\{ \because\;E\left[max\left(X,c\right)\right]=E\left[X\left|X>c\right.\right]Pr\left[X>c\right]+E\left[c\left|X\text{\ensuremath{\le}}c\right.\right]Pr\left[X\text{\ensuremath{\le}}c\right]\;\right\} 
\]
This is of the form, $E\left[\left.Y_{2}\right|Y_{2}>0\right]$ where, 

$Y_{2}=\left(\widetilde{P}_{T-1}W_{T}e^{B_{T}}+\theta W_{T}^{2}\widetilde{P}_{T-1}e^{B_{T}}+\gamma\rho X_{T-1}\widetilde{P}_{T-1}W_{T}e^{B_{T}}+\gamma\widetilde{P}_{T-1}W_{T}e^{B_{T}}\eta_{T}-W_{T}P_{T-1}\right)$.
We simplify using some notational shortcuts,
\begin{eqnarray*}
 &  & E\left[\left.\left(ae^{X}+be^{X}+ce^{X}+de^{X}Y_{1}+k\right)\right|\left.\left(ae^{X}+be^{X}+ce^{X}+de^{X}Y_{1}+k\right)>0\right]\right.
\end{eqnarray*}
\[
X\sim N\left(\mu_{X},\sigma_{X}^{2}\right);Y_{1}\sim N\left(0,\sigma_{Y_{1}}^{2}\right);X\;\text{and }Y_{1}\text{ are independent. Also, }a,b,c,d>0,k<0
\]
\begin{eqnarray*}
\equiv &  & E\left[\left.\left(e^{X}\left\{ a+b+c+dY_{1}\right\} +k\right)\right|\left.\left(e^{X}\left\{ a+b+c+dY_{1}\right\} +k\right)>0\right]\right.
\end{eqnarray*}
\begin{eqnarray*}
\equiv &  & E\left[\left.\left(e^{X}Y+k\right)\right|\left.\left(e^{X}Y+k\right)>0\right]\right.
\end{eqnarray*}
\[
X\sim N\left(\mu_{X},\sigma_{X}^{2}\right);Y\sim N\left(\mu_{Y},\sigma_{Y}^{2}\right);X\;\text{and }Y\text{ are independent. Also, }k<0
\]
Consider,
\begin{eqnarray*}
E\left[\left.\left(e^{X}Y+k\right)\right|\left(e^{X}Y+k\right)>0\right] & = & E\left[k\left|\left(e^{X}Y+k\right)>0\right.\right]+E\left[\left(e^{X}Y\right)\left|\left(e^{X}Y+k\right)>0\right.\right]
\end{eqnarray*}
\begin{eqnarray*}
 & = & k+E\left[\left.\left(Ye^{X}\right)\right|\left.\left(Ye^{X}+k\right)>0\right]\right.
\end{eqnarray*}
\begin{eqnarray*}
 & = & k+\int\int ye^{x}f\left(\left.ye^{x}\right|\left\{ ye^{x}+k\right\} >0\right)dxdy
\end{eqnarray*}
Here, $f\left(w\right)$ is the probability density function for $w$,
\begin{eqnarray*}
 & = & k+\int\int ye^{x}\frac{f\left(ye^{x};\left\{ ye^{x}+k\right\} >0\right)}{f\left(\left\{ ye^{x}+k\right\} >0\right)}dxdy
\end{eqnarray*}
\[
\left[\text{We note that, }ye^{x}>-k>0\Rightarrow y>0\right]
\]
\begin{eqnarray*}
 & = & k+\int\int ye^{x}\frac{f\left(y\right)f\left(e^{x};\left\{ ye^{x}+k\right\} >0\right)}{f\left(\left\{ ye^{x}+k\right\} >0\right)}dxdy
\end{eqnarray*}
\begin{eqnarray*}
 & = & k+\int y\left[\int\frac{e^{x}f\left(e^{x};\left\{ e^{x}>-\frac{k}{y}\right\} \right)}{f\left(e^{x}>-\frac{k}{y}\right)}dx\right]f\left(y\right)dy
\end{eqnarray*}
\begin{eqnarray*}
 & = & k+\int y\left[\int e^{x}f\left(e^{x}\left|\left\{ e^{x}>-\frac{k}{y}\right\} \right.\right)dx\right]f\left(y\right)dy
\end{eqnarray*}
\begin{eqnarray*}
 & = & k+\int_{0}^{\left(y<-k\right)}y\left[\int e^{x}f\left(e^{x}\left|\left\{ e^{x}>1\right\} \right.\right)dx\right]f\left(y\right)dy+\int_{\left(y>-k\right)}^{\infty}y\left[\int e^{x}f\left(e^{x}\left|\left\{ e^{x}<1\right\} \right.\right)dx\right]f\left(y\right)dy
\end{eqnarray*}
\begin{eqnarray*}
 & = & k+\int_{0}^{\left(-k\right)}y\left[E\left(\left.W\right|W>c\right)\right]f\left(y\right)dy+\int_{\left(-k\right)}^{\infty}y\left[E\left(\left.W\right|W<c\right)\right]f\left(y\right)dy\quad;\;\text{here, }W=e^{X}\text{ and }c=1
\end{eqnarray*}
Simplifying the inner expectations,
\[
E\left(\left.W\right|W>c\right)=\frac{1}{P\left(e^{X}>c\right)}\int_{c}^{\infty}w\frac{1}{w\sigma_{X}\sqrt{2\pi}}e^{-\frac{1}{2}\left[\frac{ln\left(w\right)-\mu_{X}}{\sigma_{X}}\right]^{2}}dw
\]
Put $t=ln\left(w\right)$, we have, $dw=e^{t}dt$
\[
E\left(\left.W\right|W>c\right)=\frac{1}{P\left(X>ln\left(c\right)\right)}\int_{ln\left(c\right)}^{\infty}\frac{e^{t}}{\sigma_{X}\sqrt{2\pi}}e^{-\frac{1}{2}\left(\frac{t-\mu_{X}}{\sigma_{X}}\right)^{2}}dt
\]
\[
t-\frac{1}{2}\left(\frac{t-\mu_{X}}{\sigma_{X}}\right)^{2}=-\frac{1}{2\sigma_{X}^{2}}\left(t-\left(\mu_{X}+\sigma_{X}^{2}\right)\right)^{2}+\mu_{X}+\frac{\sigma_{X}^{2}}{2}
\]
\[
E\left(\left.W\right|W>c\right)=\frac{e^{\left(\mu_{X}+\frac{1}{2}\sigma_{X}^{2}\right)}}{P\left(\mu_{X}+\sigma_{X}Z>ln\left(c\right)\right)}\int_{ln\left(c\right)}^{\infty}\frac{1}{\sigma_{X}\sqrt{2\pi}}e^{-\frac{1}{2}\left[\frac{t-\left(\mu_{X}+\sigma_{X}^{2}\right)}{\sigma_{X}}\right]^{2}}dt\quad;Z\sim N\left(0,1\right)
\]
Put $s=\left[\frac{t-\left(\mu_{X}+\sigma_{X}^{2}\right)}{\sigma_{X}}\right]$
and $b=\left[\frac{ln\left(c\right)-\left(\mu_{X}+\sigma_{X}^{2}\right)}{\sigma_{X}}\right]$
we have, $ds=\frac{dt}{\sigma_{X}}$
\[
E\left(\left.W\right|W>c\right)=\frac{e^{\left(\mu_{X}+\frac{1}{2}\sigma_{X}^{2}\right)}}{P\left(Z>\frac{ln\left(c\right)-\mu_{X}}{\sigma_{X}}\right)}\int_{b}^{\infty}\frac{1}{\sqrt{2\pi}}e^{-\frac{1}{2}s^{2}}ds
\]
\[
=\frac{e^{\left(\mu_{X}+\frac{1}{2}\sigma_{X}^{2}\right)}}{P\left(Z<\frac{-ln\left(c\right)+\mu_{X}}{\sigma_{X}}\right)}\left[\int_{-\infty}^{\infty}\frac{1}{\sqrt{2\pi}}e^{-\frac{1}{2}s^{2}}ds-\int_{-\infty}^{b}\frac{1}{\sqrt{2\pi}}e^{-\frac{1}{2}s^{2}}ds\right]
\]
\[
=\frac{e^{\left(\mu_{X}+\frac{1}{2}\sigma_{X}^{2}\right)}}{P\left(Z<\frac{-ln\left(c\right)+\mu_{X}}{\sigma_{X}}\right)}\left[1-\Phi\left(b\right)\right]\quad;\Phi\text{ is the standard normal CDF}
\]
\[
=\frac{e^{\left(\mu_{X}+\frac{1}{2}\sigma_{X}^{2}\right)}}{\Phi\left(\frac{-ln\left(c\right)+\mu_{X}}{\sigma_{X}}\right)}\left[\Phi\left(-b\right)\right]
\]
Similarly for the other case,
\[
E\left(\left.W\right|W<c\right)=\frac{1}{P\left(e^{X}<c\right)}\int_{0}^{c}w\frac{1}{w\sigma_{X}\sqrt{2\pi}}e^{-\frac{1}{2}\left[\frac{ln\left(w\right)-\mu_{X}}{\sigma_{X}}\right]^{2}}dw
\]
Put $t=ln\left(w\right)$, we have, $dw=e^{t}dt$
\[
E\left(\left.W\right|W<c\right)=\frac{1}{P\left(X<ln\left(c\right)\right)}\int_{-\infty}^{ln\left(c\right)}\frac{e^{t}}{\sigma_{X}\sqrt{2\pi}}e^{-\frac{1}{2}\left(\frac{t-\mu_{X}}{\sigma_{X}}\right)^{2}}dt
\]
\[
t-\frac{1}{2}\left(\frac{t-\mu_{X}}{\sigma_{X}}\right)^{2}=-\frac{1}{2\sigma_{X}^{2}}\left(t-\left(\mu_{X}+\sigma_{X}^{2}\right)\right)^{2}+\mu_{X}+\frac{\sigma_{X}^{2}}{2}
\]
\[
E\left(\left.W\right|W<c\right)=\frac{e^{\left(\mu_{X}+\frac{1}{2}\sigma_{X}^{2}\right)}}{P\left(\mu_{X}+\sigma_{X}Z<ln\left(c\right)\right)}\int_{-\infty}^{ln\left(c\right)}\frac{1}{\sigma_{X}\sqrt{2\pi}}e^{-\frac{1}{2}\left[\frac{t-\left(\mu_{X}+\sigma_{X}^{2}\right)}{\sigma_{X}}\right]^{2}}dt\quad;Z\sim N\left(0,1\right)
\]
Put $s=\left[\frac{t-\left(\mu_{X}+\sigma_{X}^{2}\right)}{\sigma_{X}}\right]$
and $b=\left[\frac{ln\left(c\right)-\left(\mu_{X}+\sigma_{X}^{2}\right)}{\sigma_{X}}\right]$
we have, $ds=\frac{dt}{\sigma_{X}}$
\[
E\left(\left.W\right|W<c\right)=\frac{e^{\left(\mu_{X}+\frac{1}{2}\sigma_{X}^{2}\right)}}{P\left(Z<\frac{ln\left(c\right)-\mu_{X}}{\sigma_{X}}\right)}\int_{-\infty}^{b}\frac{1}{\sqrt{2\pi}}e^{-\frac{1}{2}s^{2}}ds
\]
\[
=\frac{e^{\left(\mu_{X}+\frac{1}{2}\sigma_{X}^{2}\right)}}{P\left(Z<\frac{ln\left(c\right)-\mu_{X}}{\sigma_{X}}\right)}\left[\Phi\left(b\right)\right]\quad;\Phi\text{ is the standard normal CDF}
\]
\[
=\frac{e^{\left(\mu_{X}+\frac{1}{2}\sigma_{X}^{2}\right)}}{\Phi\left(\frac{ln\left(c\right)-\mu_{X}}{\sigma_{X}}\right)}\left[\Phi\left(b\right)\right]
\]
Using the results for the inner expectations,
\[
E\left[\left.\left(e^{X}Y+k\right)\right|\left(e^{X}Y+k\right)>0\right]=k+\int_{0}^{\left(-k\right)}y\left[\frac{e^{\left(\mu_{X}+\frac{1}{2}\sigma_{X}^{2}\right)}}{\Phi\left(\frac{-ln\left(c\right)+\mu_{X}}{\sigma_{X}}\right)}\left[\Phi\left(-b\right)\right]\right]f\left(y\right)dy+\int_{\left(-k\right)}^{\infty}y\left[\frac{e^{\left(\mu_{X}+\frac{1}{2}\sigma_{X}^{2}\right)}}{\Phi\left(\frac{ln\left(c\right)-\mu_{X}}{\sigma_{X}}\right)}\left[\Phi\left(b\right)\right]\right]f\left(y\right)dy
\]
\[
=k+e^{\left(\mu_{X}+\frac{1}{2}\sigma_{X}^{2}\right)}\left[\int_{0}^{\left(-k\right)}y\left\{ \frac{\Phi\left(\frac{\mu_{X}+\sigma_{X}^{2}}{\sigma_{X}}\right)}{\Phi\left(\frac{\mu_{X}}{\sigma_{X}}\right)}\right\} f\left(y\right)dy+\int_{\left(-k\right)}^{\infty}y\left\{ \frac{\Phi\left(-\left[\frac{\mu_{X}+\sigma_{X}^{2}}{\sigma_{X}}\right]\right)}{\Phi\left(-\left[\frac{\mu_{X}}{\sigma_{X}}\right]\right)}\right\} f\left(y\right)dy\right]
\]
\[
=k+e^{\left(\mu_{X}+\frac{1}{2}\sigma_{X}^{2}\right)}\left[\int_{0}^{\left(-k\right)}y\left\{ \frac{\Phi\left(\frac{\mu_{X}+\sigma_{X}^{2}}{\sigma_{X}}\right)}{\Phi\left(\frac{\mu_{X}}{\sigma_{X}}\right)}\right\} f\left(y\right)dy+\int_{\left(-k\right)}^{\infty}y\left\{ \frac{1-\Phi\left(\frac{\mu_{X}+\sigma_{X}^{2}}{\sigma_{X}}\right)}{1-\Phi\left(\frac{\mu_{X}}{\sigma_{X}}\right)}\right\} f\left(y\right)dy\right]
\]
\begin{eqnarray*}
 & = & k+e^{\left(\mu_{X}+\frac{1}{2}\sigma_{X}^{2}\right)}\left[\left\{ \frac{\Phi\left(\frac{\mu_{X}+\sigma_{X}^{2}}{\sigma_{X}}\right)}{\Phi\left(\frac{\mu_{X}}{\sigma_{X}}\right)}\right\} \int_{-\frac{\mu_{Y}}{\sigma_{Y}}}^{-\left(\frac{k+\mu_{Y}}{\sigma_{Y}}\right)}\left(\mu_{Y}+\sigma_{Y}z\right)\frac{1}{\sqrt{2\pi}}e^{-\frac{1}{2}z^{2}}dz\right.\\
 &  & +\left.\left\{ \frac{1-\Phi\left(\frac{\mu_{X}+\sigma_{X}^{2}}{\sigma_{X}}\right)}{1-\Phi\left(\frac{\mu_{X}}{\sigma_{X}}\right)}\right\} \int_{-\left(\frac{k+\mu_{Y}}{\sigma_{Y}}\right)}^{\infty}\left(\mu_{Y}+\sigma_{Y}z\right)\frac{1}{\sqrt{2\pi}}e^{-\frac{1}{2}z^{2}}dz\right]\quad;Z\sim N\left(0,1\right)
\end{eqnarray*}
\begin{eqnarray*}
 & = & k+e^{\left(\mu_{X}+\frac{1}{2}\sigma_{X}^{2}\right)}\left[\left\{ \frac{\Phi\left(\frac{\mu_{X}+\sigma_{X}^{2}}{\sigma_{X}}\right)}{\Phi\left(\frac{\mu_{X}}{\sigma_{X}}\right)}\right\} \left\{ \mu_{Y}\left[\Phi\left(-\left[\frac{k+\mu_{Y}}{\sigma_{Y}}\right]\right)-\Phi\left(-\frac{\mu_{Y}}{\sigma_{Y}}\right)\right]-\frac{\sigma_{Y}}{\sqrt{2\pi}}\left[e^{-\frac{1}{2}\left(\frac{k+\mu_{Y}}{\sigma_{Y}}\right)^{2}}-e^{-\frac{1}{2}\left(\frac{\mu_{Y}}{\sigma_{Y}}\right)^{2}}\right]\right\} \right.\\
 &  & +\left.\left\{ \frac{1-\Phi\left(\frac{\mu_{X}+\sigma_{X}^{2}}{\sigma_{X}}\right)}{1-\Phi\left(\frac{\mu_{X}}{\sigma_{X}}\right)}\right\} \left\{ \mu_{Y}\left[1-\Phi\left(-\left[\frac{k+\mu_{Y}}{\sigma_{Y}}\right]\right)\right]+\frac{\sigma_{Y}}{\sqrt{2\pi}}\left[e^{-\frac{1}{2}\left(\frac{k+\mu_{Y}}{\sigma_{Y}}\right)^{2}}\right]\right\} \right]
\end{eqnarray*}
\end{doublespace}
\end{proof}
\begin{doublespace}

\subsection{\label{subsec:Proof-of-Proposition-liquidity-constraints}Proof of
Proposition \ref{The-value-function-liquidity-constraints}}
\end{doublespace}
\begin{proof}
\begin{doublespace}
Consider,
\[
V_{T}\left(P_{T-1},O_{T-1},W_{T}\right)=\underset{\left\{ S_{T}\right\} }{\min}\:E_{T}\left[\max\left\{ \left(P_{T}-P_{T-1}\right),0\right\} S_{T}\right]
\]

\[
V_{T}\left(P_{T-1},O_{T-1},W_{T}\right)=E_{T}\left[\max\left\{ \left(\alpha P_{T-1}+\theta W_{T}P_{T-1}-\gamma\left(O_{T}-W_{T}\right)P_{T-1}+\varepsilon_{T}\right),0\right\} W_{T}\right]
\]
\[
=E_{T}\left[\max\left\{ \left(\alpha P_{T-1}W_{T}+\theta W_{T}^{2}P_{T-1}+\gamma W_{T}^{2}P_{T-1}-\gamma\rho O_{T-1}W_{T}P_{T-1}-\gamma W_{T}P_{T-1}\eta_{T}+W_{T}\varepsilon_{T}\right),0\right\} \right]
\]
Setting $\beta=\theta+\gamma$,
\begin{eqnarray*}
V_{T}\left(P_{T-1},O_{T-1},W_{T}\right) & = & E_{T}\left[\left.\left(\alpha P_{T-1}W_{T}+\beta W_{T}^{2}P_{T-1}-\gamma\rho O_{T-1}W_{T}P_{T-1}-\gamma W_{T}P_{T-1}\eta_{T}+W_{T}\varepsilon_{T}\right)\right|\right.\\
 &  & \left.\left(\alpha P_{T-1}W_{T}+\beta W_{T}^{2}P_{T-1}-\gamma\rho O_{T-1}W_{T}P_{T-1}-\gamma W_{T}P_{T-1}\eta_{T}+W_{T}\varepsilon_{T}\right)>0\right]
\end{eqnarray*}
\[
\left\{ \because\;E\left[max\left(X,c\right)\right]=E\left[X\left|X>c\right.\right]Pr\left[X>c\right]+E\left[c\left|X\text{\ensuremath{\le}}c\right.\right]Pr\left[X\text{\ensuremath{\le}}c\right]\;\right\} 
\]
This is of the form, $E\left[\left.Y\right|Y>0\right]$ where, 

$Y=\left(\alpha P_{T-1}W_{T}+\beta W_{T}^{2}P_{T-1}-\gamma\rho O_{T-1}W_{T}P_{T-1}-\gamma W_{T}P_{T-1}\eta_{T}+W_{T}\varepsilon_{T}\right)$.
We then need to calculate, 
\begin{eqnarray*}
E_{T}\left[\vphantom{\left(\frac{\frac{A}{B}}{\frac{V}{b}}\right)}\left.\left\{ \alpha P_{T-1}W_{T}+\beta W_{T}^{2}P_{T-1}-\gamma\rho O_{T-1}W_{T}P_{T-1}+W_{T}\left(\sqrt{\gamma^{2}P_{T-1}^{2}\sigma_{\eta}^{2}+\sigma_{\varepsilon}^{2}}\;\right)Z\right\} \right|\right.\\
\left.Z>\left(-\frac{\alpha P_{T-1}+\beta W_{T}P_{T-1}-\gamma\rho O_{T-1}P_{T-1}}{\sqrt{\gamma^{2}P_{T-1}^{2}\sigma_{\eta}^{2}+\sigma_{\varepsilon}^{2}}}\right)\right]\quad,where\;Z\sim N\left(0,1\right)
\end{eqnarray*}
\[
\left[\because\;X\sim N(\mu_{X},\sigma_{X}^{2})\;;\;Y\sim N(\mu_{Y},\sigma_{Y}^{2})\;;\;U=X+Y\Rightarrow U\sim N(\mu_{X}+\mu_{Y},\sigma_{X}^{2}+\sigma_{Y}^{2})\right]
\]
\begin{eqnarray*}
\left[\vphantom{\frac{A}{B}}\because\;Y\sim N\left(\alpha P_{T-1}W_{T}+\beta W_{T}^{2}P_{T-1}-\gamma\rho O_{T-1}W_{T}P_{T-1}-\gamma W_{T}P_{T-1}\eta_{T}+W_{T}\varepsilon_{T}\right)\right.\\
\left.\equiv Y\sim N\left(\mu,\sigma^{2}\right)\Rightarrow Y=\mu+\sigma Z\;;\;Y>0\Rightarrow Z>-\mu/\sigma\;\vphantom{\frac{A}{B}}\right]
\end{eqnarray*}
We have for every standard normal distribution, $Z$, and for every
$u,$ $Pr\left[Z>\text{\textminus}u\right]=Pr\left[Z<u\right]=\mathbf{\Phi}\left(u\right)$.
Here, $\phi$ and $\mathbf{\Phi}$ are the standard normal PDF and
CDF, respectively.
\begin{eqnarray*}
E\left[\left.Z\right|Z>-u\right] & = & \frac{1}{\mathbf{\Phi}\left(u\right)}\left[\int_{-u}^{\infty}t\phi\left(t\right)dt\right]\\
 & = & \frac{1}{\mathbf{\Phi}\left(u\right)}\left[\left.-\phi\left(t\right)\right|_{-u}^{\infty}\right]=\frac{\phi\left(u\right)}{\mathbf{\Phi}\left(u\right)}
\end{eqnarray*}
Hence we have, 
\begin{eqnarray*}
E\left[\left.Y\right|Y>0\right] & = & \mu+\sigma E\left[\left.Z\right|Z>\left(-\frac{\mu}{\sigma}\right)\right]\\
 & = & \mu+\frac{\sigma\phi\left(\mu/\sigma\right)}{\mathbf{\Phi}\left(\mu/\sigma\right)}
\end{eqnarray*}
Setting, $\psi\left(u\right)=u+\phi\left(u\right)/\Phi\left(u\right)$,
\[
E\left[\left.Y\right|Y>0\right]=\sigma\psi\left(\mu/\sigma\right)
\]
\begin{eqnarray*}
V_{T}\left(P_{T-1},O_{T-1},W_{T}\right) & = & W_{T}\left(\sqrt{\gamma^{2}P_{T-1}^{2}\sigma_{\eta}^{2}+\sigma_{\varepsilon}^{2}}\;\right)\left[\left(\frac{\alpha P_{T-1}+\beta W_{T}P_{T-1}-\gamma\rho O_{T-1}P_{T-1}}{\sqrt{\gamma^{2}P_{T-1}^{2}\sigma_{\eta}^{2}+\sigma_{\varepsilon}^{2}}}\right)\right.\\
 &  & \left.+\frac{\phi\left(\frac{\alpha P_{T-1}+\beta W_{T}P_{T-1}-\gamma\rho O_{T-1}P_{T-1}}{\sqrt{\gamma^{2}P_{T-1}^{2}\sigma_{\eta}^{2}+\sigma_{\varepsilon}^{2}}}\right)}{\Phi\left(\frac{\alpha P_{T-1}+\beta W_{T}P_{T-1}-\gamma\rho O_{T-1}P_{T-1}}{\sqrt{\gamma^{2}P_{T-1}^{2}\sigma_{\eta}^{2}+\sigma_{\varepsilon}^{2}}}\right)}\right]
\end{eqnarray*}
\[
=\left(\sqrt{\gamma^{2}P_{T-1}^{2}\sigma_{\eta}^{2}+\sigma_{\varepsilon}^{2}}\;\right)W_{T}\psi\left(\xi W_{T}\right)\;,\;\xi W_{T}=\left(\frac{\alpha P_{T-1}+\beta W_{T}P_{T-1}-\gamma\rho O_{T-1}P_{T-1}}{\sqrt{\gamma^{2}P_{T-1}^{2}\sigma_{\eta}^{2}+\sigma_{\varepsilon}^{2}}}\right)
\]
In the next to last period, $T-1$, the Bellman equation is,
\[
V_{T-1}\left(P_{T-2},O_{T-2},W_{T-1}\right)=\underset{\left\{ S_{T-1}\right\} }{\min}\:E_{T-1}\left[\max\left\{ \left(P_{T-1}-P_{T-2}\right),0\right\} S_{T-1}+V_{T}\left(P_{T-1},O_{T-1},W_{T}\right)\right]
\]
\begin{eqnarray*}
 & = & \underset{\left\{ S_{T-1}\right\} }{\min}\:E_{T-1}\left[\max\left\{ \left(\alpha P_{T-2}S_{T-1}+\beta S_{T-1}^{2}P_{T-2}-\gamma\rho O_{T-2}S_{T-1}P_{T-2}-\gamma S_{T-1}P_{T-2}\eta_{T-1}+S_{T-1}\varepsilon_{T-1}\right),0\right\} \right.\\
 &  & \left.\quad\qquad+\quad V_{T}\left(\left(\alpha+1\right)P_{T-2}+\beta S_{T-1}P_{T-2}-\gamma\rho O_{T-2}P_{T-2}-\gamma P_{T-2}\eta_{T-1}+\varepsilon_{T-1},\rho O_{T-2}+\eta_{T-1},W_{T-1}-S_{T-1}\right)\right]
\end{eqnarray*}
\begin{eqnarray*}
= & \underset{\left\{ S_{T-1}\right\} }{\min} & \:E_{T-1}\left\{ \vphantom{\left(\frac{\frac{\alpha P_{T-2}}{\sqrt{\sigma_{\varepsilon}^{2}}}}{\frac{\alpha P_{T-2}}{\sqrt{\sigma_{\varepsilon}^{2}}}}\right)}S_{T-1}\left(\sqrt{\gamma^{2}P_{T-2}^{2}\sigma_{\eta}^{2}+\sigma_{\varepsilon}^{2}}\;\right)\right.
\end{eqnarray*}
\[
\left[\left(\frac{\alpha P_{T-2}+\beta S_{T-1}P_{T-2}-\gamma\rho O_{T-2}P_{T-2}}{\sqrt{\gamma^{2}P_{T-2}^{2}\sigma_{\eta}^{2}+\sigma_{\varepsilon}^{2}}}\right)+\frac{\phi\left(\frac{\alpha P_{T-2}+\beta S_{T-1}P_{T-2}-\gamma\rho O_{T-2}P_{T-2}}{\sqrt{\gamma^{2}P_{T-2}^{2}\sigma_{\eta}^{2}+\sigma_{\varepsilon}^{2}}}\right)}{\Phi\left(\frac{\alpha P_{T-2}+\beta S_{T-1}P_{T-2}-\gamma\rho O_{T-2}P_{T-2}}{\sqrt{\gamma^{2}P_{T-2}^{2}\sigma_{\eta}^{2}+\sigma_{\varepsilon}^{2}}}\right)}\right]
\]
\[
+\left(W_{T-1}-S_{T-1}\right)\left(\sqrt{\gamma^{2}\left\{ P_{T-2}\left(\alpha+1+\beta S_{T-1}-\gamma\rho O_{T-2}-\gamma\eta_{T-1}\right)+\varepsilon_{T-1}\right\} ^{2}\sigma_{\eta}^{2}+\sigma_{\varepsilon}^{2}}\;\right)
\]
\[
\left[\left(\frac{\left\{ P_{T-2}\left(\alpha+1+\beta S_{T-1}-\gamma\rho O_{T-2}-\gamma\eta_{T-1}\right)+\varepsilon_{T-1}\right\} \left\{ \alpha+\beta\left(W_{T-1}-S_{T-1}\right)-\gamma\rho^{2}O_{T-2}-\gamma\rho\eta_{T-1}\right\} }{\sqrt{\gamma^{2}\left\{ P_{T-2}\left(\alpha+1+\beta S_{T-1}-\gamma\rho O_{T-2}-\gamma\eta_{T-1}\right)+\varepsilon_{T-1}\right\} ^{2}\sigma_{\eta}^{2}+\sigma_{\varepsilon}^{2}}}\right)\right.
\]
\[
\left.+\left.\frac{\phi\left(\frac{\left\{ P_{T-2}\left(\alpha+1+\beta S_{T-1}-\gamma\rho O_{T-2}-\gamma\eta_{T-1}\right)+\varepsilon_{T-1}\right\} \left\{ \alpha+\beta\left(W_{T-1}-S_{T-1}\right)-\gamma\rho^{2}O_{T-2}-\gamma\rho\eta_{T-1}\right\} }{\sqrt{\gamma^{2}\left\{ P_{T-2}\left(\alpha+1+\beta S_{T-1}-\gamma\rho O_{T-2}-\gamma\eta_{T-1}\right)+\varepsilon_{T-1}\right\} ^{2}\sigma_{\eta}^{2}+\sigma_{\varepsilon}^{2}}}\right)}{\Phi\left(\frac{\left\{ P_{T-2}\left(\alpha+1+\beta S_{T-1}-\gamma\rho O_{T-2}-\gamma\eta_{T-1}\right)+\varepsilon_{T-1}\right\} \left\{ \alpha+\beta\left(W_{T-1}-S_{T-1}\right)-\gamma\rho^{2}O_{T-2}-\gamma\rho\eta_{T-1}\right\} }{\sqrt{\gamma^{2}\left\{ P_{T-2}\left(\alpha+1+\beta S_{T-1}-\gamma\rho O_{T-2}-\gamma\eta_{T-1}\right)+\varepsilon_{T-1}\right\} ^{2}\sigma_{\eta}^{2}+\sigma_{\varepsilon}^{2}}}\right)}\right]\right\} 
\]
\end{doublespace}
\end{proof}

\end{document}